\documentclass[11pt, letterpaper]{article}
\usepackage[utf8]{inputenc}
\usepackage{float,cite,bm,graphicx,amsmath,amsthm,amssymb,subcaption,caption,bbold,amstext,array,authblk}
\usepackage{epstopdf}
\usepackage[nottoc]{tocbibind}
\newcolumntype{L}{>{$}l<{$}} 
\usepackage[margin=1in]{geometry}
\usepackage{setspace}
\usepackage{xcolor}

\newcommand{\abs}[1]{\left\lvert #1 \right\rvert}
\newcommand{\mc}[1]{\mathcal{#1}}

\newcommand{\bom}{\boldsymbol{\omega}}

\newcommand{\bv}{\bm{v}}
\newcommand{\bx}{\bm{x}}
\newcommand{\bw}{\bm{w}}
\newcommand{\bu}{\bm{u}}

\newcommand{\bp}{\bm{p}}
\newcommand{\bff}{\bm{f}}
\newcommand{\by}{\bm{y}}
\newcommand{\bL}{\bm{m}}
\newcommand{\bmm}{\bm{m}}
\newcommand{\dd}{\mathrm{d}}
\newcommand{\X}{\bm{X}}

\newcommand{\be}{\bm{e}}
\newcommand{\barz}{\overline z}

\newcommand{\trans}{^{\mathrm{T}}}

\newcommand{\RR}{\mathbb{R}}
\newcommand{\Z}{\mathbb{Z}}

\newcommand{\p}{\partial}
\renewcommand{\div}{{\rm{div}\,}}
\newcommand{\SB}{{\rm SB}}

\newtheorem{theorem}{Theorem}[section]
\newtheorem{lemma}[theorem]{Lemma}

\theoremstyle{definition}
\newtheorem{remark}[theorem]{Remark}

\graphicspath{{Figures/}}

\begin{document}
\title{An integral model based on slender body theory, with applications to curved rigid fibers}
\author[1]{Helge I. Andersson}
\author[2]{Elena Celledoni}
\author[3]{Laurel Ohm}
\author[2]{Brynjulf Owren}
\author[2]{Benjamin K. Tapley}
\affil[1]{Department of Energy and Process Engineering, The Norwegian University of Science and Technology, 7491 Trondheim, Norway}
\affil[2]{ Department of Mathematical Sciences, The Norwegian University of Science and Technology, 7491 Trondheim, Norway}
\affil[3]{Courant Institute of Mathematical Sciences, New York University,
	New York, New York 10012, USA}
\date{\today}
\setcounter{Maxaffil}{0}
\renewcommand\Affilfont{\itshape\small}

\maketitle
\begin{abstract}
	We propose a novel integral model describing the motion of curved slender fibers in viscous flow, and develop a numerical method for simulating dynamics of rigid fibers. The model is derived from nonlocal slender body theory (SBT), which approximates flow near the fiber using singular solutions of the Stokes equations integrated along the fiber centerline. In contrast to other models based on (singular) SBT, our model yields a smooth integral kernel which incorporates the (possibly varying) fiber radius naturally. The integral operator is provably negative definite in a non-physical idealized geometry, as expected from PDE theory. This is numerically verified in physically relevant geometries. We propose a convergent numerical method for solving the integral equation and discuss its convergence and stability. The accuracy of the model and method is verified against known models for ellipsoids. Finally, a fast algorithm for computing dynamics of rigid fibers with complex geometries is developed.
\end{abstract}

\section{Introduction}
The dynamics of thin fibers immersed in fluid play an important role in many biological and engineering processes, including microorganism propulsion \cite{rodenborn2013propulsion,chattopadhyay2009effect,lauga2009hydrodynamics,spagnolie2011comparative}, rheological properties of fiber suspensions used to create composite materials \cite{petrie1999rheology,hamalainen2011papermaking,fan1998direct}, and deposition of microplastics in the ocean \cite{martin2017deposition}. Here the term `fiber' is used to refer to a particle with a very large aspect ratio. In many of the applications mentioned, the cross sectional radius of the fiber is small compared to the length scales of the surrounding fluid, which can be well approximated locally by Stokes flow. This allows for the development of computationally tractable mathematical models describing the interaction between the fiber and the surrounding fluid.  \\

Slender body theory \cite{gotz2000interactions, johnson1980improved, keller1976slender, lighthill1976flagellar, tornberg2004simulating} is a popular tool for reducing computational costs of simulating these thin fibers by approximating the fibers as one dimensional curves. Here we propose an integral model based on slender body theory which involves a smooth kernel and incorporates the (possibly varying) fiber radius in a natural way. Since the integral kernels are smooth, the model resembles the method of regularized Stokeslets \cite{cortez2012slender} with an arclength-dependent regularization similar to \cite{walker2020regularised}; however, we derive our model from usual (singular) Stokeslets and doublets. As such, we avoid introducing additional parameters into the basic model. The model relies on the asymptotic cancellation of angular-dependent terms along the fiber surface (see Section \ref{model} for details), leaving an expression that retains a dependence on the fiber radius in a natural way.   \\

Furthermore, we calculate the spectrum of our integral operator in the toy scenario of a straight-but-periodic fiber with constant radius. In this toy geometry, our integral operator is negative definite, as is the well-posed PDE operator of \cite{closed_loop,free_ends} which it is designed to approximate (see \cite{spectral_calc}). This is in contrast to other models based on (non-regularized) slender body theory which rely on further asymptotic expansion with respect to the fiber radius \cite{johnson1980improved, keller1976slender, lighthill1976flagellar}. These models exhibit an instability as the eigenvalues of the operator cross zero at a high but finite wavenumber. \\

The model we derive initially yields a first-kind Fredholm integral equation for the force density along the fiber centerline. Such integral equations are known for being ill-posed \cite[Chapt. 15.1]{kress1989linear}, as they do not necessarily have a bounded inverse at the continuous level. Numerical discretization alone can provide sufficient regularization to invert first-kind integral equations at the discrete level, but to make our model more suitable for inversion, we use an integral identity to regularize the expression into a second-kind equation. 
The second-kind regularization preserves the asymptotic accuracy of the model while improving the conditioning and invertibility of the corresponding numerical method. The regularization also serves to ensure that the discretized operator is negative definite, even in the presence of numerical errors, by bounding the spectrum away from zero. 
We distinguish this type of regularization from the method of regularized Stokeslets, since our regularization is not a key component of the model derivation. In particular, we can directly compare our model \emph{with} regularization to our model \emph{without}, which we will do repeatedly throughout the paper. We also distinguish this regularization from the procedure used by Tornberg and Shelley \cite{tornberg2004simulating}, since we are not correcting for a high wavenumber instability. This allows us to compare the numerical behavior of our regularized and unregularized models at the discrete level even for very fine discretization. Moreover, the regularization used here affects all directions (both normal and tangent to the slender body centerline) in the same way.  \\

The solution of the resulting second-kind Fredholm integral equation is a force density along the slender body centerline which we integrate to find the total force and torque on the rigid fiber. We present a numerical method based on the Nystr\"om method for solving second-kind Fredholm integral equations \cite[Chapt. 12.4]{atkinson2005theoretical}. Based on the results of \cite{atkinson2005theoretical}, we show that the method is convergent and that the accuracy of the method depends on the quadrature error. Numerical tests confirm the convergence results. Not surprisingly, we note significant improvements in the conditioning of the second-kind versus first-kind formulation of the model. We also numerically verify the spectral properties of the model in different geometries. \\

We develop an algorithm for dynamic simulations of a rigid fiber. The rigidity of the fiber can be exploited such that only matrix-vector products need to be performed within the time loop, resulting in a fast algorithm for computing dynamics. We compare the dynamics of our model to the well-studied dynamics of a slender prolate spheroid \cite{jeffery1922motion,brenner1964stokes,chwang1975hydromechanics}. We then apply our model to compare the dynamics of curved fibers whose centerlines deviate randomly from straight lines by varying magnitudes.  \\

The structure of the paper is as follows. Section \ref{SBmodel} presents the slender body model, which is derived in greater detail and justified via spectral comparisons with other slender body theories in Section \ref{model}. In Section \ref{sec:disc} we discuss a method for numerically solving Fredholm integral equations and integrating the result, and demonstrate the convergence of the method for our model. Section \ref{sec:dynamics} outlines a fast algorithm for computing the dynamics of a rigid slender fiber in viscous flow. We apply the dynamical algorithm to simulate the dynamics of fibers with complex shapes. Finally, we comment on conclusions and outlook for the model in Section \ref{conclude}.

\subsection{Fiber geometry}
We begin by introducing some notation for the slender geometries considered throughout the paper. Fix $\epsilon$, $L$ with $0<\epsilon\ll L$ and let $\X_{\text{ext}}:  [-\sqrt{L^2+\epsilon^2},\sqrt{L^2+\epsilon^2}] \to \RR^3$ denote the coordinates of a $C^2$ curve in $\RR^3$, parameterized by arclength $s$. Defining $\be_{\rm s}(s) = \frac{d\X_\text{ext}}{ds}/\abs{\frac{d\X_\text{ext}}{ds}}$, the unit tangent vector to $\X_\text{ext}(s)$, we parameterize points near $\X_\text{ext}(s)$ with respect to the orthonormal frame $(\be_{\rm s}(s),\be_{n_1}(s),\be_{n_2}(s))$ defined in \cite{free_ends}. Letting
\[ \be_r(s,\theta) := \cos\theta\be_{n_1}(s)+\sin\theta\be_{n_2}(s), \]
we define the slender body $\Sigma_\epsilon$ as
\begin{equation}\label{SB_def_free}
\Sigma_\epsilon:= \big\{\bx\in \RR^3 \; : \; \bx = \X_{\text{ext}}(s) + \rho\be_r(s,\theta), \; \rho <\epsilon r(s), \;  s\in [-\sqrt{L^2+\epsilon^2},\sqrt{L^2+\epsilon^2}] \big\}.
\end{equation}
Here the radius function $r\in C^2(-\sqrt{L^2+\epsilon^2},\sqrt{L^2+\epsilon^2})$ is required to satisfy $0<r(s)\le1$ for each $s\in (-\sqrt{L^2+\epsilon^2},\sqrt{L^2+\epsilon^2})$, and $r(s)$ must decay smoothly to zero at the fiber endpoints $\pm \sqrt{L^2+\epsilon^2}$. There are many admissible radius functions $r$ which can be considered. For the simulations in this paper, we will use a thin prolate spheroid as our geometrical model for a slender fiber. In this case, the radius function $r(s)$ is given by
\begin{equation}\label{prolate}
r(s) = \frac{1}{\sqrt{L^2+\epsilon^2}}\sqrt{L^2+\epsilon^2-s^2 }.
\end{equation}
We consider the subset
\begin{equation}
\X := \{ \X_\text{ext}(s) \; : \; -L\le s\le L \}
\end{equation}
extending from focus to focus of the prolate spheroid \eqref{prolate}, and define $\X(s)$ to be the effective centerline of the slender body so that $r=O(\epsilon)$ at the effective endpoints $s=\pm L$.  \\

The slender body model described in Section \ref{SBmodel} may also be used in the case of a closed curve, in which case we take $\X(L)=\X(-L)$ and consider $s\in\RR / 2L$. We may take the radius function $r\equiv1$ in this case.

\section{Slender body model}\label{SBmodel}
To describe the motion of the thin fiber $\Sigma_\epsilon$ \eqref{SB_def_free} in Stokes flow, we will use an expression derived from nonlocal slender body theory \cite{gotz2000interactions, johnson1980improved, keller1976slender, tornberg2004simulating}. Letting $\bm{f}(s,t)$ denote the force per unit length exerted by the fiber on the surrounding fluid at time $t$, we approximate  the velocity $\frac{\p\X}{\p t}$ of the fiber relative to a given background flow $\bu_0$ by
\begin{align}
\label{SB_new3}
8\pi\mu \bigg(\frac{\p\X}{\p t} - \bu_0(\X(s,t),t) \bigg) &= - 2\log\eta\, \bm{f}(s,t) - \int_{-L}^L\bigg(\bm{S}_{\epsilon,\eta} + \frac{\epsilon^2r^2(s')}{2}\bm{D}_\epsilon \bigg) \bm{f}(s',t) \, ds' ,\\
\label{Sdef}
\bm{S}_{\epsilon,\eta}(s,s',t) &= \frac{{\bf I}}{(|\overline{\X}|^2+\eta^2\epsilon^2r^2(s))^{1/2}} + \frac{\overline{\X}\overline{\X}^{\rm T} }{(|\overline{\X}|^2+\epsilon^2r^2(s))^{3/2}} \\
\label{Ddef}
\bm{D}_\epsilon(s,s',t) &= \frac{{\bf I}}{(|\overline{\X}|^2+\epsilon^2r^2(s))^{3/2} } - \frac{3\overline{\X}\overline{\X}^{\rm T} }{(|\overline{\X}|^2+\epsilon^2r^2(s))^{5/2}}
\end{align}
where $\overline{\X}(s,s',t) = \X(s,t)-\X(s',t)$. Here $\eta>1$ is a regularization parameter chosen so that \eqref{SB_new3} is a second-kind Fredholm equation for $\bm{f}$. Notice that $\eta$ must also appear in the first term of $\bm{S}_{\epsilon,\eta}$ in order to retain the asymptotic consistency of the model \eqref{SB_new3}. This is due to an integral identity \eqref{eta_ID} used to convert the integral model from a first-kind equation for for $\bm{f}$. The model accounts for a varying radius $r(s)$ through the denominators of each term as well as the coefficient of $\bm{D}_\epsilon$. Note that since $r(s)$ is nonzero for $-L\le s\le L$, the integral kernel is smooth for each $s\in [-L,L]$. We provide a more detailed derivation of \eqref{SB_new3}--\eqref{Ddef} in Section \ref{model}. \\

The model given by equations \eqref{SB_new3}--\eqref{Ddef} and the analysis in Section \ref{model} can be used to describe both flexible and rigid fibers. In Section \ref{sec:dynamics} we apply our model to the dynamics of a rigid fiber, since the invertibility properties of \eqref{SB_new3}--\eqref{Ddef} make the model especially suitable for simulating rigid filaments.  \\

In the case of a rigid fiber, at each time $t$ we additionally impose the constraint
\begin{equation}\label{rigid_vel}
\frac{\p\X}{\p t} = \bv + \bm{\omega}\times \X(s),
\end{equation}
where $\bv$, $\bm{\omega}\in \RR^3$ are the given linear and angular velocity of the fiber (see \cite{gustavsson2009gravity, rigid_SBT, tornberg2006numerical}). We then use \eqref{SB_new3} to solve for the total force $\bm{F}(t)$ and torque $\bm{T}(t)$ exerted on the slender body at time $t$ via
\begin{equation}\label{FandN}
\int_{-L}^L \bm{f}(s,t) \, ds = \bm{F}(t), \qquad \int_{-L}^L \X(s,t)\times \bm{f}(s,t) = \bm{T}(t).
\end{equation}
Note that solving for $\bm{F}$ and $\bm{T}$ involves inverting the expression \eqref{SB_new3}, so we are particularly concerned with the invertibility of the integral equation.  

\section{Derivation and justification of the slender body model}\label{model}
Our model for the motion of the fiber is based on slender body theory \cite{gotz2000interactions, johnson1980improved, keller1976slender, tornberg2004simulating}. According to slender body theory, the fluid velocity $\bu^\SB(\bx,t)$ at any point $\bx$ away from the fiber centerline $\X(s,t)$ is approximated by the integral expression
\begin{equation}\label{stokes_SB}
\begin{aligned}
8\pi\mu \big(\bu^{\SB}(\bx,t) - \bu_0(\bx,t) \big) &=-\int_{-L}^L \bigg( \mc{S}\big(\bx-\X(s',t) \big)+\frac{\epsilon^2r^2(s')}{2}\mc{D}\big(\bx-\X(s',t) \big) \bigg)\bm{f}(s',t) \, ds' \\
\mc{S}(\bx)&=\frac{{\bf I}}{\abs{\bx}}+\frac{\bx\bx^{\rm T}}{\abs{\bx}^3}, \;
\mc{D}(\bx)=\frac{{\bf I}}{\abs{\bx}^3}-\frac{3\bx\bx^{\rm T}}{\abs{\bx}^5},
\end{aligned}
\end{equation}
where $\bu_0(\bx,t)$ is the fluid velocity in the absence of the fiber and $\mu$ is the fluid viscosity. The expression $\frac{1}{8\pi\mu}\mc{S}(\bx)$ is the free space Green's function for the Stokes equations in $\RR^3$, commonly known as the Stokeslet, while $\frac{1}{8\pi\mu}\mc{D}(\bx)=\frac{1}{16\pi\mu}\Delta\mc{S}(\bx)$ is a higher order correction to the velocity approximation, often known as a doublet. The force-per-unit-length $\bm{f}(s,t)$ exerted by the fluid on the body is distributed between the generalized foci of the slender body at $s=\pm L$.  \\

In the stationary setting, the velocity field given by \eqref{stokes_SB} is an asymptotically accurate approximation to the velocity field around a three-dimensional semi-flexible rod satisfying a well-posed {\it slender body PDE}, defined in \cite{closed_loop,free_ends} as the following boundary value problem for the Stokes equations:
\begin{equation}\label{SB_PDE}
\begin{aligned}
-\mu\Delta \bu +\nabla p &= 0, \quad \div \, \bu = 0 \qquad \text{ in } \RR^3\backslash \overline{\Sigma_\epsilon} \\
\int_0^{2\pi} (\bm{\sigma}\bm{n})\big|_{(\varphi(s),\theta)} \mc{J}_\epsilon(\varphi(s),\theta) \varphi'(s) \, d\theta &= -\bm{f}(s) \hspace{2.1cm} \text{ on } \p\Sigma_\epsilon \\
\bu\big|_{\p\Sigma_\epsilon} &= \bu(s), \hspace{2.2cm} \text{ unknown but independent of }\theta \\
\abs{\bu} \to 0 & \text{ as } \abs{\bx} \to \infty.
\end{aligned}
\end{equation}
Here $\bm{\sigma}=\mu\big(\nabla\bu+(\nabla\bu)^{\rm T}\big)-p{\bf I}$ is the fluid stress tensor, $\bm{n}(\bx)$ denotes the unit normal vector pointing into $\Sigma_\epsilon$ at $\bx\in \p\Sigma_\epsilon$, $\mc{J}_\epsilon(s,\theta)$ is the Jacobian factor on $\p\Sigma_\epsilon$, and $\varphi(s):= \frac{s\sqrt{L^2+\epsilon^2}}{L}$ is a stretch function to address the discrepancy between the extent of $\bm{f}$ and the extent of the actual slender body surface. Given a force density $\bm{f}\in C^1(-L,L)$ which decays like $r(s)$ at the fiber endpoints ($\bm{f}(s)\sim r(\varphi(s))$ as $s\to \pm L$), the difference between the slender body approximation $\bu^{\SB}$ and the solution of \eqref{SB_PDE} is bounded by an expression proportional to $\epsilon\abs{\log\epsilon}$. Note that $r(s)$ need not be spheroidal \eqref{prolate} for this error analysis to hold, but $r(s)$ must decay smoothly to zero at the physical endpoints of the fiber at $s=\pm \sqrt{L^2+\epsilon^2}$. \\

A key component of the well-posedness theory for the slender body PDE to which \eqref{stokes_SB} is an approximation is the {\it fiber integrity condition} on $\bu\big|_{\p\Sigma_\epsilon}$. The fiber integrity condition requires the velocity across each cross section $s$ of the slender body to be constant; i.e. the velocity $\bu(\bx)$ at any point $\bx(s,\theta)=\X(s) +\epsilon r(s) \be_r(s,\theta)\in \p\Sigma_\epsilon$ satisfies $\p_\theta \bu(\bx(s,\theta))=0$. This is to ensure that the cross sectional shape of the fiber does not deform over time. An important aspect of the accuracy of slender body theory is that the expression \eqref{stokes_SB} satisfies this fiber integrity condition to leading order in $\epsilon$. Specifically, by Propositions 3.9 and 3.11 in \cite{closed_loop,free_ends}, respectively, the angular dependence in $\bu^\SB(\bx)$ over each cross section $s$ of the slender body is only $\mc{O}(\epsilon\log\epsilon)$. \\

Another important general feature of the slender body PDE \eqref{SB_PDE} is that the operator mapping the force data $\bm{f}(s)$ to the $\theta$-independent fiber velocity $\bu|_{\p\Sigma_\epsilon}(s)$ is negative definite (see \cite{spectral_calc}; note that the sign convention for $\bm{f}$ is opposite). \\

Now, the velocity expression \eqref{stokes_SB} is singular at $\bx=\X(s,t)$ and can be used only away from the fiber centerline; however, \eqref{stokes_SB} presents a starting point for approximating the velocity of the slender body itself. Various methods can be used to obtain an expression for the relative velocity of the fiber centerline $\frac{\p\X(s,t)}{\p t}$ which depends only on the arclength parameter $s$ and time $t$ \cite{cortez2012slender, gotz2000interactions, johnson1980improved, keller1976slender, lighthill1976flagellar, maxian2020integral, MEKiT, tornberg2004simulating}. Here we consider a different approach to deriving a limiting centerline expression from \eqref{stokes_SB} which evidently results in a negative definite integral operator mapping $\bm{f}$ to $\bu|_{\p\Sigma_\epsilon}$. We then regularize this first-kind integral equation in an asymptotically consistent way to yield the second-kind integral equation \eqref{SB_new3}. We outline our approach here and provide a more detailed justification in Section \ref{spec_sec}.  \\

The first step in approximating $\frac{\p\X(s,t)}{\p t}$ is to evaluate \eqref{stokes_SB} on the surface of the slender body at $\bx=\X(s,t) + \epsilon r(s)\be_r(s,\theta,t)$. Written out, the velocity field along the fiber surface is given by
\begin{equation}\label{SB_new}
\begin{aligned}
&8\pi\mu \bigg(\bu^\SB(\bx(s,\theta,t),t) - \bu_0(\X(s,t),t) \bigg) = \\
&\hspace{1cm} - \int_{-L}^L \bigg(\frac{{\bf I}}{\abs{\bm{R}}} + \frac{\overline{\X}\overline{\X}^{\rm T}+\epsilon r(\overline{\X} \be_r^{\rm T} +\be_r\overline{\X}^{\rm T}) + \epsilon^2 r^2 \be_r\be_r^{\rm T} }{\abs{\bm{R}}^3} \\
&\hspace{2cm} + \frac{\epsilon^2r^2(s')}{2} \bigg(\frac{{\bf I}}{\abs{\bm{R}}^3} -3 \frac{\overline{\X}\overline{\X}^{\rm T}+\epsilon r(\overline{\X} \be_r^{\rm T} +\be_r\overline{\X}^{\rm T}) + \epsilon^2 r^2 \be_r\be_r^{\rm T} }{\abs{\bm{R}}^5}\bigg) \bigg) \bm{f}(s',t) \, ds' ,
\end{aligned}
\end{equation}
where unless otherwise specified, we have $r=r(s)$,  $\overline{\X}=\overline{\X}(s,s',t) = \X(s,t) - \X(s',t)$ and $\bm{R}=\bm{R}(s,s',\theta,t)= \overline{\X}+\epsilon r(s)\be_r(s,\theta,t)$. Now, along the fiber surface, the expression \eqref{SB_new} satisfies the fiber integrity condition to leading order in $\epsilon$; i.e. the terms containing $\be_r(s,\theta,t)$ in \eqref{SB_new} vanish to $\mc{O}(\epsilon\log\epsilon)$. In particular, both the Stokeslet and doublet include a $\theta$-dependent term with $\epsilon^2 r^2 \be_r\be_r^{\rm T}$ in the numerator. Due to the form of $\bm{R}$ in the denominator, both of these terms are $\mc{O}(1)$ at $s=s'$; however, upon integrating in $s'$, these terms cancel each other asymptotically to order $\epsilon\log\epsilon$ (see estimates 3.62 and 3.65 in \cite{closed_loop} and estimates 3.40 and 3.43 in \cite{free_ends}). Furthermore, the terms $\epsilon r(\overline{\X} \be_r^{\rm T} +\be_r\overline{\X}^{\rm T})$ in both the Stokeslet and doublet approximately integrate to zero in $s'$, while the $\be_r$ term in each denominator from $\abs{\bm{R}(s,\theta,t)}^2=\abs{ \overline{\X}}^2 + 2\epsilon r\be_r\cdot\overline{\X}+\epsilon^2r^2$ is also only $O(\epsilon \log\epsilon)$ (see Propositions 3.9 and 3.11 in \cite{closed_loop}, \cite{free_ends}, respectively). \\

Due to these cancellations and the fact that dropping these terms still approximates the slender body PDE solution of \cite{closed_loop,free_ends} to at least $O(\epsilon\log\epsilon)$, we may eliminate all terms containing $\be_r(s,\theta,t)$ in \eqref{SB_new} to obtain a $\theta$-independent expression which approximates the velocity of the fiber itself:
\begin{equation}\label{SB_new2}
\begin{aligned}
&8\pi\mu \bigg(\frac{\p\X}{\p t} - \bu_0(\X(s,t),t) \bigg) = - \int_{-L}^L \bigg(\frac{{\bf I}}{(|\overline{\X}|^2+\epsilon^2r^2(s))^{1/2}} + \frac{\overline{\X}\overline{\X}^{\rm T} }{(|\overline{\X}|^2+\epsilon^2r^2(s))^{3/2}} \\
&\hspace{4cm} + \frac{\epsilon^2r^2(s')}{2} \bigg(\frac{{\bf I}}{(|\overline{\X}|^2+\epsilon^2r^2(s))^{3/2} } - \frac{3\overline{\X}\overline{\X}^{\rm T} }{(|\overline{\X}|^2+\epsilon^2r^2(s))^{5/2}}\bigg) \bigg) \bm{f}(s',t) \, ds' .
\end{aligned}
\end{equation}

The expression \eqref{SB_new2} serves as the model underlying our final slender body velocity expression \eqref{SB_new3}.
In Section \ref{spec_sec}, we show that in a simplified setting, \eqref{SB_new2} results in a negative definite operator mapping the force density $\bm{f}$ to the fiber velocity $\frac{\p\X}{\p t}$, whereas other models which rely on further asymptotic expansion of \eqref{SB_new} about $\epsilon=0$ do not, and incur high wavenumber instabilities. This phenomenon is well known for the Keller--Rubinow model \cite{keller1976slender,gotz2000interactions}, but for other possible centerline expressions, including models similar to Lighthill \cite{lighthill1976flagellar}, this high wavenumber instability has not been documented previously. It seems that our model \eqref{SB_new2} may be the simplest that can be obtained by expanding from \eqref{SB_new} while still guaranteeing a negative definite operator. \\

Now, since the integral operator in \eqref{SB_new2} has a smooth kernel, the expression \eqref{SB_new2} yields a first-kind Fredholm integral equation for $\bm{f}$ when the fiber velocity $\frac{\p\X}{\p t}$ is supplied. Describing the motion of a rigid fiber involves inverting this expression to solve for $\bm{f}$, which in general is an ill-posed problem for a first-kind equation. Thus we want to regularize the integral operator \eqref{SB_new2} to create a second-kind integral equation while keeping the same order of accuracy in the map $\bm{f}\mapsto\frac{\p\X}{\p t}$.   \\

We first note that, for $\eta>1$, we have the following identity:
\begin{equation}\label{eta_ID}
\int_{-L}^L \bigg(\frac{1}{(|\overline{\X}|^2+\epsilon^2r^2(s))^{1/2}} - \frac{1}{(|\overline{\X}|^2+\eta^2\epsilon^2r^2(s))^{1/2}}  \bigg) g(s')\, ds' = 2\log\eta \, g(s)+ \mc{O}(\eta\epsilon\log(\eta\epsilon)).
\end{equation}

\begin{proof}
By Lemma 3.8 in \cite{free_ends}, for $a> 0$ sufficiently small, we have
\begin{equation}\label{ID1}
\begin{aligned}
 \int_{-L}^L &\bigg(\frac{g(s')}{(|\overline{\X}|^2+ a^2r^2(s))^{1/2}} - \frac{g(s')}{|\overline{\X}|} + \frac{g(s)}{\abs{s-s'}} \bigg) \, ds' \\
 &\hspace{4cm}= \log\bigg(\frac{2(L^2-s^2)+2\sqrt{(L^2-s^2)^2+a^2r^2(s)}}{a^2r^2(s)} \bigg) + \mc{O}(a\log a).
\end{aligned}
\end{equation}
Subtracting \eqref{ID1} with $a=\eta\epsilon$ from \eqref{ID1} with $a=\epsilon$ and using that
\begin{align*}
\abs{\log\bigg(\frac{(L^2-s^2) + \sqrt{L^2+\epsilon^2r^2}}{(L^2-s^2) + \sqrt{L^2+\eta^2\epsilon^2r^2}} \bigg) } =\abs{\log\bigg(\frac{(L^2-s^2) + \sqrt{L^2+\epsilon^2r^2}}{(L^2-s^2) + \sqrt{L^2+\eta^2\epsilon^2r^2}} \bigg)  - \log(1) }\le C\epsilon^2,
\end{align*}
we obtain \eqref{eta_ID}.
\end{proof}

Using \eqref{eta_ID}, we replace the first term in the integrand of \eqref{SB_new2} to obtain \eqref{SB_new3}. We can compare the expression \eqref{SB_new3} to that of Tornberg and Shelley in \cite{tornberg2004simulating}, where a regularization of the Keller--Rubinow model is used to obtain a second-kind integral equation for $\bm{f}$. One thing to note is that, due to the form of the local term in our model \eqref{SB_new3}, the effect of the regularization parameter $\eta$ is the same in all directions (both tangent and normal to the fiber centerline). This is not necessarily the case for the Tornberg and Shelley model (see Section \ref{reg_comp} for a spectral comparison given a simplified fiber geometry).

\subsection{Spectral comparison of slender body integral operators}\label{spec_sec}
In this subsection we provide evidence that our model \eqref{SB_new3} is well suited for approximating the map $\frac{\p\X}{\p t}\mapsto\bm{f}$ needed to simulate the motion of a rigid fiber.
Here we consider the spectrum of the integral operator taking the force density $\bm{f}$ to the fiber velocity $\frac{\p\X}{\p t}$ in the non-physical but nevertheless instructive case of a straight, periodic fiber with constant radius $\epsilon$. In this scenario we can explicitly calculate the eigenvalues of both the slender body PDE operator \eqref{SB_PDE} as well as the integral operator \eqref{SB_new2} and related models. This allows us to directly compare the properties of different models in the same simple setting and serves as a starting point for understanding more complicated geometries. In particular, we expect this analysis to roughly capture the high wavenumber behavior of these models in different geometries -- on length scales much smaller than the variation in curvature and fiber radius. The high wavenumber behavior is of particular interest for the invertibility and stability of the slender body theory integral operator. \\

For comparison, we first recall the form of the eigenvalues of the slender body PDE \eqref{SB_PDE}, calculated in \cite{spectral_calc}. In Section \ref{pre_reg}, we consider the model \eqref{SB_new2}, before regularization, and show that the integral operator is negative definite. We compare the spectrum of \eqref{SB_new2} to three other possible models based on slender body theory which do not result in negative definite operators. Then in Section \ref{reg_comp}, we consider the regularized version of our model \eqref{SB_new3} and compare its spectrum to the regularized model of Tornberg and Shelley \cite{tornberg2004simulating}. We note that in our model, a uniform regularization parameter appears to give the best approximation of the slender body PDE spectrum in directions both normal and tangent to the slender body centerline, whereas in the Tornberg--Shelley model, the parameter required by the tangential direction may not be optimal in the normal direction.

\subsubsection{Spectrum of the slender body PDE}\label{SB_PDE_eigs}
Here we consider a straight, periodic fiber with constant radius $\epsilon$. We take the fiber centerline to be 2-periodic and lie along the $z$-axis, $\X(z)= z\be_z$, $z\in \RR/2\Z$, and for simplicity take $\mu=1$ and zero background flow. We consider the stationary setting and omit the time dependence in our notation; in particular, we denote the fiber velocity by $\overline{\bu}(z)$ to distinguish from the fluid velocity away from the fiber. \\

We consider this scenario because we can explicitly calculate the eigenvalues of the slender body PDE \eqref{SB_PDE} as well as various possible integral expressions for approximating the map $\bm{f}\mapsto\overline{\bu}$. In particular, the eigenvectors of this map can be decomposed into tangential ($\be_z$) and normal ($\be_x,\be_y$) directions and are given by $\bm{f}_m(z)= e^{i\pi k z}\be_m$, $m=x,y,z$. We may then explicitly solve for $\lambda^m_k$ satisfying
\begin{equation}\label{evaleq1}
\overline{\bu}(z) = \lambda^m_k \bm{f}_m(z), \qquad m=x,y,z
\end{equation}
for both the slender body PDE operator and various approximations based on slender body theory. To avoid logarithmic growth of the corresponding bulk velocity field at spatial infinity, we will ignore translational modes ($k=0$) in the following spectral analysis. Clearly these modes are important, especially for a rigid body; however, we are mainly interested in the high wavenumber behavior of these operators. High wavenumber instabilities are a known issue for nonlocal slender body theory \cite{gotz2000interactions,shelley2000stokesian,tornberg2004simulating}, and the following analysis likely captures the behavior of these models at high wavenumbers (small length scales) even in curved geometries.   \\

To begin, the eigenvalues of the slender body PDE operator \eqref{SB_PDE} mapping $\bm{f}$ to $\overline{\bu}$ were calculated in \cite{spectral_calc}, Proposition 1.4. Note that the sign convention in this paper is opposite, as we are considering $\bm{f}$ to be the hydrodynamic force exerted \emph{by} rather than \emph{on} the slender body. For the slender body PDE, the eigenvalues satisfying \eqref{evaleq1} in the tangential and normal directions, respectively, are given by
\begin{equation}\label{PDE_eigs}
\lambda^m_k = \begin{cases}
- \frac{2K_0K_1 + \pi\epsilon\abs{k} \big( K_0^2 - K_1^2 \big) }{ 4\pi^2\epsilon\abs{k} K_1^2}, & m=z \\
- \frac{2K_0K_1K_2 + \pi\epsilon\abs{k} \big(K_1^2(K_0+K_2)-2K_0^2K_2 \big)}{2\pi^2\epsilon\abs{k}\big(4K_1^2K_2+\pi\epsilon\abs{k} K_1(K_1^2-K_0K_2)\big)}, & m=x,y
\end{cases}
\end{equation}
where each $K_j=K_j(\pi\epsilon\abs{k})$, $j=0,1,2$, is a $j^\text{th}$ order modified Bessel function of the second kind. Note that both sets of eigenvalues $\lambda^z_k$ and $\lambda^x_k,\lambda^y_k$ are strictly negative and decay to 0 at a rate proportional to $1/\abs{k}$ as $\abs{k}\to\infty$. We will compare our approximation and various other slender body approximations to \eqref{PDE_eigs}.

\subsubsection{Pre-regularization comparison}\label{pre_reg}
Before we consider the regularized version \eqref{SB_new3} of our model, we consider the base model \eqref{SB_new2} and compare its spectrum to other existing models based on slender body theory, before regularization.
In the straight-but-periodic scenario, our model \eqref{SB_new2} becomes the periodization of the expression
\begin{equation}\label{SB_straight}
\overline{\bu}(z) = -\frac{1}{8\pi} \int_{-1}^1 \bigg(\frac{\bf I}{(\barz^2 + \epsilon^2)^{1/2}} + \frac{\barz^2\be_z\be_z^{\rm T} }{(\barz^2 + \epsilon^2)^{3/2}} + \frac{\epsilon^2}{2} \bigg( \frac{\bf I}{(\barz^2 + \epsilon^2)^{3/2}} - 3\frac{\barz^2\be_z\be_z^{\rm T} }{(\barz^2 + \epsilon^2)^{5/2}} \bigg) \bigg) \bm{f}(z-\barz) \, d\barz .
\end{equation}

For this geometry, we may calculate the eigenvalues $\lambda^m_k$ satisfying \eqref{evaleq1}, which are given by
\begin{align}\label{eig_eqs}
\lambda^m_k &= \begin{cases}
\displaystyle -\frac{1}{8\pi} \int_{-1}^1 \frac{2\barz^4+ 2\epsilon^2\barz^2+ \frac{3}{2}\epsilon^4}{(\barz^2 + \epsilon^2)^{5/2}} e^{-i\pi k\barz} \, d\barz ,  & m=z \\
\displaystyle -\frac{1}{8\pi} \int_{-1}^1 \frac{\barz^2 + \frac{3}{2}\epsilon^2}{(\barz^2 + \epsilon^2)^{3/2}} e^{-i\pi k \barz} \, d\barz,  & m=x,y.
\end{cases}
\end{align}

These integrals may be computed explicitly to obtain
\begin{equation}\label{our_eigs}
\lambda^m_k =
\begin{cases}
\displaystyle -\frac{1}{8 \pi} \bigg( (4+ \pi^2\epsilon^2k^2 ) K_0(\pi\epsilon\abs{k}) - 2\pi\epsilon\abs{k} K_1(\pi\epsilon\abs{k}) \bigg),  & m=z \\
 \displaystyle -\frac{1}{8\pi} \bigg( 2K_0(\pi\epsilon\abs{k}) + \pi\epsilon\abs{k} K_1(\pi\epsilon\abs{k}) \bigg), & m=x,y.
\end{cases}
\end{equation}
Here $K_0$ and $K_1$ are zero and first order modified Bessel functions of the second kind, respectively. The eigenvalues $\lambda^m_k$ lie along the curves plotted in Figure \ref{PreRegFig}. Importantly, these eigenvalues satisfy the following lemma.
\begin{lemma}\label{positivity}
For all $\abs{k}\ge 1$ and $m=x,y,z$, the eigenvalues $\lambda^m_k$ given by \eqref{our_eigs} satisfy $\lambda^m_k<0$.
\end{lemma}

\begin{proof}
The case $m=x,y$ is immediate, since $K_0(t)>0$ and $K_1(t)>0$ for any $t>0$. \\

For the tangential direction $m=z$, we first note that, by Lemma 1.16 in \cite{spectral_calc}, we have
\[ 1 \le \frac{K_1(t)}{K_0(t)} \le 1+ \frac{1}{2t}\]
for all $t>0$. Letting $g(t)= (4+t^2)K_0(t) - 2tK_1(t)$, it suffices to show that $g(t)/K_0(t)>0$. But
\[ \frac{g(t)}{K_0(t)} = 4+t^2 - 2t\frac{K_1(t)}{K_0(t)} \ge 3+t^2- 2t  > (t-\sqrt{3})^2 \ge 0.\]
\end{proof}

Now, at a continuous level, regularization is necessary to make sense of inverting the integral operator \eqref{SB_straight}, since $K_0$ and $K_1$ decay exponentially as $\abs{k}\to\infty$. However, at a discrete level, numerical approximation of \eqref{SB_straight} will be invertible, albeit with a large condition number, due to Lemma \ref{positivity}.
This negativity does not hold for other popular slender body approximations which rely on further asymptotic expansion of \eqref{SB_new2} with respect to $\epsilon$ to obtain a limiting centerline velocity expression. In particular, we consider the models of Keller and Rubinow \cite{keller1976slender} and of Lighthill \cite{lighthill1976flagellar}. \\

The Keller--Rubinow model, proposed in \cite{keller1976slender} and further studied by \cite{gotz2000interactions,johnson1980improved,shelley2000stokesian,tornberg2004simulating}, is equivalent to a full matched asymptotic expansion of \eqref{SB_new} about $\epsilon=0$. In the straight-but-periodic setting, the Keller--Rubinow expression for the slender body velocity is given by
\begin{equation}\label{KR_straight}
8\pi \overline{\bu}(z) = -\bigg(({\bf I} -3\be_z\be_z^{\rm T}) - 2\log(\pi\epsilon/8)({\bf I}+\be_z\be_z^{\rm T})\bigg)\bm{f}(z) - ({\bf I}+\be_z\be_z^{\rm T}) \frac{\pi}{2}\int_{-1}^1 \frac{\bm{f}(z-\barz)-\bm{f}(z)}{\abs{\sin(\pi\barz/2)}} \, d\barz.
\end{equation}

The eigenvalues of the periodic Keller--Rubinow operator taking $\bm{f}$ to $\overline{\bu}$ have been calculated in \cite{gotz2000interactions, shelley2000stokesian, tornberg2004simulating} and are given by
\begin{equation}\label{KR_eigs_OG}
\lambda^m_k = \begin{cases}
\displaystyle \frac{1}{4\pi} \big(1 + 2\log(\pi\epsilon\abs{k}/2) + 2\gamma \big) ,  & m=z \\
\displaystyle -\frac{1}{8\pi} \big( 1 - 2\log(\pi\epsilon\abs{k}/2) - 2\gamma \big) , & m=x,y.
\end{cases}
\end{equation}
Here $\gamma\approx 0.5772$ is the Euler gamma.\\

In both the tangent and normal directions, however, the Keller--Rubinow approximation runs into stability issues at moderately high wavenumbers, apparent in Figure \ref{PreRegFig} at $\abs{k}=\frac{2 e^{-\gamma-1/2}}{\pi\epsilon} \approx 0.217/\epsilon$ (tangent) and $\abs{k}=\frac{2 e^{-\gamma+1/2}}{\pi\epsilon} \approx 0.589/\epsilon$ (normal). In particular, the curve containing the eigenvalues $\lambda^m_k$ crosses zero and becomes negative. This is an issue both because the slender body PDE eigenvalues \eqref{PDE_eigs} are strictly negative, and because, for arbitrary $\epsilon$, there is no clear way to guarantee that $\lambda^m_k\neq 0$, especially for more complicated fiber geometries. Thus some sort of regularization of \eqref{KR_straight} is necessary before approximating the inverse map $\overline{\bu}\mapsto \bm{f}$.  \\

\begin{figure}[ht]
\centering
   \includegraphics[width=.49\textwidth]{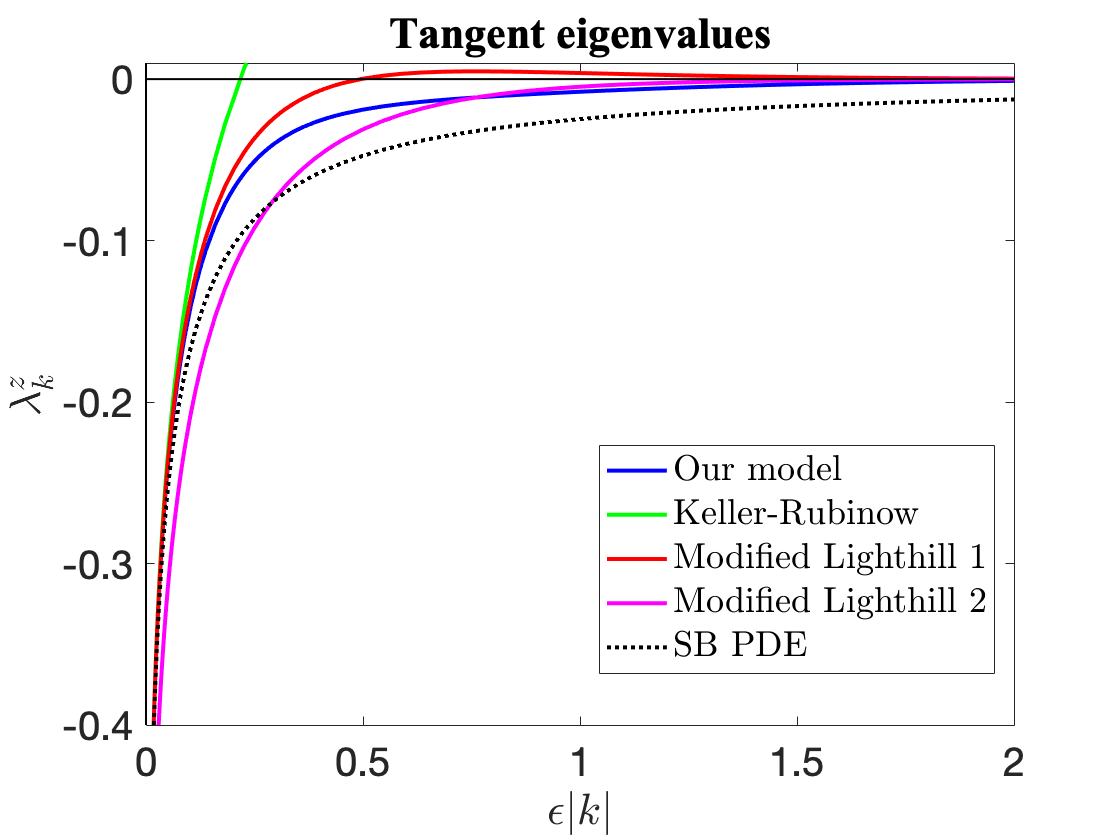}
   \includegraphics[width=.49\textwidth]{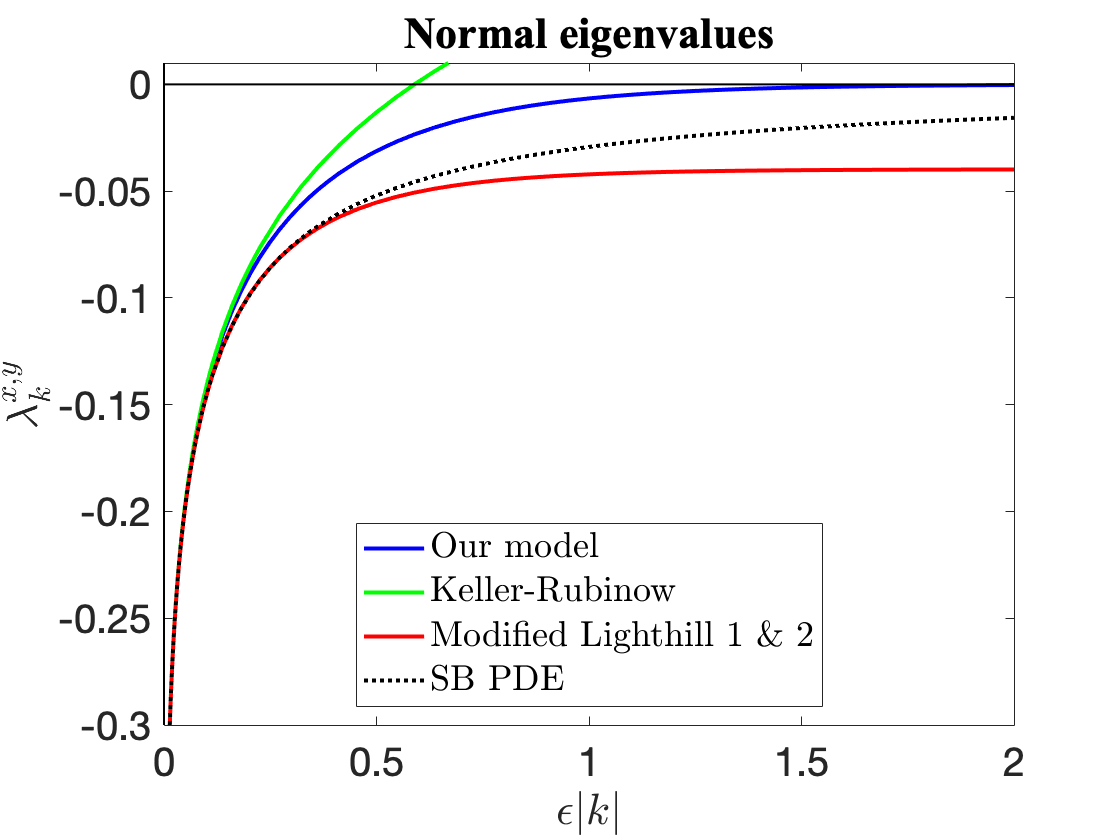}
  \caption{The eigenvalues $\lambda^m_k$ of the operator mapping $\bm{f}\mapsto \overline{\bu}$ in various slender body models lie along the curves plotted here in the case of a straight-but-periodic fiber. Our model (blue) results in strictly negative eigenvalues in both the tangential and normal directions, as does the slender body PDE (dotted). The Keller--Rubinow approximation (green) exhibits instabilities at wavenumbers $\abs{k}\approx 0.2/\epsilon$ (tangential direction) and $\abs{k}\approx 0.6/\epsilon$ (normal direction) as the eigenvalues of the operator mapping $\bm{f}\mapsto \overline{\bu}$ become negative. For the modified Lighthill models, the normal direction eigenvalues $\lambda^x_k$ and $\lambda^y_k$ (red) remain negative at high wavenumber, but in the tangential direction, the eigenvalues of Modified Lighthill 1 (red) become negative when $\abs{k}>0.5/\epsilon$. Furthermore, the tangential eigenvalues of Modified Lighthill 2 (magenta) do not agree with the slender body PDE at low wavenumber. }
  \label{PreRegFig}
\end{figure}

In addition to the Keller--Rubinow model, we consider what we will term the \emph{modified Lighthill} approach to deriving a fiber velocity approximation. This approach takes advantage of the fact that the doublet term of \eqref{SB_new2} only has an $O(1)$ contribution to the fiber velocity very close to $s'=s$, and thus can be integrated asymptotically to leave only a local term. This results in a model similar to that of Lighthill \cite{lighthill1976flagellar}, which was derived via different reasoning but also includes a local doublet term and a nonlocal Stokeslet contribution (see Remark \ref{remark}). \\

There are two ways to consider the nonlocal Stokeslet contribution. The first expression, which we will term Modified Lighthill 1, is given by the periodization of
\begin{equation}\label{ML_straight1}
\overline{\bu}(z) = -\frac{1}{8\pi} \bigg( ({\bf I}-\be_z\be_z^{\rm T}) \bm{f}(s) + \int_{-1}^1 \bigg(\frac{\bf I}{(\barz^2 + \epsilon^2)^{1/2}} + \frac{\barz^2\be_z\be_z^{\rm T} }{(\barz^2 + \epsilon^2)^{3/2}} \bigg) \bm{f}(z-\barz) \, d\barz \bigg) .
\end{equation}
Here the local term $({\bf I}-\be_z\be_z^{\rm T})$ comes from asymptotically integrating the doublet term of \eqref{SB_new} (see estimate 3.65 of \cite{closed_loop} for more detail). Note that in \eqref{ML_straight1}, the Stokeslet term inside the integral is equal to $\bm{f}/\epsilon$ when $\barz=0$.\\

For the second expression, which we will call Modified Lighthill 2, the $\be_z\be_z^{\rm T}$ component of the Stokeslet term is normalized to give the same order contribution at $\barz=0$ as in \eqref{SB_new}; namely, $({\bf I}+\be_z\be_z^{\rm T})\bm{f}/\epsilon$. This yields the periodization of the expression
\begin{equation}\label{ML_straight2}
\overline{\bu}(z) = -\frac{1}{8\pi} \bigg( ({\bf I}-\be_z\be_z^{\rm T}) \bm{f}(s) + \int_{-1}^1 \frac{{\bf I}+\be_z\be_z^{\rm T}}{(\barz^2 + \epsilon^2)^{1/2}} \bm{f}(z-\barz) \, d\barz \bigg) .
\end{equation}

\begin{remark}\label{remark}
The actual model proposed by Lighthill in \cite{lighthill1976flagellar}, written in the periodic, straight setting, has the form
\begin{equation}\label{lighthill0}
\overline{\bu}(z) =  -\frac{1}{8\pi} \bigg( 2({\bf I}-\be_z\be_z^{\rm T}) \bm{f}(z) + \int_{\abs{\barz} > q} \frac{{\bf I}+\be_z\be_z^{\rm T}}{\abs{\barz}} \bm{f}(z-\barz) \, d\barz \bigg); \quad q= \epsilon\sqrt{e}/2.
\end{equation}
At first glance, this looks like a slightly different model from \eqref{ML_straight1} and \eqref{ML_straight2}, due to the 2 in front of the $({\bf I}-\be_z\be_z^{\rm T}) \bm{f}(z)$ term. However, the extra factor here is precisely due to the removal of the section $\abs{\barz}\le q$ from the integral term. Indeed, if we consider the integrand of \eqref{ML_straight1}, we note that
\begin{align*}
\int_{-q}^q \bigg(\frac{\bf I}{(\barz^2 + \epsilon^2)^{1/2}} + \frac{\barz^2\be_z\be_z^{\rm T} }{(\barz^2 + \epsilon^2)^{3/2}} \bigg) \bm{f}(z-\barz) \, d\barz &= \big(2\log(2q/\epsilon)({\bf I}+\be_z\be_z^{\rm T}) - 2\be_z\be_z^{\rm T} \big)\bm{f}(z) + O(\epsilon^2/q^2) \\
&= ({\bf I}-\be_z\be_z^{\rm T})\bm{f}(z) + O(\epsilon^2/q^2)
\end{align*}
for $q$ as in \eqref{lighthill0}. Now, this particular choice of $q$ is not large relative to $\epsilon$, so the $O(\epsilon^2/q^2)$ error term is not small asymptotically. However, this is merely a heuristic and we will not be considering the expression \eqref{lighthill0} in greater depth here. Furthermore, the expressions \eqref{ML_straight1} and \eqref{ML_straight2} are more amenable to calculating eigenvalues.
\end{remark}

The eigenvalues of \eqref{ML_straight1} are given by
\begin{equation}\label{ML1_eigs}
\begin{aligned}
\lambda^m_k &= \begin{cases}
\displaystyle -\frac{1}{4\pi} \bigg( 2K_0(\pi\epsilon\abs{k}) - \pi\epsilon\abs{k} K_1(\pi\epsilon\abs{k}) \bigg),  & m=z \\
\displaystyle -\frac{1}{8\pi} \bigg( 1 + 2K_0(\pi\epsilon\abs{k}) \bigg), & m=x,y.
\end{cases}
\end{aligned}
\end{equation}
Now the normal eigenvalues $\lambda^x_k$ and $\lambda^y_k$ are always negative. However, there is still a high wavenumber instability in the tangent direction. In particular, $\lambda^z_k=0$ when $\pi\epsilon\abs{k} \approx 1.55265$, and becomes positive at higher wavenumbers (see Figure \ref{PreRegFig}). Thus the instability issue is not fully resolved by expanding only the doublet term of \eqref{SB_new}. \\

For Modified Lighthill 2, the eigenvalues of \eqref{ML_straight2} are given by
\begin{equation}\label{ML2_eigs}
\begin{aligned}
\lambda^m_k &= \begin{cases}
\displaystyle -\frac{1}{2\pi}K_0(\pi\epsilon\abs{k}) ,  & m=z \\
\displaystyle -\frac{1}{8\pi} \bigg( 1 + 2K_0(\pi\epsilon\abs{k}) \bigg), & m=x,y.
\end{cases}
\end{aligned}
\end{equation}
Here the eigenvalues $\lambda^x_k$ and $\lambda^y_k$ in the normal directions are identical to \eqref{ML1_eigs}, but the tangential eigenvalues $\lambda^z_k$ are very different. In fact, they are too different: Recall that near $t=0$, the modified Bessel functions $K_0(t)$ and $K_1(t)$ satisfy
\begin{equation}\label{bessel_expand}
K_0(t) = -\log(t/2) - \gamma + O(t^2); \quad tK_1(t) = 1 + O(t^2).
\end{equation}

Therefore, at low wavenumber ($k=O(1)$), the tangential eigenvalues of Modified Lighthill 2 \eqref{ML_straight2} look like
\[ \lambda^z_k = \frac{1}{2\pi} (\log(\pi\epsilon\abs{k}/2)+\gamma ) + O(\epsilon^2k^2). \]
This does not agree with the low wavenumber behavior of the slender body PDE \eqref{PDE_eigs} (see Figure \ref{PreRegFig}). It appears that the normalization in Modified Lighthill 2 \eqref{ML_straight2} results in the wrong model. \\

For the sake of completeness, we also consider a modification of our model \eqref{SB_new2} in which the $\overline{\X}\overline{\X}^{\rm T}$ terms are normalized as in Modified Lighthill 2 \eqref{ML_straight2} to yield a nonzero contribution to the fiber velocity when $s=s'$. In the case of the periodic straight centerline, the modified version of our model becomes the periodization of
\begin{equation}\label{SB_straight2}
\overline{\bu}(z) = -\frac{1}{8\pi} \int_{-1}^1 \bigg(\frac{{\bf I} + \be_z\be_z^{\rm T}}{(\barz^2 + \epsilon^2)^{1/2}} + \frac{\epsilon^2}{2} \frac{{\bf I}- 3\be_z\be_z^{\rm T} }{(\barz^2 + \epsilon^2)^{3/2}}  \bigg) \bm{f}(z-\barz) \, d\barz.
\end{equation}

The eigenvalues of \eqref{SB_straight2} are given by
\begin{equation}\label{our_eigs2}
\begin{aligned}
\lambda^m_k &= \begin{cases}
\displaystyle -\frac{1}{4\pi} \bigg( 2K_0(\pi\epsilon\abs{k}) - \pi\epsilon\abs{k} K_1(\pi\epsilon\abs{k}) \bigg),  & m=z \\
 \displaystyle -\frac{1}{8\pi} \bigg( 2K_0(\pi\epsilon\abs{k}) + \pi\epsilon\abs{k} K_1(\pi\epsilon\abs{k}) \bigg), & m=x,y.
\end{cases}
\end{aligned}
\end{equation}
Now, the eigenvalues $\lambda^x_k$ and $\lambda^y_k$ in the directions normal to the fiber are unchanged from our original expression \eqref{our_eigs}. However, the tangent eigenvalues $\lambda^z_k$ are now given by the same expression as Modified Lighthill 1 \eqref{ML1_eigs}, which we recall exhibits a high wavenumber instability (Figure \ref{PreRegFig}). \\

The takeaway here is that, at least in the case of a straight, periodic fiber, our model \eqref{SB_new2}, before regularization, captures the negative-definiteness of the the slender body PDE and provides a better approximation than other slender body models \eqref{KR_straight}, \eqref{ML_straight1}, \eqref{ML_straight2}, \eqref{SB_straight2}.

\subsubsection{Regularized comparison}\label{reg_comp}
To make our model truly suitable for inversion, we need to regularize the integral kernel as in \eqref{SB_new3}. In the straight-but-periodic setting, the operator in \eqref{SB_new3} becomes the periodization of
\begin{equation}\label{SB_straight_reg}
\begin{aligned}
8\pi \overline{\bu}(z) &= -2\log\eta \, \bm{f}(z) - \int_{-1}^1 \bigg(\frac{\bf I}{(\barz^2 + \eta^2\epsilon^2)^{1/2}} + \frac{\barz^2\be_z\be_z^{\rm T} }{(\barz^2 + \epsilon^2)^{3/2}} \\
&\hspace{4cm}+ \frac{\epsilon^2}{2} \bigg( \frac{\bf I}{(\barz^2 + \epsilon^2)^{3/2}} - 3\frac{\barz^2\be_z\be_z^{\rm T} }{(\barz^2 + \epsilon^2)^{5/2}} \bigg) \bigg) \bm{f}(z-\barz) \, d\barz .
\end{aligned}
\end{equation}
The eigenvalues of \eqref{SB_straight_reg} are then given by
\begin{equation}\label{ourREGeigs}
\lambda^m_k =
\begin{cases}
\displaystyle -\frac{1}{8 \pi} \bigg( 2\log\eta + 2K_0(\eta \pi\epsilon\abs{k}) + (2+ \pi^2\epsilon^2k^2 ) K_0(\pi\epsilon\abs{k}) - 2\pi\epsilon\abs{k} K_1(\pi\epsilon\abs{k}) \bigg),  & m=z \\
 \displaystyle -\frac{1}{8\pi} \bigg( 2\log\eta + 2K_0(\eta \pi\epsilon\abs{k}) + \pi\epsilon\abs{k} K_1(\pi\epsilon\abs{k}) \bigg), & m=x,y.
\end{cases}
\end{equation}
For $\eta>1$, the spectrum of our operator is bounded away from 0 and \eqref{SB_straight_reg} is a second-kind integral equation for $\bm{f}$. \\

We can compare the behavior of \eqref{SB_straight_reg} with the Tornberg--Shelley regularization of the Keller--Rubinow model. In \cite{tornberg2004simulating,shelley2000stokesian}, the high wavenumber instability in \eqref{KR_straight} is removed by replacing the denominator of the integral term, which vanishes at $\barz=0$, with an expression proportional to $\epsilon$ at $\barz=0$. Using the relation
\[ \int_{-1}^1 \bigg(\frac{\pi}{\abs{2\sin(\pi z/2)}} - \frac{1}{\abs{z}} \bigg) \, dz = -2\log(\pi/4),\]
to rewrite \eqref{KR_straight}, a regularization $\delta\epsilon$, $\delta>0$, is added to the denominator to obtain
\begin{equation}\label{KR_reg}
8\pi \overline{\bu}(z) = -\bigg(({\bf I} -3\be_z\be_z^{\rm T}) + 2\log(\delta)({\bf I}+\be_z\be_z^{\rm T})\bigg)\bm{f}(z) - ({\bf I}+\be_z\be_z^{\rm T}) \int_{-1}^1 \frac{\bm{f}(z-\barz)}{(\barz^2+\delta^2\epsilon^2)^{1/2}} \, d\barz.
\end{equation}
Here we have also used that the second term in the original Keller--Rubinow integral expression can now be integrated up to $O(\epsilon^2)$ errors to nearly cancel the logarithmic term in \eqref{KR_straight}, leaving only $\log(\delta)$. The idea is to then choose $\delta$ such that all eigenvalues of the operator taking $\bm{f} \mapsto\overline{\bu}$ are negative. Since the integral kernel is now smooth, \eqref{KR_reg} is now a second-kind integral equation for $\bm{f}$. \\

The eigenvalues of this $\delta$-regularized Keller--Rubinow operator are given by
\begin{equation}\label{KR_eigs_reg}
\begin{aligned}
\lambda^m_k &= \begin{cases}
\displaystyle -\frac{1}{4\pi} \bigg( -1+ 2\log\delta + 2 K_0(\delta \pi\epsilon\abs{k}) \bigg) ,  & m=z \\
\displaystyle -\frac{1}{8\pi} \bigg( 1+ 2\log\delta + 2K_0(\delta \pi\epsilon\abs{k}) \bigg),  & m=x,y.
\end{cases}
\end{aligned}
\end{equation}
Since $K_0$ is positive, $\lambda^z_k$ is guaranteed to be negative and bounded away from 0 as long as $\delta> \sqrt{e}$ (see Figure \ref{eig_KR_reg}). \\

\begin{figure}[ht]
\centering
   \includegraphics[width=.49\textwidth]{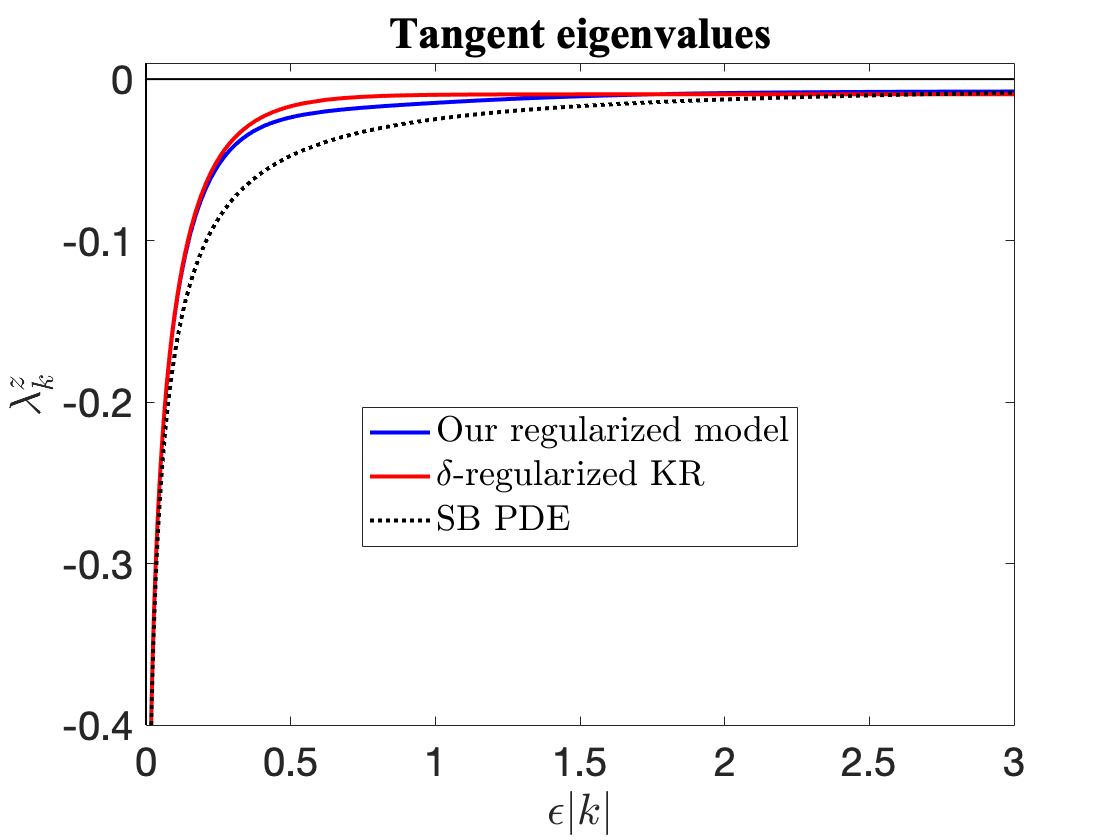}
   \includegraphics[width=.49\textwidth]{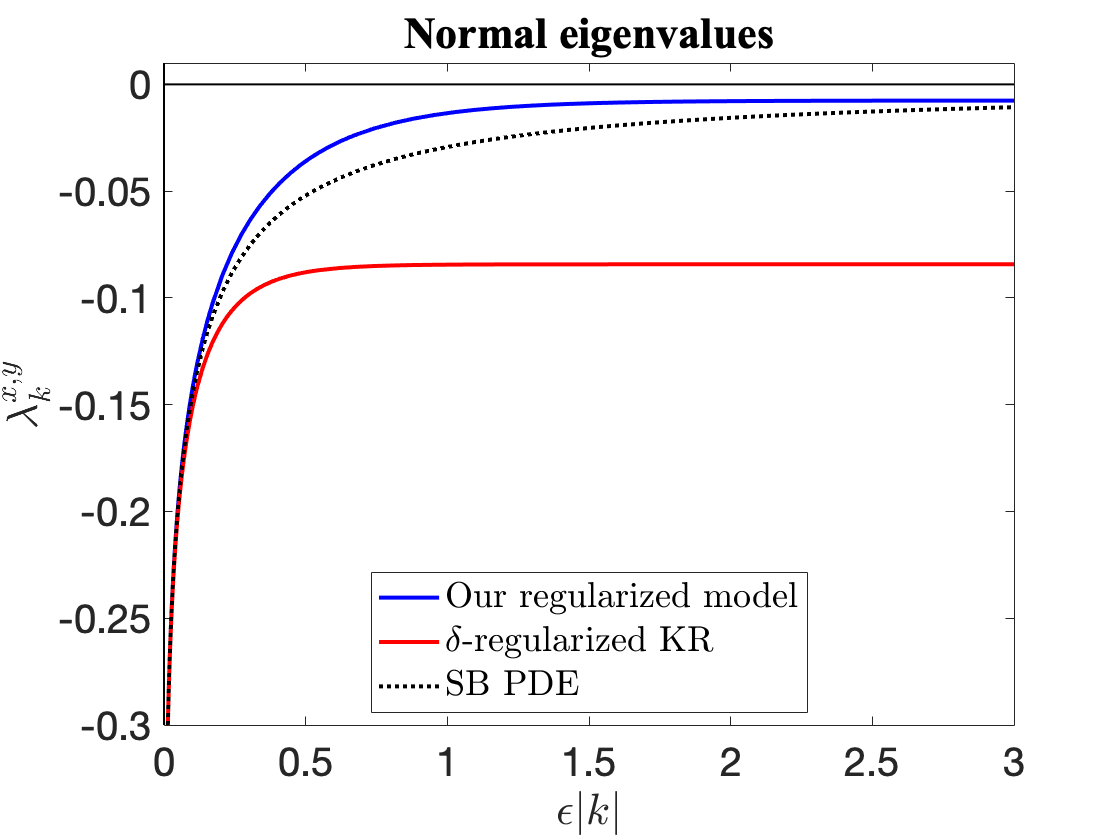}
\caption{The eigenvalues of our regularized model \eqref{SB_straight_reg} with $\eta=1.1$ and the Tornberg--Shelley $\delta$-regularized model \eqref{KR_reg} with $\delta=\sqrt{e}+0.1$ lie along the blue and red curves, respectively. Note that the regularization parameter $\eta$ in our model affects the tangential and normal eigenvalues in a similar way; in particular, $\eta>1$ is required in both cases to ensure that \eqref{SB_straight_reg} is a second-kind integral equation. In the $\delta$-regularized model, the tangential direction requires $\delta>\sqrt{e}$, but the normal direction does not, resulting at least visually in a greater disparity between the $\lambda^x_k,\lambda^y_k$ for the PDE (dotted) and the $\delta$-regularized approximation.}
\label{eig_KR_reg}
\end{figure}

Note that in our model \eqref{SB_straight_reg}, the regularization parameter $\eta$ affects the spectrum of the operator mapping $\bm{f}$ to $\overline{\bu}$ in the same way in both the tangential and normal directions. In particular, in both directions, $\eta>1$ is required to obtain the desired second-kind integral equation. In the Tornberg--Shelley model, the bound $\delta>\sqrt{e}\approx 1.649$ is required to ensure negativity of the tangential eigenvalues, but this lower bound does not apply to the normal direction; in fact, $\delta>e^{-1}\approx 0.368$ is sufficient for ensuring strictly negative normal eigenvalues. This may mean that our model can achieve better agreement with the slender body PDE in both the tangent and normal directions at the same time.  \\

In \cite{spectral_calc}, it is shown that using the $\delta$-regularized model \eqref{KR_reg} to approximate the map $\overline{\bu}\mapsto \bm{f}$ yields $\epsilon^2$ convergence to the slender body PDE for sufficiently smooth $\overline{\bu}$. It is also shown that the constant in the resulting error estimate has the form $C_1\delta^2(1+\log\delta)+C_2/(-1+\log\delta)$ for constants $C_1$ and $C_2$. We expect that a similar error estimate and analogous $\eta$ dependence hold for our model \eqref{SB_straight_reg}; i.e. the constant should look like $C_1\eta^2+C_2/\log\eta$. If $C_1\approx C_2$, this yields an optimal $\eta$ of approximately $1.5$. This should give a rough guideline for a good choice of $\eta$ for more general curved geometries, at least in the periodic setting.



\section{Numerical discretization of the slender body model}\label{sec:disc}
We turn now to a numerical method for simulating thin rigid fibers in flows.  
We begin by discussing a general method for numerically solving Fredholm integral equations where the solution must be integrated (i.e. to find the total force and torque on a rigid fiber). 
We apply these general methods to the slender body model \eqref{SB_new3} and perform convergence tests. We note improvements in conditioning and stability for the second kind ($\eta>1$) versus first kind ($\eta=1$) integral equation. Finally, we look at the spectrum of the discretized integral operator in different geometries to verify the negative definite nature of the operator. 

\subsection{Solving the second-kind Fredholm integral equation}\label{sec:fred}

Denote by $\mathbf{K}:L^2([-L,L],\mathbb{R}^3)\rightarrow L^2([-L,L],\mathbb{R}^3)$ the integral operator
\begin{equation}\label{int1}
\mathbf{K}[\bff](s) := \int_{-L}^{L}\!K(s,s')\bff(s')ds'.
\end{equation}
Then a Fredholm integral equation of the first kind reads
\begin{equation}\label{intop}
\by(s) = \mathbf{K}[\bff](s).
\end{equation}
It is well known that the inversion of such an integral operator is an ill-posed problem, meaning that the solution may not be unique or not even exist \cite{atkinson2005theoretical,hansen1992numerical, kress1989linear}. Furthermore, small perturbations to the left hand side of \eqref{intop} can lead to relatively large perturbations of the solution $\bff(s)$. 
The ill-posedness of this problem can be circumvented by regularizing the integral operator into a second-kind Fredholm integral equation, which takes the form
\begin{equation}\label{intop2}
\by(s) = (\alpha \mathbf{I} + \mathbf{K})[\bff](s)
\end{equation}
for some parameter $\alpha$. Discretization of \eqref{intop2} yields a linear system with a far better condition number. The connection between equation \eqref{intop2} and our model is illustrated in Section \ref{app2SBT}.\\

Numerical methods for solving Fredholm integral equations are well documented \cite{hansen1992numerical,twomey1963numerical} and the approach we adopt is based on the Nystr{\"o}m method \cite[Chapt. 12.4]{atkinson2005theoretical}. The main additional consideration for rigid fibers is that after numerically inverting a second-kind Fredholm integral equation, linear functionals \eqref{FandN} need to be applied to the resulting $\bff(s)$ to find the total force and torque.  \\ 

We consider now the numerical approximation of a general linear functional of $\bff(s)$, given by
\begin{equation}\label{Fexpr}
	{\phi}_M(\bff) = \int_{-L}^{L}\! M(s) \bff(s)\dd s.
\end{equation}
Here $M(s)\in\mathbb{R}^{3\times 3}$ is a bounded, smooth operator and $\bff(s)$ is found by numerically inverting a second-kind Fredholm integral equation of the form \eqref{intop2}.
The numerical method is obtained discretizing the equation \eqref{intop2}  by replacing the integral with a convergent quadrature formula with nodes $-L=s_1<s_2<...<s_n=L$ and weights $\bw = (w_1,w_2,...,w_n)^T\in\mathbb{R}^{n}$,  and requiring the numerical approximation $\bff_i^{[n]} \approx \bff(s_i)$ to satisfy
\begin{equation} \label{dSBT}
\by(s_i) =  \alpha \, \bff_i^{[n]} + \sum_{j=1}^n w_jK(s_i,s_j)\bff_j^{[n]} \quad\text{for}\quad i=1,...,n.
\end{equation}

Introducing the vectors $\underline{\bff}^{[n]} = ((\bff_1^{[n]})^T,...,(\bff_n^{[n]})^T)^T$ and $\underline{\by} = (\by(s_1)^T,...,\by(s_n)^T)^T$, 
equation \eqref{dSBT} can be written compactly as 
\begin{equation}\label{yeq}
\underline{\by} =\left(\alpha\,I+\underline{K}\,\underline{W}\right)\,\underline{{\bff}}^{[n]}.
\end{equation}
Here $I$ denotes the $3n\times 3n$ identity matrix, 
and 
\begin{gather}
\underline{W} = \mathrm{diag}({\bm w})\otimes {\bf I}, \quad\text{and}\quad\underline{K} = \left(\begin{array}{ccc}
	K(s_1,s_1) & \dots & K(s_1,s_n) \\
	\vdots & \ddots & \vdots \\
	K(s_n,s_1) & \dots & K(s_n,s_n)\\
	\end{array}\right)\in\mathbb{R}^{3n\times 3n}
\end{gather}
with $\otimes : \mathbb{R}^{n_1\times m_1}\times\mathbb{R}^{n_2\times m_2}\rightarrow\mathbb{R}^{(n_1n_2)\times (m_1m_2)}$ the Kronecker product of matrices and ${\bf I}$ the $3\times 3$ identity matrix. 
We then approximate \eqref{Fexpr} by the same quadrature formula
\begin{align}
{\phi}_M(\bff)\approx \sum_{i=1}^{n}w_iM(s_i){\bff}_i^{[n]}
= (\mathbb{1}^T\otimes {\bf I})\underline{M}\,\underline{W}\,\underline{{\bff}}^{[n]}:={\phi}^{[n]}_M,
\end{align}
where 
\begin{equation}
\underline{M} = \left(\begin{array}{ccc}
M(s_1) & &  0\\
& \ddots &  \\
0 &  & M(s_n)\\
\end{array}\right)\in\mathbb{R}^{3n\times 3n}
\end{equation}
and $\mathbb{1}=(1,\dots, 1)^T\in\mathbb{R}^n.$
Here we have used ${\phi}^{[n]}_M$ to denote the approximation of ${\phi}_M({\bff})$ obtained by quadrature. After inserting the solution of \eqref{yeq}, we obtain 
\begin{equation}\label{F1}
{\phi}^{[n]}_M = (\mathbb{1}^T\otimes {\bf I})\underline{M}\,\underline{W} \left(\alpha I+\underline{K}\,\underline{W}\right)^{-1}\underline{\by}.
\end{equation}

\begin{remark}\label{int_reg}
The linear functional \eqref{Fexpr} has a regularizing effect on numerical solutions to first-kind integral equation \eqref{intop} in the case of a constant fiber radius. This is illustrated via singular value expansion for a thin ring using $M={\bf I}$ in Appendix \ref{app:reg}. This effect applies whenever such a functional is the final quantity of interest, such as in the case of a rigid slender body. 
\end{remark}

\subsubsection{Convergence and error bounds}

We are interested in obtaining an estimate for the error when approximating \eqref{Fexpr} by its discrete approximation \eqref{F1}, which we denote by 
\begin{equation}\label{totalerror}
\mathbf{d}^{[n]} = {\phi_M}({\bff})-\phi^{[n]}_M = \int_{-L}^{L}\!M(s)\bff(s) ds - \sum_{j=1}^{n}w_jM(s_j)\bff_j^{[n]}.
\end{equation}

This error will depend on the error committed in the numerical approximation of \eqref{intop2} by the solution $\underline{\bff}^{[n]}$ of \eqref{yeq}. For this reason, we first analyze the convergence of Nystr{\"o}m's method in using \eqref{yeq} to approximate the solution of \eqref{intop2} \cite[Chapt. 12.4]{atkinson2005theoretical}. At each quadrature node, we define the error of this approximation as
\begin{equation}
\mathbf{e}^{[n]}_i:=\bff(s_i)- \bff_i^{[n]}, \quad\text{for}\quad i=1,\dots , n,
\end{equation}
and let $\mathbf{\underline{e}}^{[n]}:=((\mathbf{e}^{[n]}_1)^T,\dots, (\mathbf{e}^{[n]}_n)^T)^T$ denote the error vector.
We want to show that $\| \mathbf{\underline{e}}^{[n]}\|_{\infty}\rightarrow0$ as $ n\rightarrow \infty$. Let $\underline{{\bff}}:=(\bff(s_1)^T,\dots , \bff(s_n)^T)^T$ and define $\mathbf{\underline{\tau}}^{[n]}:=(\mathbf{\tau}_1^T,\dots ,\mathbf{\tau}_n^T)^T$ with components
\begin{equation}\label{tau}
\mathbf{\tau}_i:=\by(s_i)-\alpha \bff(s_i)-\sum_{j}^nK_{i,j}w_j\bff(s_j),
\end{equation}
the truncation error for the discrete second kind equation \eqref{yeq} -- i.e. the residual obtained replacing $\underline{\bff}^{[n]}$ by $\underline{{\bff}}$ in \eqref{yeq}. 
We obtain
\begin{equation}
\label{te}
\qquad \left(\alpha\,I+\underline{K}\,\underline{W}\right)\,\underline{{\bff}}=\underline{\by} -\mathbf{\underline{\tau}}^{[n]}.
\end{equation}
It is easily seen using \eqref{intop2} that
\begin{equation}
\mathbf{\tau}_i=\int_{-L}^{L}K(s_i,s')\bff(s') ds'-\sum_{j}^nK_{i,j}w_j\bff(s_j),
\end{equation} 
which is simply quadrature error, and for any convergent quadrature formula we have
\begin{equation}
\lim_{n\rightarrow \infty}\|\mathbf{\tau}^{[n]}\|_{\infty}=0.
\end{equation}
We next bound the norm of the error $\mathbf{\underline{e}}^{[n]}$ by the norm of  $\mathbf{\underline{\tau}}^{[n]}$ to prove the convergence of the method. Subtracting \eqref{yeq} from \eqref{te} we obtain a linear system satisfied by $\mathbf{\underline{e}}^{[n]}$:
\begin{equation}\label{ferr}
\left(\alpha\,I+\underline{K}\,\underline{W}\right)\,\mathbf{\underline{e}}^{[n]}= -\mathbf{\underline{\tau}}^{[n]}.
\end{equation}
From \cite[Chapt. 12.4]{atkinson2005theoretical} Theorem 12.4.4 and equation (12.4.51) we have that for sufficiently large $n$, say $n\ge n^*$, the matrix  $\left(\alpha\,I+\underline{K}\,\underline{W}\right)$ is invertible and
\begin{equation}\label{C1}
	\|\left(\alpha\,I+\underline{K}\,\underline{W}\right)^{-1}\|_{\infty}\le C_1 \qquad \forall n \ge n^*.
\end{equation}
Thus we can conclude that
\begin{equation}
\|\mathbf{\underline{e}}^{[n]}\|_{\infty}\le \|\left(\alpha\,I+\underline{K}\,\underline{W}\right)^{-1}\|_{\infty}\, \|\mathbf{\underline{\tau}}^{[n]}\|_{\infty}\le C_1\,\|\mathbf{\underline{\tau}}^{[n]}\|_{\infty}.
\end{equation}
Since $C_1$ is independent of $n$ for $n\ge n^*$ and $\|\mathbf{\underline{\tau}}^{[n]}\|_{\infty}\rightarrow 0$ as $n\rightarrow\infty$, this implies that 
$$\lim_{n\rightarrow \infty}\|\mathbf{\underline{e}}^{[n]}\|_{\infty}=0.$$
Consider now the quadrature error
\begin{align}
\delta^{[n]} :=& \int_{-L}^{L}M(s)\bff(s) ds - \sum_{j=1}^{n}w_jM(s_j)\bff(s_j).  \label{two}
\end{align}
From \eqref{totalerror} we obtain
\begin{equation}
\mathbf{d}^{[n]} = \delta^{[n]} -\sum_{j=1}^{n}w_jM(s_j)\mathbf{e}_j,
\end{equation}
and using  \eqref{ferr} the total discretization error for our methods is given by 
\begin{equation}\label{err}
\mathbf{d}^{[n]} =  (\mathbb{1}^T\otimes \mathbf{I})\underline{W}\,\underline{M}(\alpha I + \underline{K}\,\underline{W})^{-1}\underline{\tau}^{[n]} + \delta^{[n]}
\end{equation}
Since both $\delta^{[n]}$ and $\mathbf{\tau}^{[n]}$ are quadrature errors, $\|(\alpha I + \underline{K}\,\underline{W})^{-1}\|\le C_1$ for all $n\ge n^*$, and $M$ is bounded, the method converges at the same rate as the underlying quadrature.

\subsection{Application to the slender body model and convergence tests}\label{app2SBT}
We next apply our numerical method from Section \ref{sec:fred} to approximate the force and torque on a slender body. Note that the equations \eqref{FandN} are given by setting $M(s)=\mathbf{I}$ and $M(s)=\widehat{\X}(s)$ in the functional \eqref{Fexpr}. That is, 
\begin{equation}\label{eqqq}
	\bm{F} = {\phi}_\mathbf{I}(\bff)\quad\text{and}\quad\bm{T} = {\phi}_\mathbf{\widehat{\X}}(\bff).
\end{equation}
Letting $\alpha = 2\log(\eta)$ and
\begin{align}\label{K}
K(s,s') = \bm{S}_{\epsilon,\eta}(s,s')+\frac{\epsilon^2r^2(s')}{2}\bm{D}_\epsilon(s,s'),
\end{align}
\begin{equation}
	\by(s) = -8\pi\mu(\bv - \widehat{ \X}(s)\bom - \bu_0(\X(s,t),t)),
\end{equation}
our model \eqref{SB_new3} is of the form \eqref{intop2}, and we may write the discretization of \eqref{SB_new3} in the form \eqref{yeq}. Here we have introduced the hat operator $\widehat{\cdot}: \mathbb{R}^3\rightarrow\mathfrak{so}(3)$ which maps vectors in $\mathbb{R}^3$ to $3\times3$ skew symmetric matrices by
\begin{equation}
\bm{\omega}=\left(
\begin{array}{c}
\omega_1\\
\omega_2\\
\omega_3\\
\end{array}
\right)
\mapsto
\widehat{\bom} = \left(
\begin{array}{ccc}
0 & -\omega_3 &  \omega_2 \\
\omega_3 & 0  & -\omega_1 \\
-\omega_2  & \omega_1 & 0   \\
\end{array}
\right).
\end{equation}
Here $\mathfrak{so}(3)$ is the Lie algebra of $SO(3)$, and such that $ \mathbf{\bom}\times\bv = \hat{\bom}\bv$ for $\bom, \bv\in\mathbb{R}^3$. \\

Denote the numerical approximations to \eqref{eqqq} by 
\begin{equation}\label{eqqqn}
\bm{F}^{[n]} = {\phi}^{[n]}_\mathbf{I}\quad\text{and}\quad\bm{T}^{[n]} = {\phi}^{[n]}_\mathbf{\widehat{\X}}.
\end{equation}
Defining the matrices  $\Phi$ and $\Psi\in\mathbb{R}^{3\times 3n}$ as
\begin{align}
	\Phi =&  (\mathbb{1}^T\otimes \mathbf{I})\underline{W}\left(\alpha I+\underline{K}\,\underline{W}\right)^{-1},\label{Phi}\\
	\Psi =&  (\mathbb{1}^T\otimes \mathbf{I})\underline{W}\,\underline{X}\left(\alpha I+\underline{K}\,\underline{W}\right)^{-1}, \label{Psi}
\end{align}
we may then write equations \eqref{eqqqn} as 
\begin{align}
\bm{F}^{[n]} &=\Phi \underline{\by} \quad\text{and}\quad \bm{T}^{[n]} =\Psi \underline{\by} \label{FandT}.
\end{align}

In the next section we perform convergence tests for our discrete model \eqref{FandT} for both a thin ring and a prolate spheroid. With these geometries we are able to calculate accurate reference solutions against which we can compare the accuracy of our numerical solution. We provide an error bound (see \eqref{bound} below) depending on the speed of convergence of the quadrature rule, this error bound is confirmed by our numerical experiments. Furthermore, we will look at how the conditioning of the linear system associated with the discretized integral operator improves as the regularization parameter $\eta$ is increased from $\eta=1$ to $\eta>1$. 

\subsubsection{Convergence of numerical method for closed loop geometry}
By applying the formula \eqref{err}, we now show how one can achieve spectral convergence in the case of a closed fiber geometry with constant radius $\epsilon$ and periodic integration domain. In this setting, we will use trapezoidal quadrature. We begin by bounding the norms of the integration kernels to which we apply the trapezoidal quadrature rules to, namely the integrals \eqref{int1} and \eqref{Fexpr}. Using this, and some smoothness assumptions, we are able bound the quadrature errors $\tau_i^{[n]}$ and $\delta^{[n]}$ using classical error estimates. This leads to a bound on the total error $\mathbf{d}^{[n]}$ for both the force and torque calculation.\\ 

Let $C_2$ be a constant such that 
\begin{equation}
\|\bff(s')\|_\infty\le C_2\quad \text{ for }\quad s\in[-L,L].
\end{equation}
From the definition of $K(s,s')$ (equations \eqref{Sdef}, \eqref{Ddef}, and \eqref{K}) in the constant radius case, we observe that
\begin{equation}
\|K(s,s')\|_\infty\le \frac{3}{2 \epsilon}
\end{equation}
with equality when $s=s'$. From equation \eqref{eqqq} we have $\|M(s)\|_{\infty} = 1$ for the force calculation, while for the torque calculation, $M(s) = \widehat{ \X}(s)$ and therefore
\begin{equation}
\|M(s)\|_{\infty}\le \max_{s\in[-L,L]}\|\X(s)\|_1.
\end{equation}
 Therefore we can bound the integration kernels of \eqref{int1} and \eqref{Fexpr} by
\begin{equation}
\|K(s,s')\bff(s')\|_\infty \le \frac{3}{2 \epsilon} C_2
\end{equation}
and 
\begin{equation}
\|M(s)\bff(s)\|_\infty \le \|M(s)\|_\infty C_2.
\end{equation}
Note that in the constant radius case, $K(s,s')$ has the same regularity as $\X(s)$. If we assume that $\X(s)$, $\bff(s)$ and $M(s)$ are analytic, then using \cite[Theorem 3.2]{trefethen2014exponentially} we can bound the trapezoidal rule quadrature error from equation \eqref{tau} by 
\begin{equation}
\|\tau_i^{[n]}\|_\infty \le \frac{6 L C_2}{\epsilon(e^{a n} - 1)} \quad \text{for}\quad i=1,...,n.
\end{equation}
Similarly, we can bound equation \eqref{two} by
\begin{equation}
\|\delta^{[n]}\|_\infty \le \frac{4 L \|M(s)\|_\infty C_2}{e^{a n} - 1}.
\end{equation}
Here $a$ is some constant. Using equation \eqref{err}, the total discretization error is therefore bounded as
\begin{equation}
\|\mathbf{d}^{[n]}\|_\infty \le  \left(\|(\mathbb{1}^T\otimes \mathbf{I})\underline{W}\,\underline{M}(\alpha I + \underline{K}\,\underline{W})^{-1}\|_\infty\frac{3}{2\epsilon} + \|M(s)\|_\infty\right)\frac{4 LC_2}{e^{an} - 1}.
\end{equation}
Using that $\|\underline{M}\|_{\infty} \le \|{M(s)}\|_{\infty}$, $\|\underline{W}\|_\infty = \frac{2L}{n}$ and $C_1$ is given by equation \eqref{C1}, this simplifies to
\begin{equation}\label{bound}
\|\mathbf{d}^{[n]}\|_\infty \le  \left(\frac{6C_1L}{2\epsilon} + 1\right)\frac{4 L\|{M(s)}\|_{\infty}C_2}{e^{an} - 1}.
\end{equation}
Hence, the method shares the same exponential convergence as the underlying trapezoidal rule. We remark that one could perform an analogous analysis for open ended fiber geometries with, e.g., Gauss-Lobatto quadrature, and derive similar results. Furthermore, we also remark that one could require less stringent regularity assumptions on the integration on the kernels or the fiber centreline $\X(s)$, e.g., $M(s)\bff(s)\in C^{2m+2}[-L,L]$. Then \cite[Thm. 5.5]{atkinson1978introduction} can be used to derive asymptotic error estimates for $\tau_i^{[n]}$ and $\delta^{[n]}$ of order $O(h^{2m+2})$. Nonetheless, we do observe spectral convergence in numerical experiments in the following sections, as predicted by the bound \eqref{bound}.

\subsubsection{Thin ring translating with unit velocity}
As a convergence test, we use \eqref{FandT} to calculate the force on a thin ring of unit length in the $xy$-plane translating in the $z$ direction with unit velocity in zero background flow. We will consider both the first- and second-kind formulations of the model. In this setting, the force on the ring can be calculated to arbitrarily high precision by evaluating elliptic integrals, which can be used as a reference solution. For a circular centerline parametrized by  $$\X(s) = 
\left(\frac{\cos(\pi s)}{2\pi},\frac{\sin(\pi s)}{2\pi},0\right)^T,$$ the $z$-component of our unregularized ($\eta = 1$) model becomes
\begin{equation}\label{ring_int}
8\pi\mu =-\int_{-\frac{1}{2}}^{\frac{1}{2}} {\frac {\sqrt {2}\,\pi \left( 3\,{\epsilon}^{2}{\pi}^{2}-\cos \left( 2\pi\,(s-s')
		\right) +1 \right) }{ \left( 2\,{\epsilon}^{2}{\pi}^{2}
		-\cos \left( 2\pi\,(s-s') \right) +1 \right) ^{3/2}}}
f^z(s') \, ds' .
\end{equation}
As in the straight-but-periodic geometry of Section \ref{spec_sec}, the eigenfunctions of this operator are the Fourier modes $f_k^z(s) = \exp(i2\pi ks)$. The force $\bm{F} = (F,0,0)^T$ is therefore given by 
\begin{equation}
F = \int_{-\frac{1}{2}}^{\frac{1}{2}}\!f^z(s){\rm d}s = \frac{8\pi\mu}{\lambda^z_0}
\end{equation}
where $\lambda^z_0$ is the $k=0$ eigenvalue. This can be found by evaluating the integral in equation \eqref{ring_int} with $f^z(s) = f^z_0(s)=1$, which gives
\begin{equation}
\lambda^z_0=-c_\epsilon\left( 2\,\phi_K \left( c_\epsilon \right) +\phi_E \left( c_\epsilon \right)  \right).
\end{equation}
Here $c_\epsilon = \sqrt { \left( {\epsilon}^{
		2}{\pi}^{2}+1 \right) ^{-1}}$, and 
\begin{equation}
\phi_K(x) = \int_{0}^{1}\!{\frac {1}{\sqrt {1-{\theta}^{2}}\sqrt {1-{x}^{2}{\theta
			}^{2}}}}\,{\rm d}\theta\quad\text{and}\quad \phi_E(x) = \int_{0}^{1}\!{\frac {\sqrt {1-{x}^{2}{\theta}^{2}}}{\sqrt {1-{\theta}
			^{2}}}}\,{\rm d}\theta
\end{equation}
are the complete elliptic integrals of the first and second kind, respectively. \\

For $\epsilon = 0.05,0.025,0.01$ and $0.005$, equation \eqref{ring_int} is discretized using trapezoidal quadrature, and our numerical method is used to approximate $F$ by equation \eqref{FandT}. Figure \ref{ring_error} plots the error as a function of $n$ for four different values of $\epsilon$. We observe spectral convergence of the error to machine precision, which is consistent with our error estimates from expression \eqref{bound}. We note that the condition number of the unregularized discrete integral operator grows exponentially as $n$ increases, as shown in Figure \ref{cond_num_unreg}. However, because we are considering a rigid fiber with constant radius, the regularizing effect of computing $F$ comes into effect (see Remark \ref{int_reg}) and thus the effect of this ill-conditioning is not apparent in the final force calculation. This may be contrasted with the case of the prolate spheroid, where, as we will see in Section \ref{subsec:prolate}, the conditioning does have a noticeable effect on the error. Nevertheless, we note that by setting $\eta>1$ we can improve the condition number of the linear system (see Figure \ref{cond_num_reg}). 

\begin{figure}[!ht]
	\centering
	\includegraphics[width=0.5\linewidth]{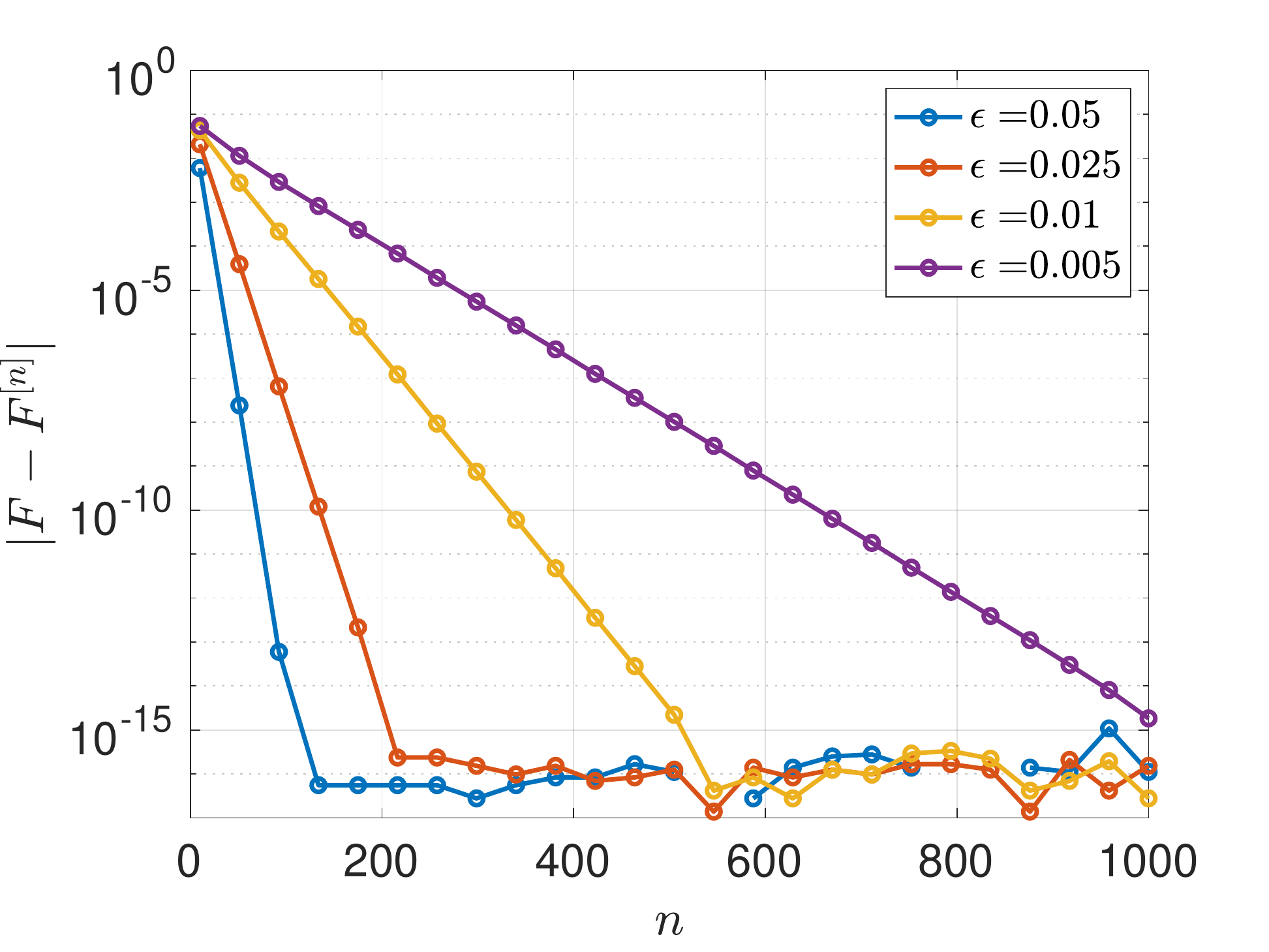}
	\caption{The approximate drag force $F^{[n]}$ on a thin ring translating broadwise with unit velocity converges with spectral accuracy to the true force $F$.}
	\label{ring_error}
\end{figure}

\begin{figure}[!ht]
	\centering
	\begin{subfigure}{.45\textwidth}
		\includegraphics[width=\linewidth]{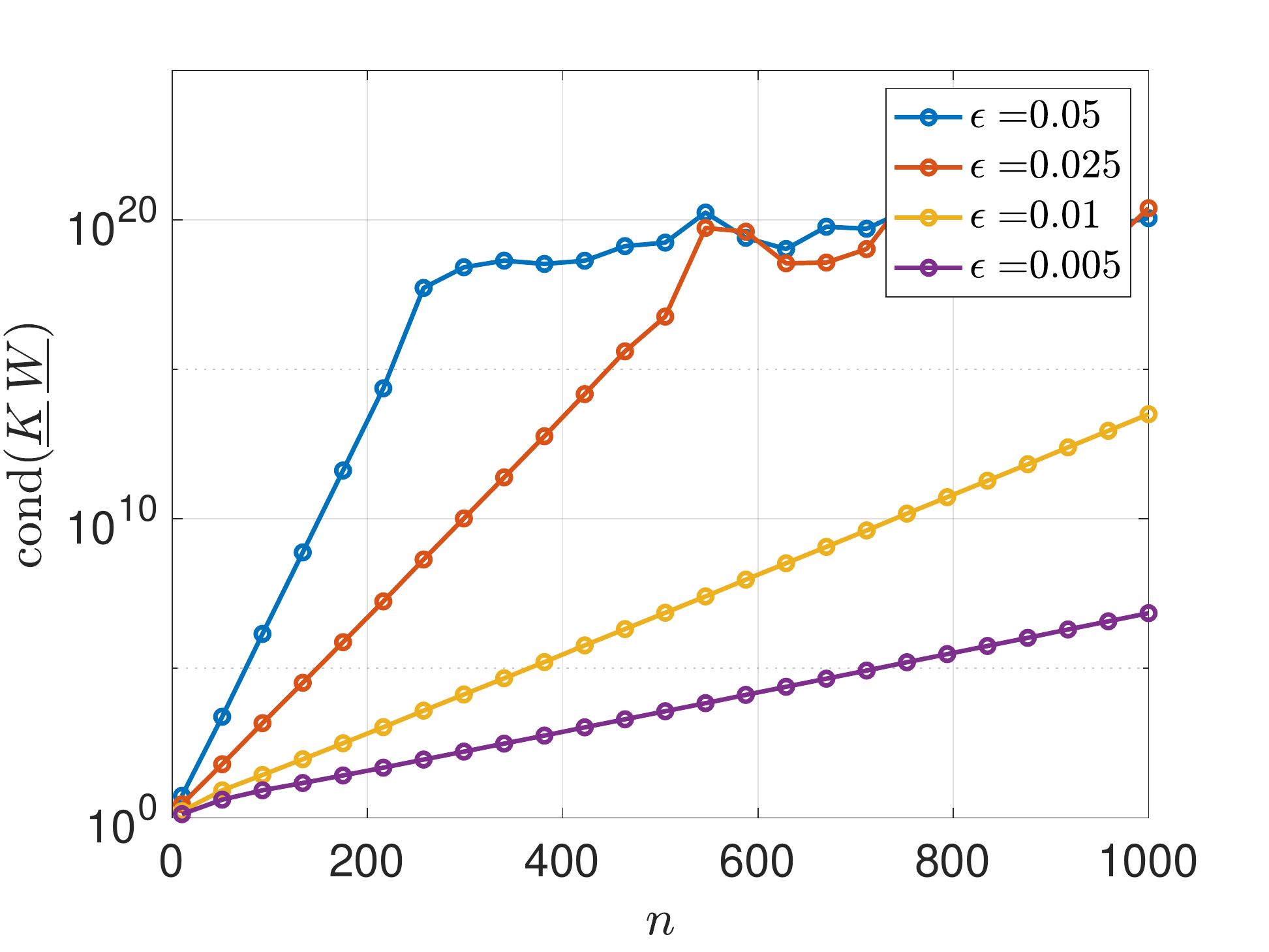}
		\caption{Unregularized ($\eta=1$).}
		\label{cond_num_unreg}
	\end{subfigure}
	\begin{subfigure}{.45\textwidth}
		\includegraphics[width=\linewidth]{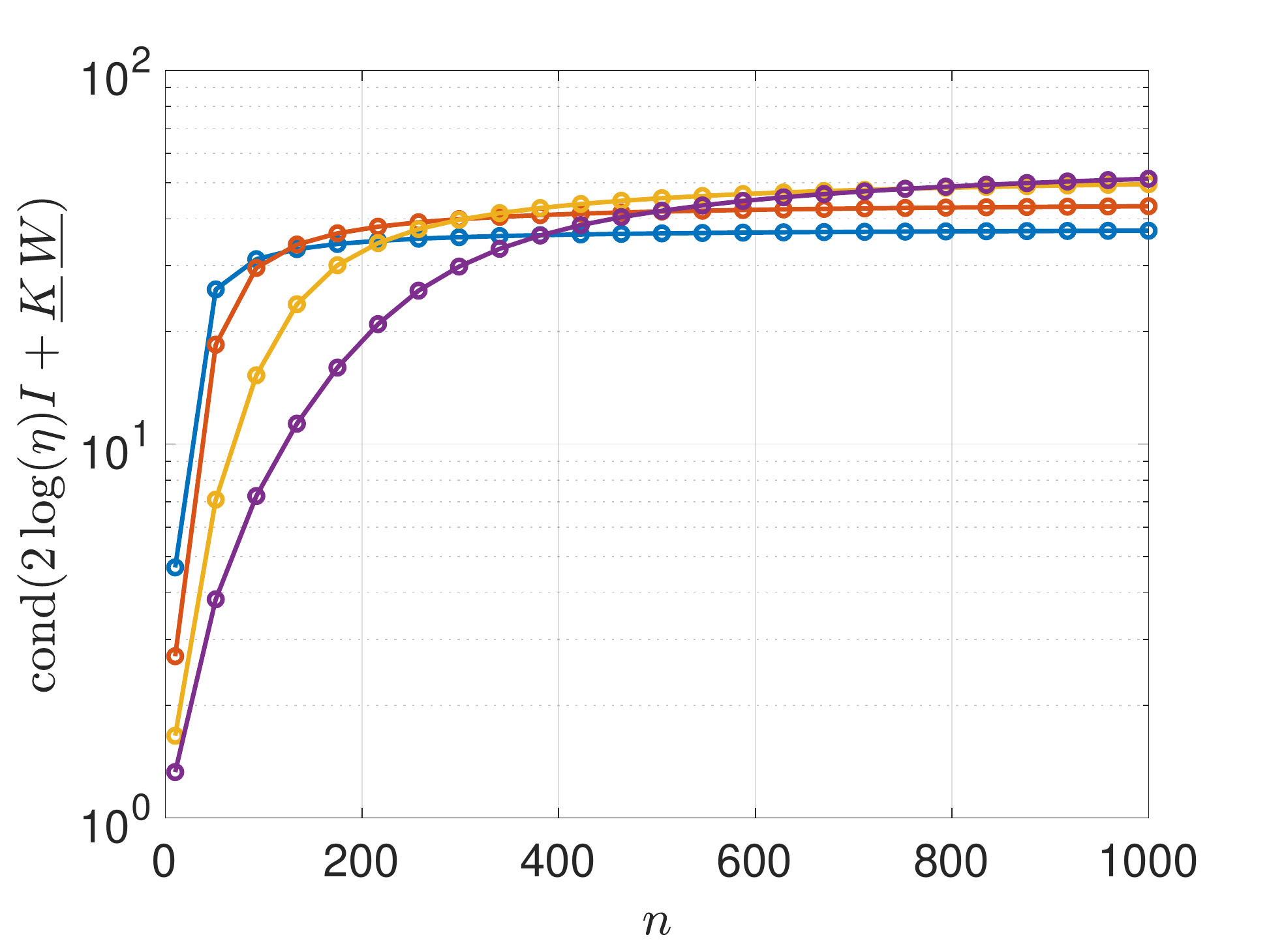}
		\caption{Regularized ($\eta=1.5$).}
		\label{cond_num_reg}
	\end{subfigure}
	\caption{The condition numbers associated with the discretized versions of the unregularized ($\eta = 1$) and regularized ($\eta = 1.5$) slender body models for calculating the force on a thin ring. Note the change in scale between the two figures.}
	\label{cond_num_ring}
\end{figure}

\subsubsection{Prolate spheroid with artificial fluid velocity field}\label{subsec:prolate}
We next use \eqref{FandT} to compute the drag force for a stationary prolate spheroid immersed in an artificial fluid velocity field. The particle centerline is aligned in the $z$-direction, parameterized by $\X(s) = (0,0,s)^T$, $s\in[-1,1]$. The fluid velocity field $\bu(s) = (u(s),0,0)^T$ is designed such that $\bff(s) = (f^x(s),0,0)^T$ is a known analytic function. We choose this function to be a Gaussian $f^x(s) = \exp\left(-\frac{s^2}{\epsilon^2}\right)$ such that the force decays to zero at the fiber endpoints and use high order Gauss-Lobatto quadrature for the discretization of the integral operator. Denote the set of $n$ quadrature nodes by $\{s_i\}^{n}_{i=1}$. Inserting the above expression for $f^x(s)$ into our model \eqref{SB_straight}, the fluid velocity at $s_i$ is found by solving the integral 
\begin{equation}\label{u_sph}
	u(s_i) = \frac{-1}{8\pi}\left(2\log(\eta)\,\exp\left(-\frac{s_i^2}{\epsilon^2}\right)+ \int_{-1}^{1} {\frac {{\epsilon}^{2} r \left( s_i \right)^{2}+\frac{1}{2}{\epsilon}^{2} r \left( s' \right)  ^{2}+ \left( s_i-s' \right) ^{2}}{ \left( {\epsilon}^{2} r \left( s_i\right)^{2}+ \left( s_i-s' \right) ^{2} \right) ^{3/2}}} \exp\left(-\frac{s'^2}{\epsilon^2}\right) ds'\right)
\end{equation}
where the ellipsoidal radius function is given by equation \eqref{prolate}. We also take the viscosity $\mu=1$. To solve for $u(s_i)$ for $i=1,...,n$, the integral in equation \eqref{u_sph} is evaluated to machine precision using MATLAB's built-in \texttt{integral} function, which uses adaptive quadrature. For this fluid velocity field, the total force $\bm{F} = (F,0,0)^T$ on the ellipsoid is found by 
\begin{equation}\label{F_sph}
	F = \int_{-1}^{1}\exp\left(-\frac{s^2}{\epsilon^2}\right)ds = \sqrt{\pi}\,\epsilon\,\mathrm{erf}(\frac{1}{\epsilon}).
\end{equation}

We compute numerical approximations to $F$ using equation \eqref{FandT} for four choices of $\epsilon$. We initially set $\eta=1$ and compute these numerical approximations for the non-regularized, first-kind equation. The errors are presented in Figure \ref{convergence_error_sph_unreg}. We see that the error converges spectrally up to a certain point where the method begins to diverge due to numerical instabilities and poor conditioning of the discrete integral operator, which is plotted in Figure \ref{cond_num_sph_unreg}. \\

However, by choosing $\eta>1$, we can amend the condition number and therefore improve the accuracy of the numerical solution. In Figure \ref{sph_error_reg}, we fix $\epsilon=0.025$ and calculate the errors for four choices of $\eta$. We see from Figure \ref{convergence_error_sph_reg} that the error converges spectrally to machine precision for all such choices of $\eta$. This level of accuracy was unattainable for the non-regularized problem. Furthermore, we observe from Figure \ref{cond_num_sph_reg} that the condition number of the discrete integral operator is bounded by a value that becomes smaller for larger $\eta$.

\begin{figure}[!ht]
	\centering
	\begin{subfigure}{.45\textwidth}
		\includegraphics[width=\linewidth]{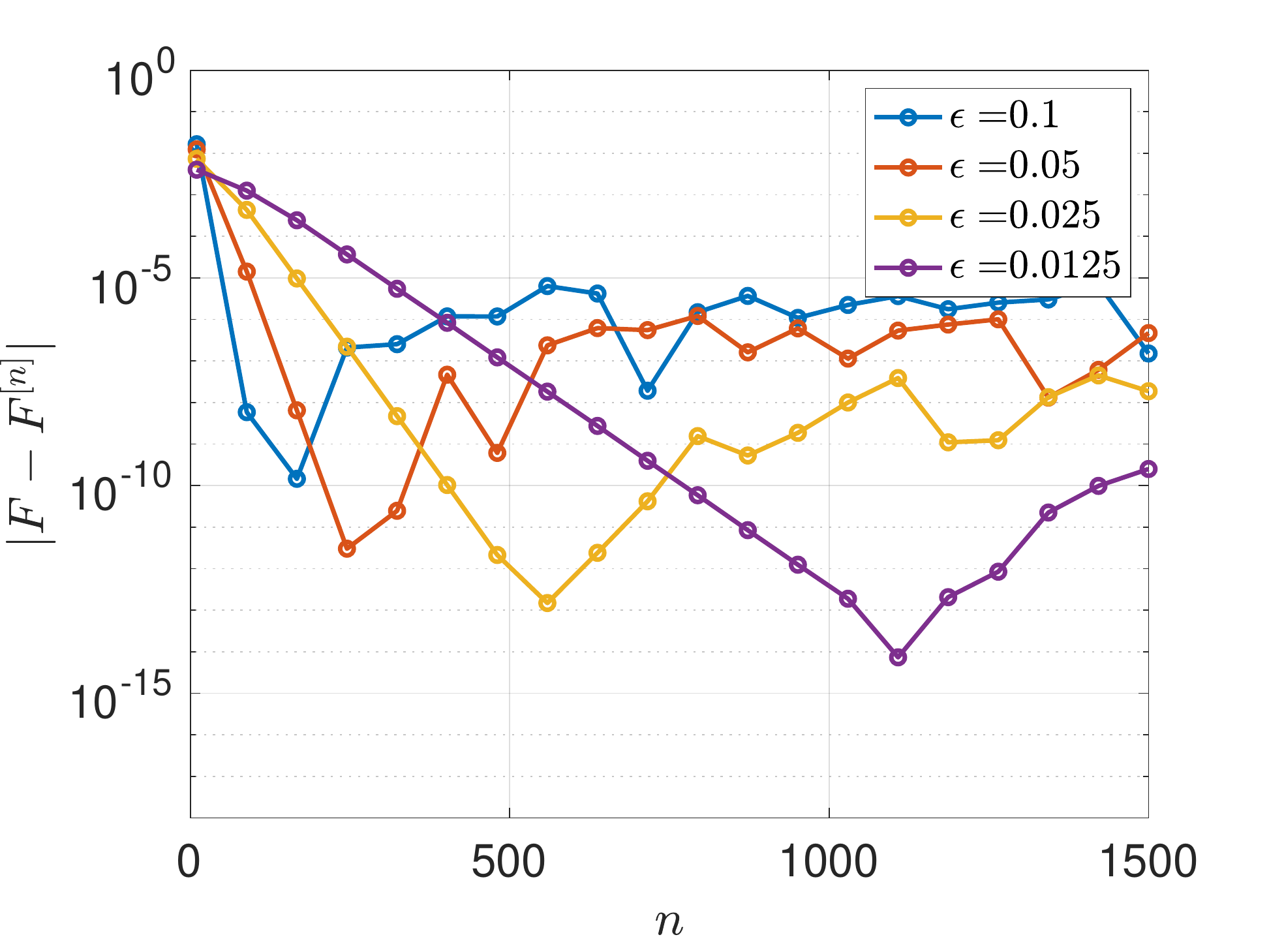}
		\caption{}
		\label{convergence_error_sph_unreg}
	\end{subfigure}
	\begin{subfigure}{.45\textwidth}
		\includegraphics[width=\linewidth]{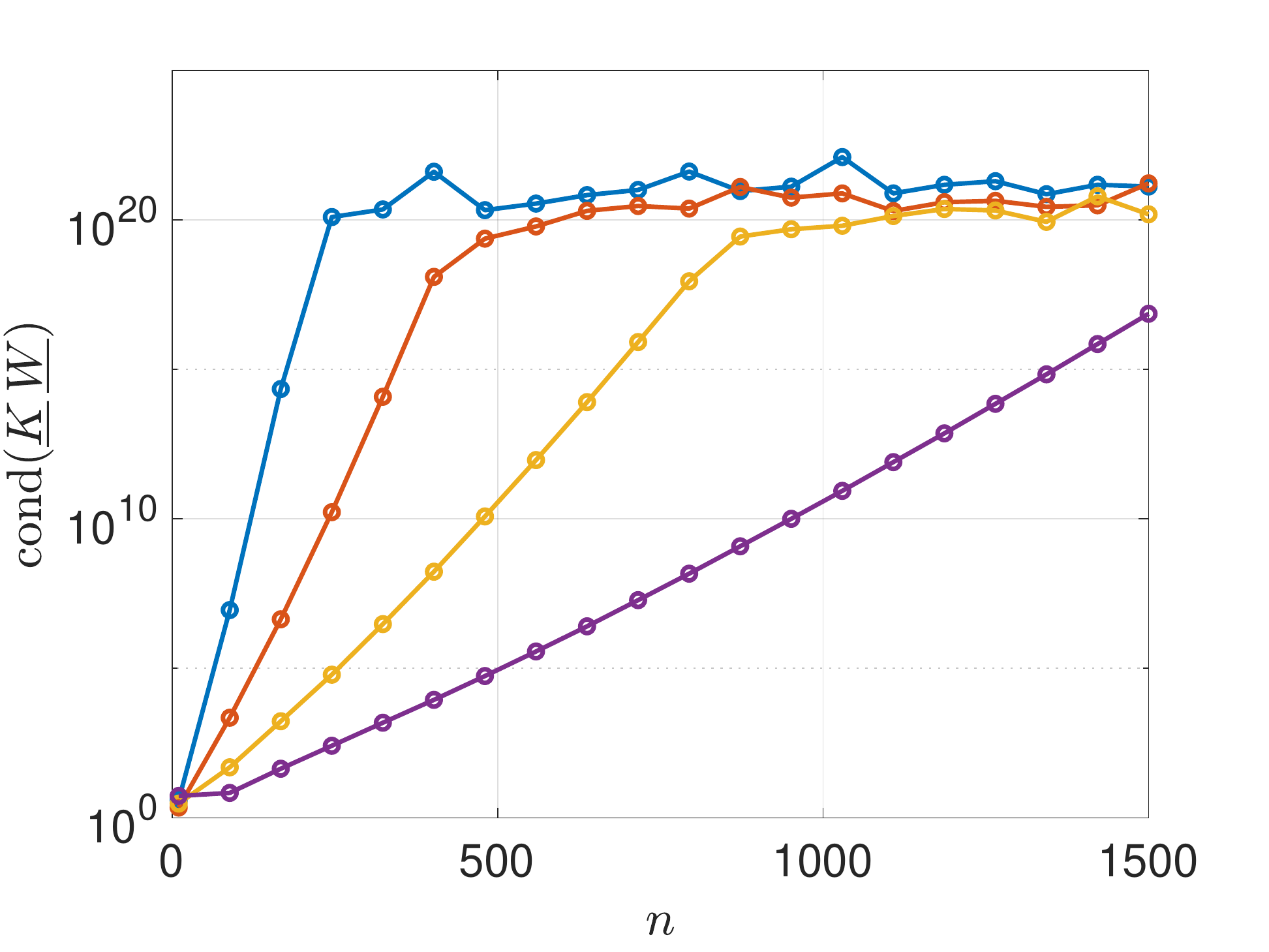}
		\caption{}
		\label{cond_num_sph_unreg}
	\end{subfigure}
	\caption{The errors (a) and condition numbers (b) associated with the unregularized ($\eta = 1$) numerical method for the calculation of the force on a prolate spheroid for different values of $\epsilon$. }
	\label{sph_error_unreg}
\end{figure}

\begin{figure}[!ht]
	\centering
	\begin{subfigure}{.45\textwidth}
		\includegraphics[width=\linewidth]{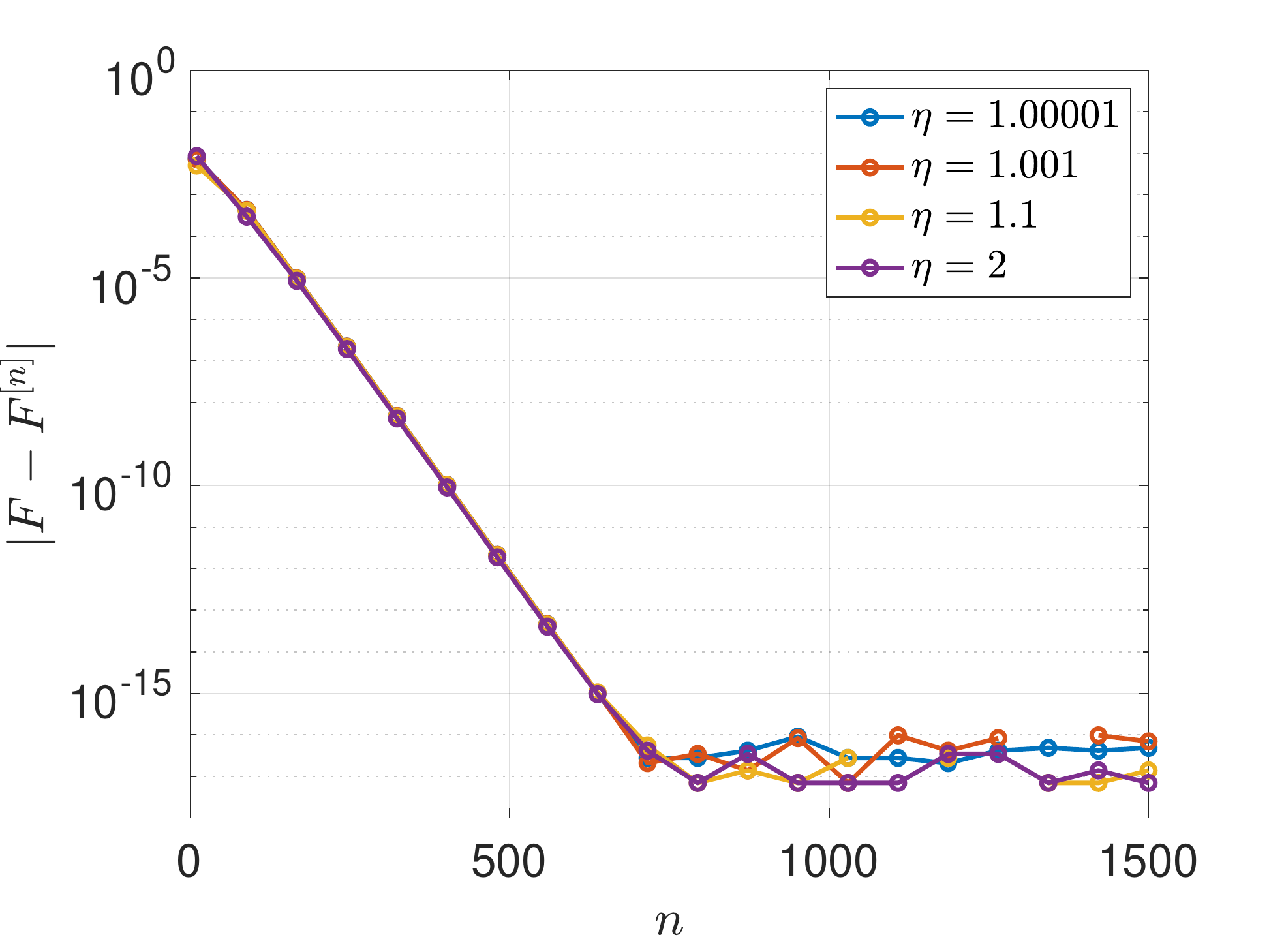}
		\caption{}
		\label{convergence_error_sph_reg}
	\end{subfigure}
	\begin{subfigure}{.45\textwidth}
		\includegraphics[width=\linewidth]{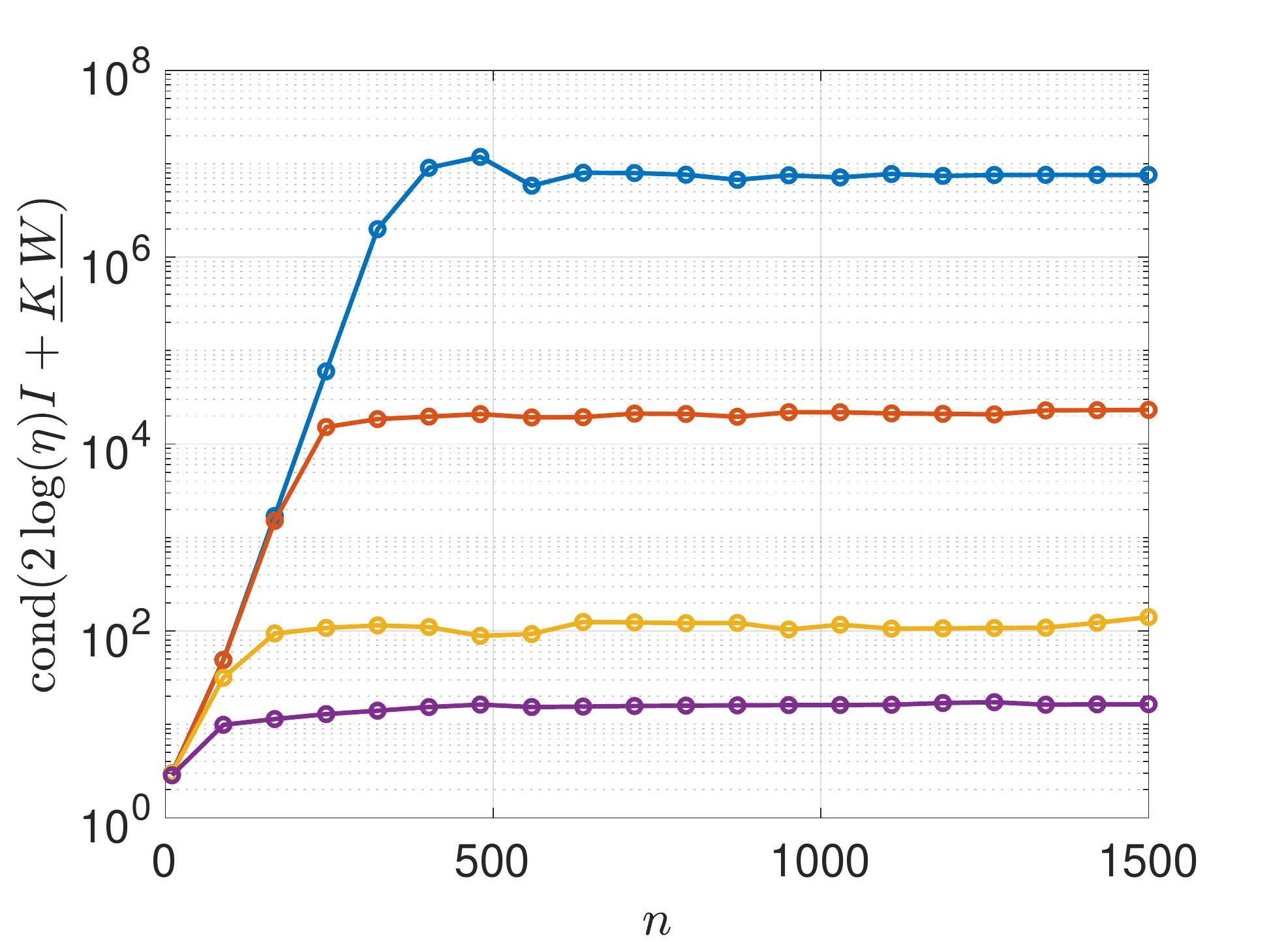}
		\caption{}
		\label{cond_num_sph_reg}
	\end{subfigure}
	\caption{The errors (a) and condition numbers (b) associated with the regularized numerical method for the calculation of the force on a prolate spheroid for $\epsilon = 0.025$. Similar results are observed for other values of $\epsilon$. }
	\label{sph_error_reg}
\end{figure}

\subsection{Spectrum of the slender body operator in different geometries}\label{sec:deigs}
One important unresolved question about the slender body model \eqref{SB_new3} is the effect of different geometries, including curvature, endpoints, and non-uniform fiber radius, on the spectrum of the integral operator. Although we cannot analytically determine the spectrum of the continuous operator in general, we \emph{can} determine the eigenvalues of the discrete operator $(2\log(\eta)I+\underline{K}\,\underline{W})$ \eqref{yeq}. We consider first the unregularized version $\eta=1$, recalling that in the straight-but-periodic geometry of Section \ref{spec_sec}, the continuous operator was provably negative definite. Ideally we would like to see evidence that this negative definiteness persists in general geometries, as this would be the physically correct behavior and also would agree with the underlying slender body PDE operator \eqref{SB_PDE}. \\

%

We begin by calculating the eigenvalues $\{\lambda_i\}_{i=1}^{3n}$ of $\underline{K}\,\underline{W}$ for the thin ring. Letting $\lambda_{\max} = \max_{i}(\lambda_i)$, in Figure \ref{deigs_ring} we plot $\lambda_{\max}$ versus $n$ for five different values of $\epsilon$. Note that for very large $n$ relative to $\epsilon^{-1}$ (roughly $n=O(\epsilon^{-2})$), we begin to see numerical error resulting in very small positive eigenvalues of $\underline{K}\,\underline{W}$ (denoted by red markers). However, the magnitude of these positive eigenvalues are on the order of machine precision and therefore can be attributed to round-off errors. \\

We next consider the effects of endpoints and a non-uniform radius by calculating the eigenvalues of $\underline{K}\,\underline{W}$ for a slender prolate spheroid \eqref{prolate}, keeping in mind the above level of numerical error. In Figure \ref{deigs_sph} we again plot $\lambda_{\max}$ versus $n$ for four different values of $\epsilon$. Again for $n=O(\epsilon^{-2})$ we begin to see small positive eigenvalues which are significantly larger than for the thin ring (around $O(10^{-10})$). However, the magnitude of the positive eigenvalues is still very small and bounded as $n$ increases. It is not clear whether this is a numerical artifact or an actual eigenvalue crossing 0 for the continuous operator. At any rate, the non-regularized operator would never actually be used for simulations with such large $n$ because the condition number of $\underline{K}\,\underline{W}$ is prohibitive (see Figure \ref{cond_num_sph_unreg}). It appears that a very reasonable choice of regularization parameter $\eta$ will ensure that none of these near-zero eigenvalues actually cross zero. \\

\begin{figure}[!ht]
	\centering
	\begin{subfigure}{.45\textwidth}
		\includegraphics[width=\linewidth]{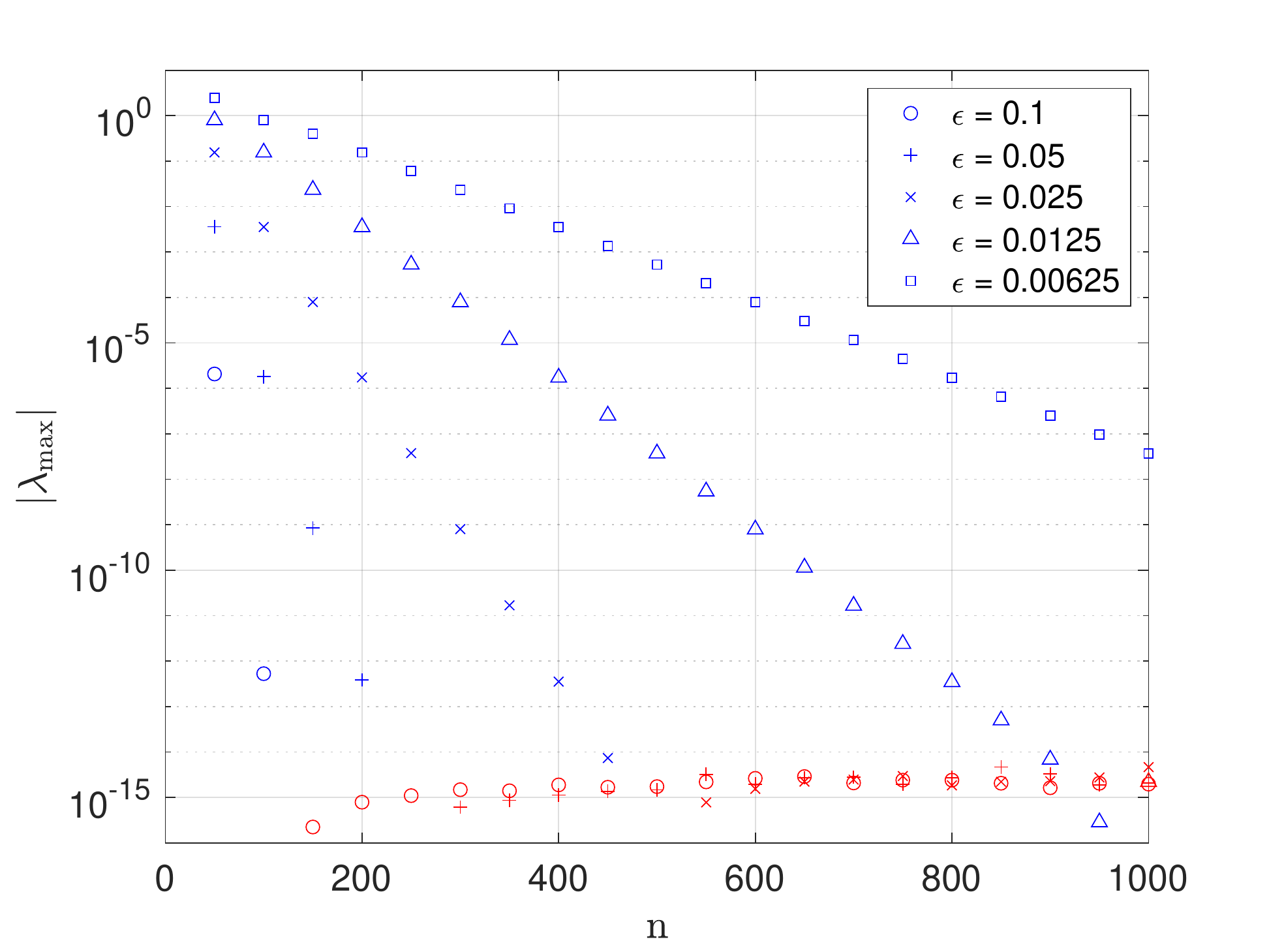}
		\caption{Thin ring}
		\label{deigs_ring}
	\end{subfigure}
	\begin{subfigure}{.45\textwidth}
		\includegraphics[width=\linewidth]{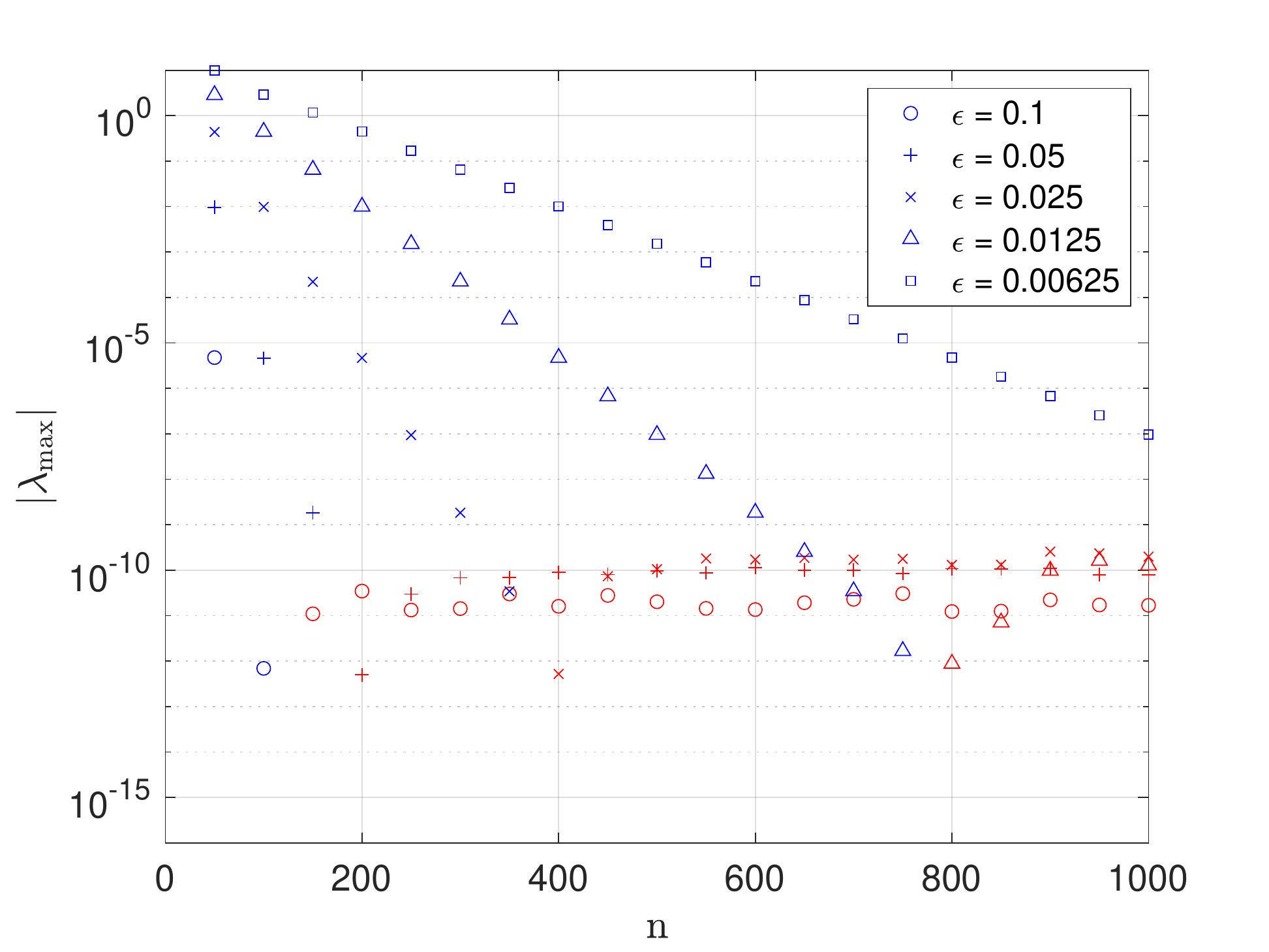}
		\caption{Spheroid}
		\label{deigs_sph}
	\end{subfigure}
	\caption{Magnitude of the maximum eigenvalue of the non-regularized discrete slender body operator $\underline{K}\,\underline{W}$. Blue markers mean that $\lambda_{\max}<0$ implying $\underline{K}\,\underline{W}$ is negative definite, while red markers mean that $\lambda_{\max}>0$.}
	\label{deigs}
\end{figure}

As a final test, we calculate the spectrum of $\underline{K}\,\underline{W}$ for randomly but systematically generated curvy fibers with complicated shapes (Figure \ref{rand_fibres}). Here the magnitude of the fiber's deviation from a straight line is controlled by a small parameter $\delta\ge 0$. The fiber shapes are generated by interpolating $m$ points $(x_i,y_i,z_i)\in \RR^3$, $i=1,...,m$, with cubic splines. Here $z_i=(i-1)\frac{2L}{m}$ while $x_i,y_i\in[-\delta,\delta]$ are given by a random number generator and are of size at most $\delta$. Setting $\delta = 0$ corresponds to a straight fiber. Examples of the fiber centerline for $m=10$ and four different values of $\delta$ are given in Figure \ref{rand_fibres}. \\
\begin{figure}[!ht]
	\centering
	\begin{subfigure}{.225\textwidth}
		\includegraphics[width=\linewidth]{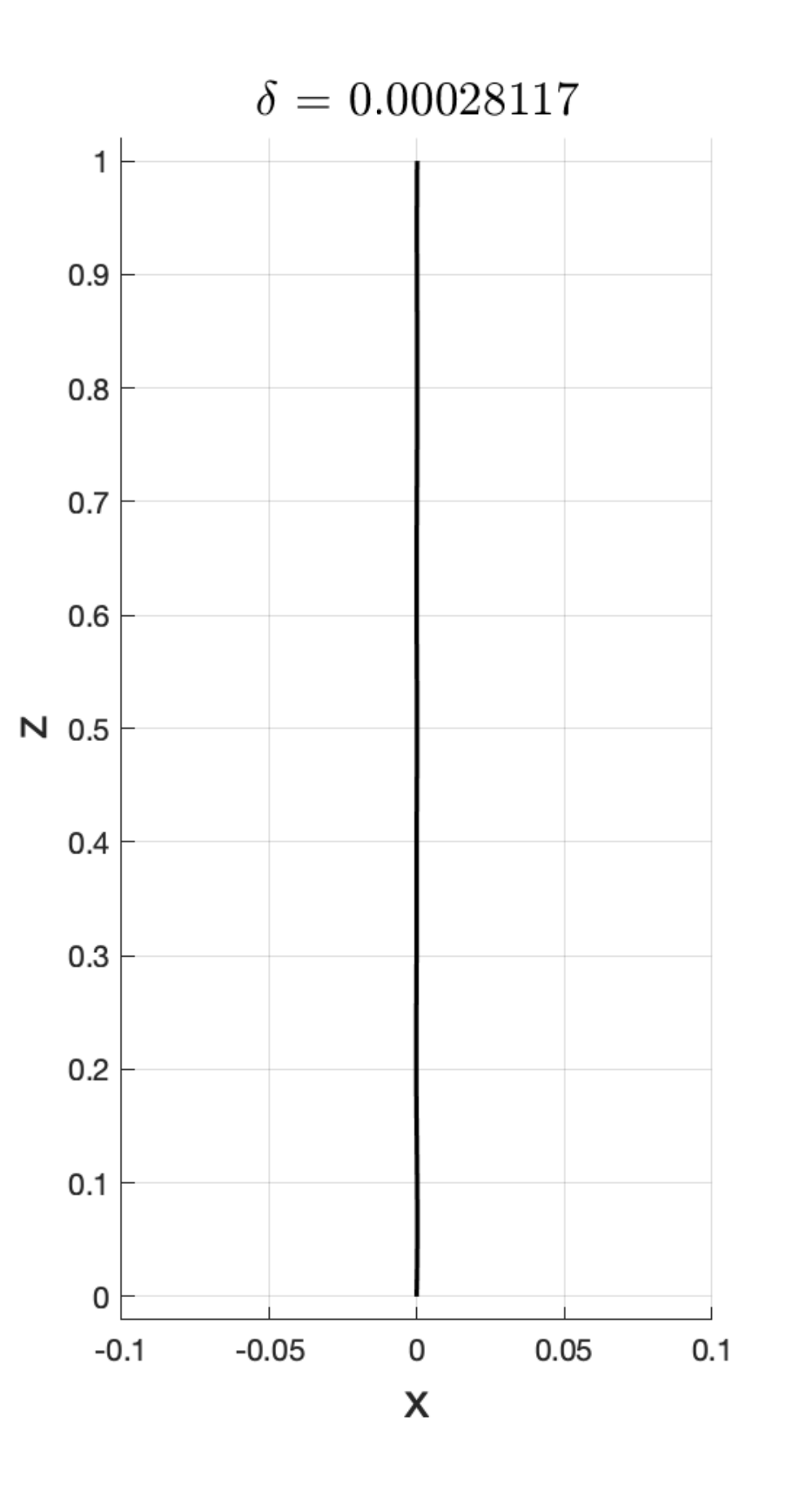}
		\caption{}
		\label{fig:squiggle4}
	\end{subfigure}
	\begin{subfigure}{.225\textwidth}
		\includegraphics[width=\linewidth]{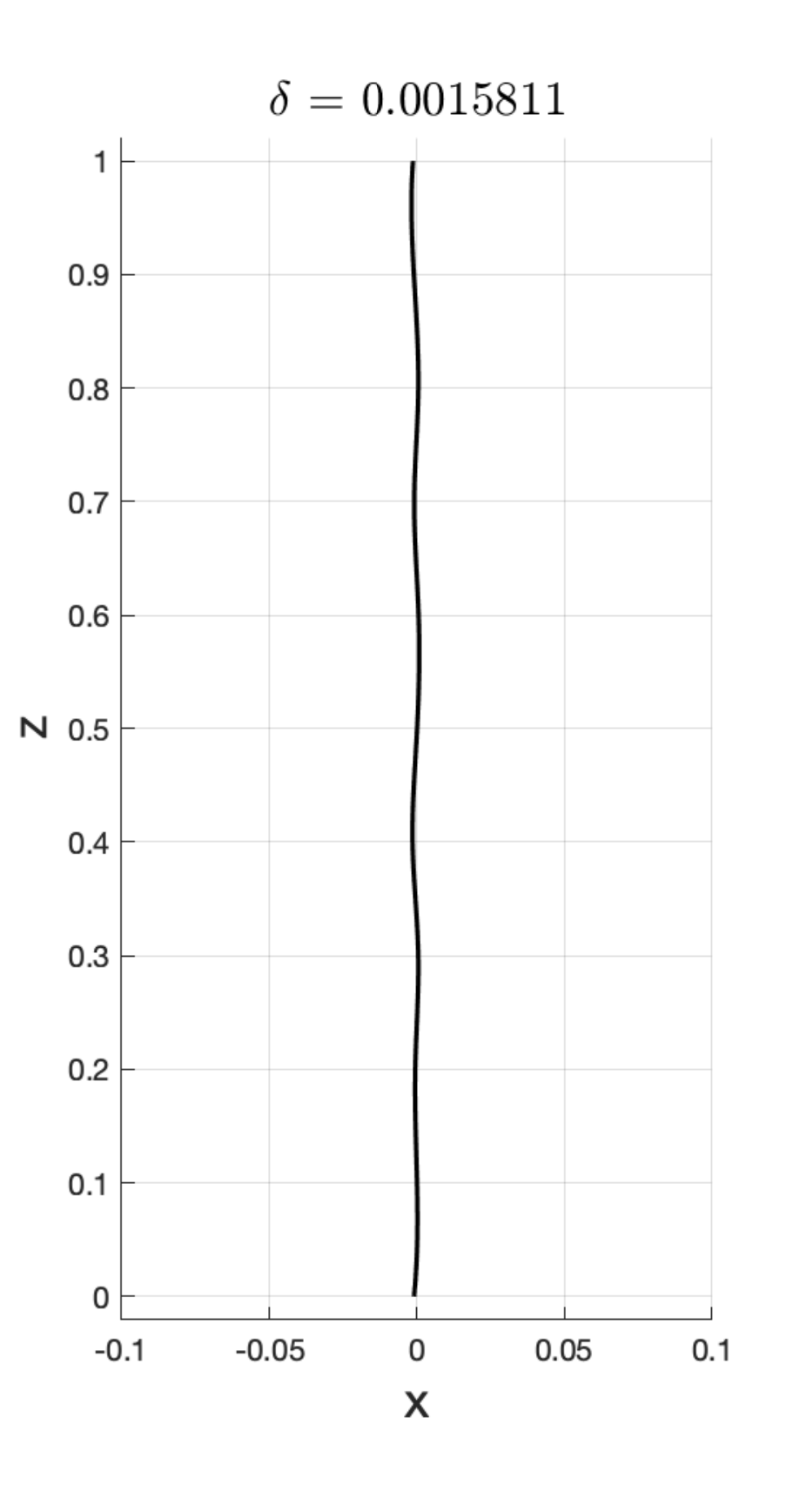}
		\caption{}
		\label{fig:squiggle6}
	\end{subfigure}
	\begin{subfigure}{.225\textwidth}
		\includegraphics[width=\linewidth]{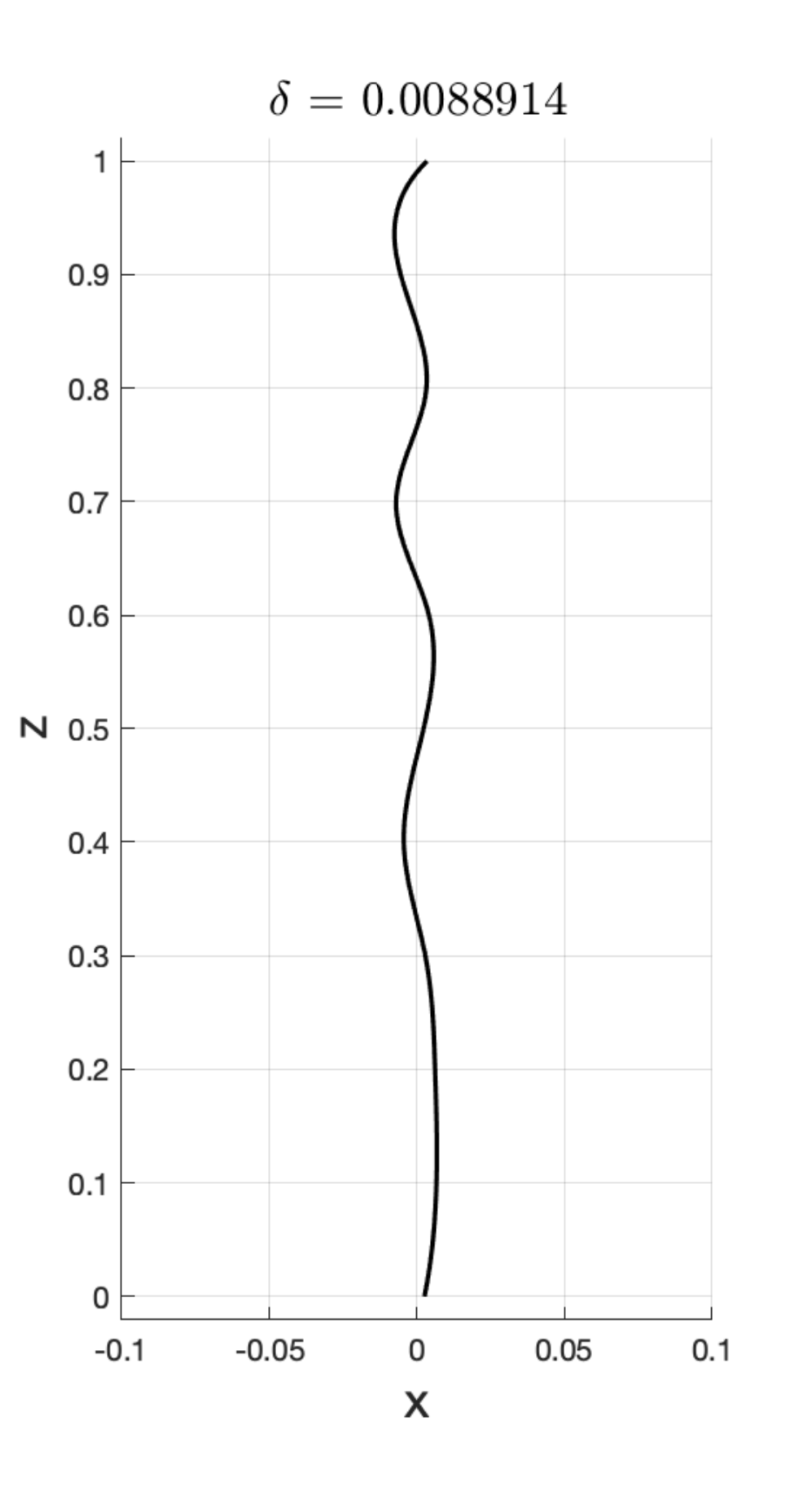}
		\caption{}
		\label{fig:squiggle8}
	\end{subfigure}
	\begin{subfigure}{.225\textwidth}
		\includegraphics[width=\linewidth]{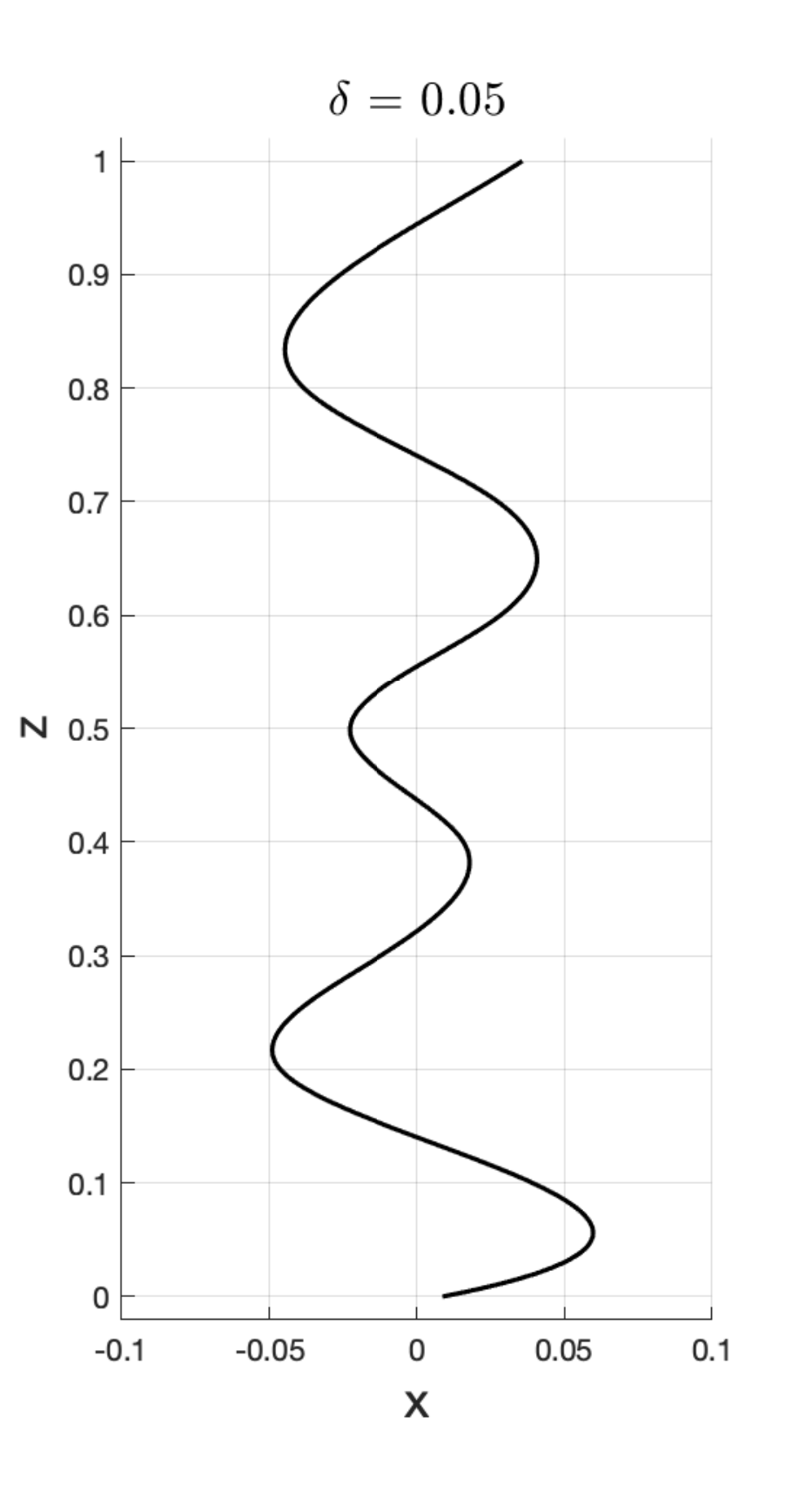}
		\caption{}
		\label{fig:squiggle10}
	\end{subfigure}
	\caption{The centerlines of four curved fiber shapes.}\label{rand_fibres}
\end{figure}

We fix $\epsilon=0.1$ and use the spheroidal radius function \eqref{prolate}. Taking $m=10$, we generate 6 different curvy fibers for different magnitudes $\delta\in [0,\frac{1}{10}]$. For each fiber we compute the spectrum $\{\lambda^\delta_i\}_{i=1}^n$ of its corresponding (non-regularized) integral operator $\underline{K}\,\underline{W}$. We plot the most positive eigenvalue $\lambda^\delta_{\max} = \max_i(\lambda^\delta_i)$ for each fiber in Figure \ref{squiggle_eigs}. For each value of $\delta$ we note that there is an eigenvalue crossing zero when $n=O(\epsilon^{-2})$. As $\delta$ increases and the magnitude of the curviness of the fiber increases, we can note a slight increase in the magnitude of the largest positive eigenvalue, but $\lambda^\delta_{\max}$ is still small -- roughly $O(10^{-8})$. Again, we can be assured to have a negative spectrum bounded away from 0 by a reasonable choice of regularization $\eta>1$. This effect is displayed in Figure \ref{squiggle_eigs2}, which shows the maximum eigenvalue $\lambda^{\delta,\eta}_{\mathrm{max}}$ of the now \textit{regularized} discrete integral operator $(2\log(\eta)I+\underline{K}\,\underline{W})$ for a fixed value of $\epsilon$ and $\delta$ and varying values of $\eta$. We see here that for all choices of $\eta>1$ in this range, the spectrum of $(2\log(\eta)I+\underline{K}\,\underline{W})$ remains negative definite. 

\begin{figure}
	\centering
		\begin{subfigure}{.45\textwidth}
		\includegraphics[width=\linewidth]{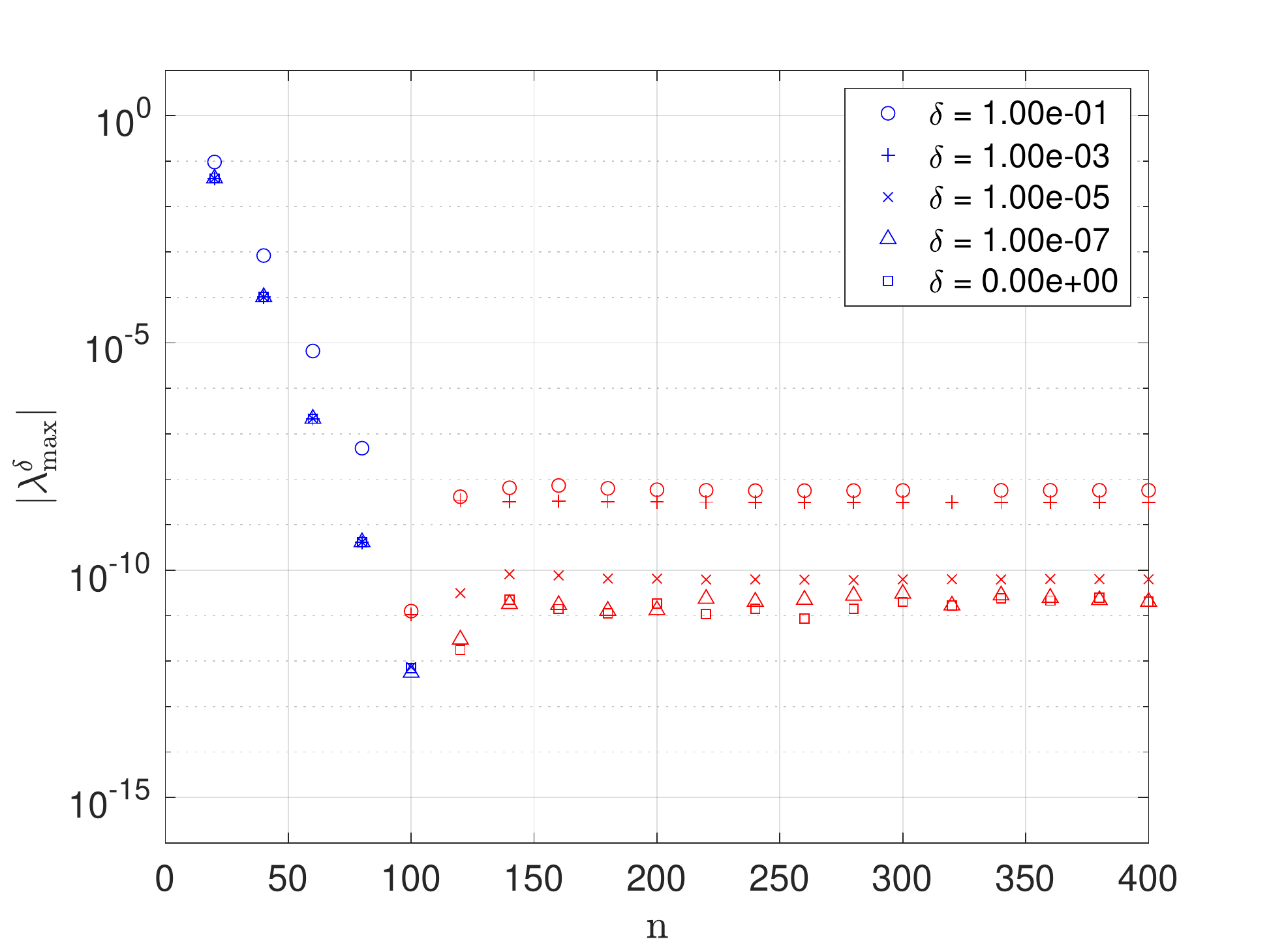}
		\caption{$\lambda^\delta_{\max}$ for different values of $\delta$ and $\eta=1$ . }
		\label{squiggle_eigs}
	\end{subfigure}
	\begin{subfigure}{.45\textwidth}
	\includegraphics[width=\linewidth]{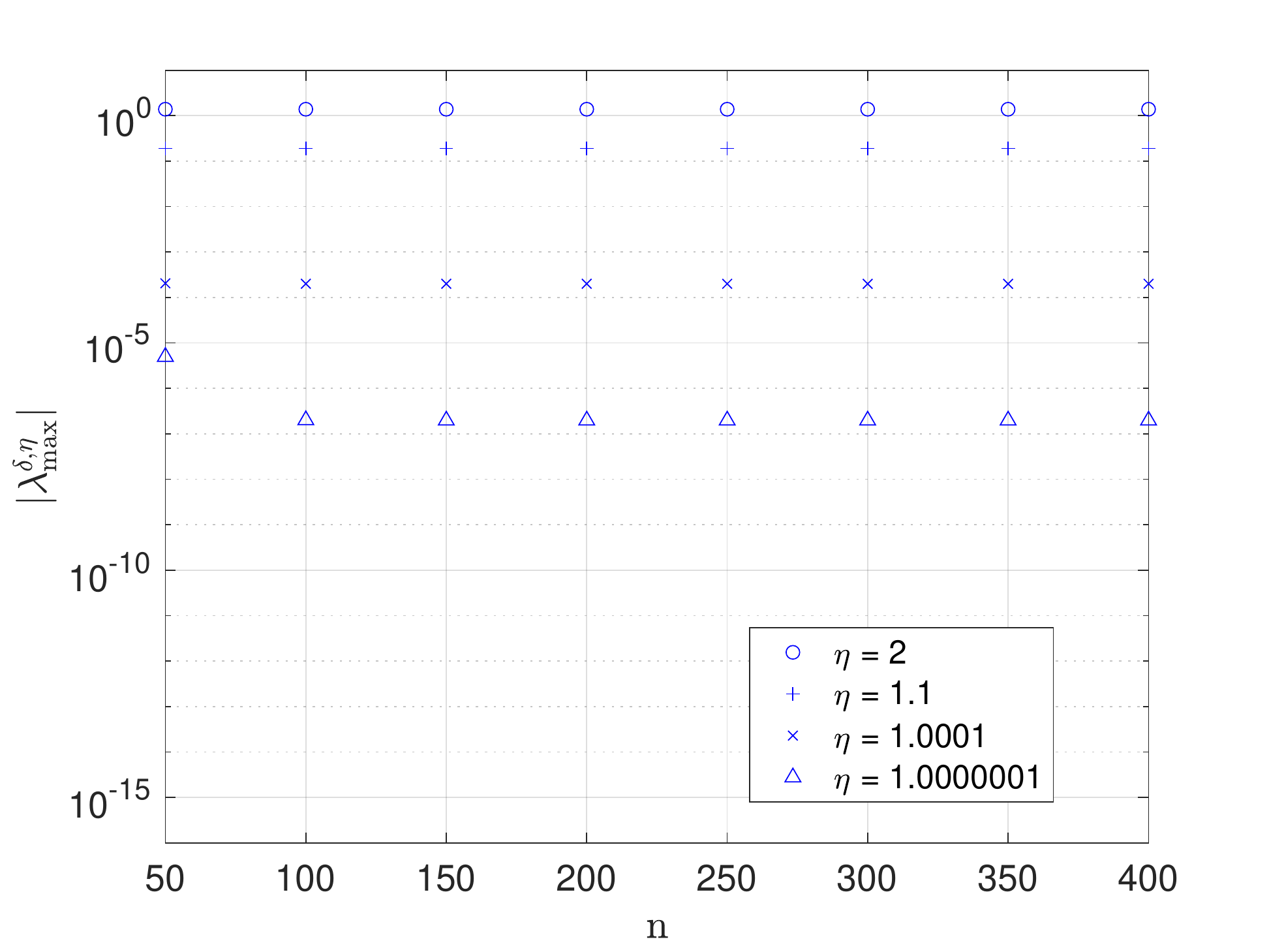}
	\caption{$\lambda^{\delta,\eta}_{\mathrm{max}}$ for different values of $\eta$ with $\epsilon=0.1$ and $\delta = 0.001$. }
	\label{squiggle_eigs2}
\end{subfigure}
	\caption{The maximum eigenvalues for the unregularized (a) and regularized (b) discrete integral operators for the curved fibers. The color blue denotes a negative maximum eigenvalue and red denotes a positive maximum eigenvalue. }
	\label{}
\end{figure}

\section{Dynamics of curved rigid fibers}\label{sec:dynamics}
We next use the slender body model \eqref{SB_new3} and the discretization procedure of Section \ref{sec:disc} to simulate the dynamics of curved rigid fibers in Stokes flow. After outlining the dynamical equations, we validate the model against known dynamics for a slender prolate spheroid. Finally, we compare the rotational dynamics of randomly curved fibers as in Figure \ref{rand_fibres} to straight fibers.

\subsection{Dynamical equations}
The dynamics of the slender body are governed by the rigid body equations. The angular momentum $\bm{m}$ of a rigid particle with torque $\bm{T}(t)$ is found by solving 
\begin{equation}\label{eq:rotation}
\dot{\bm{m}}=\bm{m}\times\bom+\bm{T},
\end{equation}
where $\bom = J^{-1} \bm{m}$ for moment of inertia tensor $J$. Each of these quantities are given in the fiber frame of reference. The fiber orientation (with respect to a fixed inertial frame of reference) is specified using Euler parameters $q\in\mathbb{R}^4$ which satisfy the constraint $||q||_2=1$ and are determined by solving the ODE
\begin{equation}\label{qode}
\dot{q} = \frac{1}{2} q\cdot w,
\end{equation}
where $w=(0,\bom\trans)\trans\in\mathbb{R}^4$ and $\cdot$ denotes the Hamilton product of two quaternions \cite{goldstein2002classical}. A vector $\bx_p$ in the fiber reference frame can be rotated to a vector $\bx_T$ in an inertial co-translating reference frame via $\bx_T = Q \bx_p$ where the rotation matrix $Q$ is the image of $q$ under the Euler-Rodriguez map. We refer the reader to \cite{goldstein2002classical} for details on quaternion algebra and rigid body mechanics. \\

The translational dynamics are given by Newton's second law
\begin{equation}\label{pode}
\dot{\bm{p}} = \bm{F},
\end{equation}
where $\bm{p} = \bm{v}/m$ is the inertial frame linear momentum for a fiber of mass $m$. The position of the fiber center of mass is found by solving
\begin{equation}\label{xode}
\dot{\bx} = \bm{v}.
\end{equation}
The ODEs \eqref{eq:rotation} - \eqref{xode} are integrated using the second order Strang splitting method of \cite{tapley2019novel}.\\

Recall the equations \eqref{FandT} for $\bm{F}^{[n]}$ and $\bm{T}^{[n]}$.
Since $\bm F^{[n]}$ and $\bm T^{[n]}$ depend linearly on the linear and angular momenta $\bp$ and $\bmm$, we may update them according to the linear equation
\begin{equation} \label{FandTmethod}
\left(\begin{array}{c}
\bm{F}^{[n]}\\\bm{T}^{[n]}\\
\end{array}\right)= A\left(\begin{array}{c}
\bp\\\bmm\\
\end{array}\right) + \bm{b},
\end{equation}
where $A$ is a negative definite dissipation matrix and $\bm b$ is due to the background fluid velocity and is independent of $\bp$ and $\bmm$. We have that
\begin{equation}\label{dissipation matrix}
A = \left(\begin{array}{cc}
\Phi\,\left(\mathbb{1}\otimes (I/m) \right), & \Phi\, \left(-\underline{X}(\mathbb{1}\otimes J^{-1})\right)\\
\Psi \left(\mathbb{1}\otimes (I/m)\right), & \Psi \left(-\underline{X}(\mathbb{1}\otimes J^{-1})\right)
\end{array}\right)\quad\mathrm{and}\quad\bm{b} = - \left(\begin{array}{c}
\Phi \underline{\bu} \\
\Psi\underline{\bu}\\
\end{array}\right),
\end{equation}
where $m$ and $J$ are the filament mass and moment of inertia tensor, respectively. We have also introduced the vector $\underline{\bu} = (\bu_0(\X(s_1))^T,...,\bu_0(\X(s_n))^T)^T$ containing the background fluid velocities at the location of the quadrature nodes along the centerline. 

\subsubsection{Overview and cost of algorithm}
The algorithm used to compute the dynamics of a slender fiber is as follows:
\begin{enumerate}
	\item Define particle geometry $\X(s)$, $\epsilon$, regularization parameter $\eta$ and discretization parameter $n$. 
	\item Choose a quadrature rule and compute the matrices $\underline{W}$ and $\underline{K}$. 
	\item Compute the matrices  $\Phi$, $\Psi$ and $A$ from equations \eqref{Phi}, \eqref{Psi} and \eqref{dissipation matrix}. 
	\item Time loop: for $t = 0,\Delta t,...,m\Delta t$
	\subitem{a)} Compute $\bm{F}^{[n]}$ and $\bm{T}^{[n]}$ using equation \eqref{FandTmethod}
	\subitem{b)} Numerically integrate the ODEs \eqref{eq:rotation} - \eqref{xode} 
\end{enumerate}
For step (2), we use the trapezoidal quadrature rule for closed fibers (i.e., a periodic integration interval) or Gauss-Lobatto quadrature rule for fibers with open ends. For step (4b), we use a splitting method \cite{tapley2019novel}. We note that for simulations where the fluid velocity field is calculated from a direct numerical simulation of the Navier-Stokes equations, the fluid field needs to be approximated onto the centerline of the particle using an interpolation method \cite{tapley2019computational}. \\

The above algorithm exploits the rigidity of the fiber by using the fact that $A$, $\Phi$ and $\Psi$ are constant in time and therefore can be computed outside of the time loop. The calculation of these matrices, which involves solving a linear system, is the most costly operation in the algorithm but only needs to be done once. If, for example, Gaussian elimination is used, this step has complexity of $O(n^3)$. Within the time loop, however, the most costly operation is the calculation of $\bm{F}^{[n]}$ and $\bm{T}^{[n]}$, which involves only $3\times 3n$ by $3n\times 1$ matrix-vector products, which has $O(n)$ complexity. We assume that the cost of numerically integrating the ODEs is negligible compared to this. For a single fiber, the total complexity of the algorithm is therefore $O(n^3 + nm)$, where $m$ is the total number of time steps used in the simulation. Hence, for simulations where many time steps are needed, the algorithm scales by $O(n)$. We remark that for problems where the background flow is zero, the cost of computing $\bm{F}^{[n]}$ and $\bm{T}^{[n]}$ is independent of $n$ (after $A$ has been computed) and therefore is $O(1)$. This is relevant, for example, when simulating fibers sedimenting in a still fluid under the influence of gravity \cite{newsom1994dynamics}.  \\

\subsection{Numerical validation of model dynamics}

\subsubsection{Dissipation matrix of a prolate spheroid}
Here we compare our model and numerical method with accurate closed form expressions for the force and torque given by Brenner \cite{brenner1964stokes} and Jeffery \cite{jeffery1922motion}. These expressions are valid for an ellipsoid when the fluid Jacobian is approximately constant throughout the volume of the particle. When the flow is linear, these terms are essentially exact and therefore serve as a good reference model against which to validate our model.\\

The purpose of this numerical experiment is therefore twofold. Firstly, we aim to show that our model converges to the reference model as $\epsilon\rightarrow 0$. This is primarily to validate the accuracy of the model. However, the numerical approximation of the force and torques also introduces a numerical error that is related to the discretization parameter $n$. Clearly, taking $n$ too small means that we will not exploit the accuracy of the model to its entirety. On the other hand, it is unwise to take $n$ as large as possible as this will incur unnecessary computational costs that go to minimizing numerical error beyond the accuracy of the model. So the second question we address here is what is an ideal choice of discretization parameter to use such that the numerical error is roughly the same as the modeling error. \\ 

Using $\eta = 1 + \epsilon^2$, the dissipation matrix for our slender body model $A$ is numerically approximated by equation \eqref{dissipation matrix}. The reference dissipation matrix $A_{sph}$ is found using the closed form expressions from Jeffery and Brenner, which are given in Appendix \ref{app:dissmat}. Denote the six eigenvalues of $A$ and $A_{sph}$, by $\lambda_i$ and $\lambda_i^{sph}$, respectively. Note that due to symmetry of the spheroid, $\lambda_1=\lambda_2$ and $\lambda_4=\lambda_5$ and similarly for the eigenvalues of $A_{sph}$. Furthermore, the slender body model is essentially a one dimensional filament and therefore $\lambda_6=0$ meaning that spinning motion about the centerline doesn't dissipate. This is in contrast to the Jeffrey term, which does dissipate spinning motion. We remark that this phenomenon only occurs in the case where the centerline is perfectly straight. Hence for curved fiber geometries where the application of the slender body is most useful, this nonphysical phenomenon is not observed. Note that for this geometry the dissipation matrices are diagonal and therefore the eigenvalues are directly proportional to the calculation of $\bm F^{[n]}$ and $\bm T^{[n]}$ in zero background flow. \\

The eigenvalues of $A$ are calculated using equation \eqref{dissipation matrix} after discretizing equation \eqref{intop2} on the Gauss-Lobatto nodes. The values $|\lambda_i-\lambda_i^{sph}|$ for $i=1,3,4$ are plotted in Figure \ref{fig:dissipationeigs} as a function of the discretization parameter $n$. We see that $\lambda_i$ converges exponentially to a point near $\lambda^{sph}_i$, which is likely due to the slender body modelling error. As $\epsilon$ decreases, we make two observations. First, for large $n$ the rate at which $\lambda_i$ converges to $\lambda_i^{sph}$ is approximately $-\epsilon^2\eta^2 \log (\epsilon\eta)$, as seen by the horizontal dash-dot lines. Second, as $\epsilon$ decreases, the convergence rate slows down and one must use a larger value of $n$ to reach the most accurate solution. This means that one must pay careful attention to the choice of $n$ when taking $\epsilon$ to be very small. In fact, we observe empirically that the convergence rate is approximately bounded by $e^{-4\epsilon n}$. Motivated by this, we will take $n$ in future experiments to be approximately the intersection of these two lines, that is
\begin{equation}\label{n}
	n  \approx -\frac{\log(-\epsilon^2 \log(\epsilon))}{4\epsilon}.
\end{equation}

\begin{figure}[h]
	\centering
		\begin{subfigure}{.45\textwidth}
		\includegraphics[width=\linewidth]{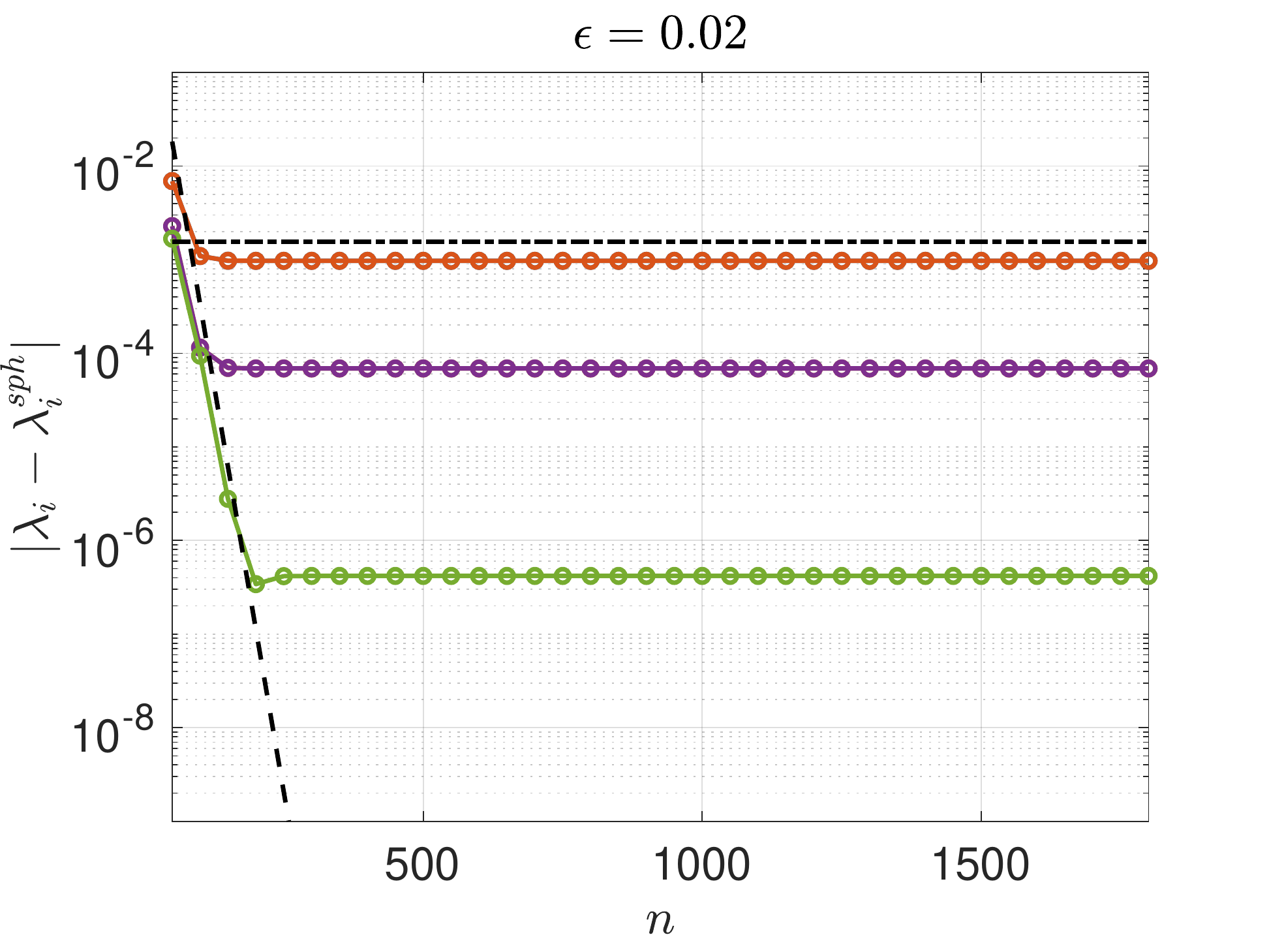}
		\caption{}
		\label{dissa}
	\end{subfigure}
		\begin{subfigure}{.45\textwidth}
	\includegraphics[width=\linewidth]{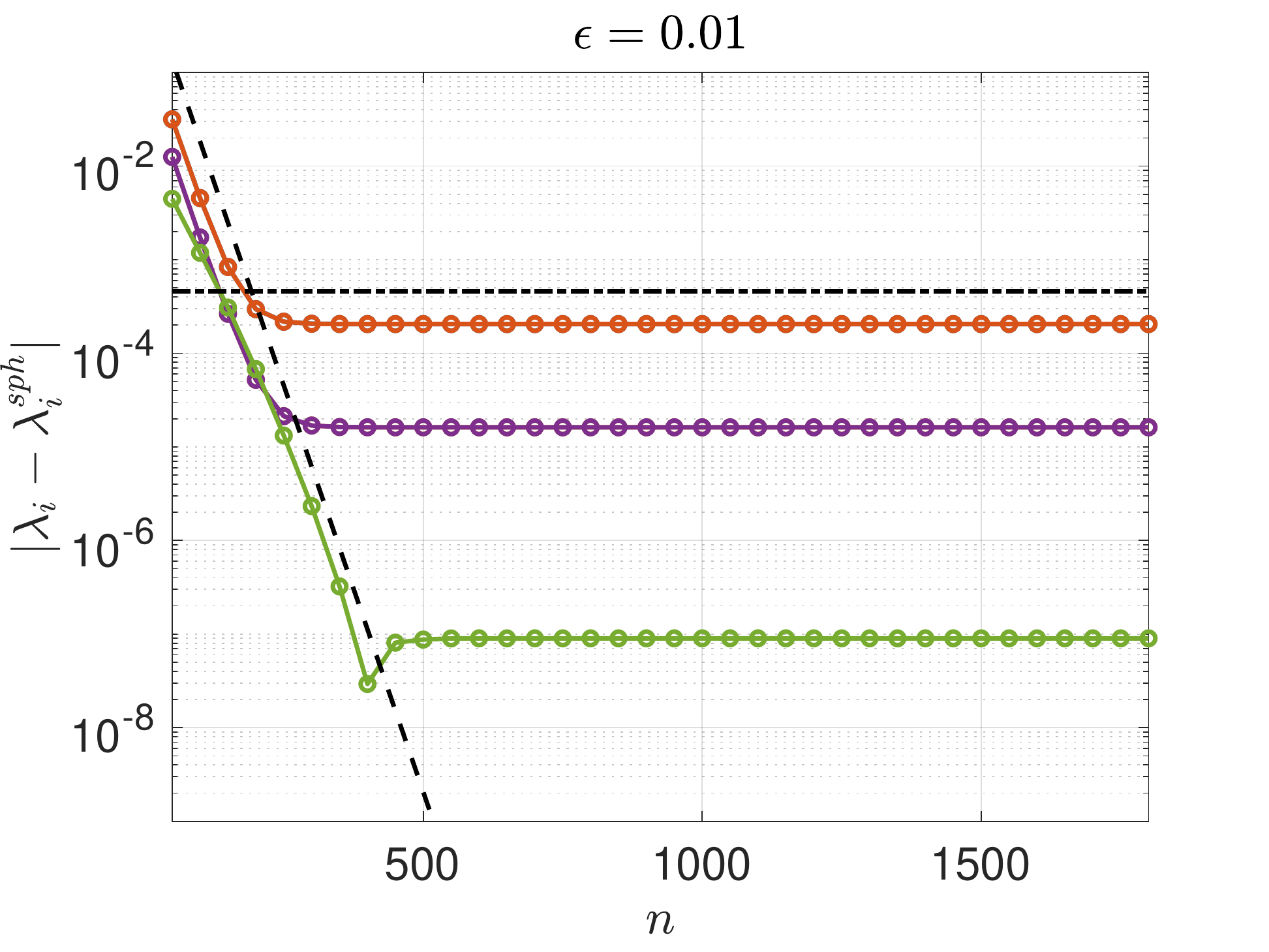}
	\caption{}
	\label{dissb}
\end{subfigure}

		\begin{subfigure}{.45\textwidth}
	\includegraphics[width=\linewidth]{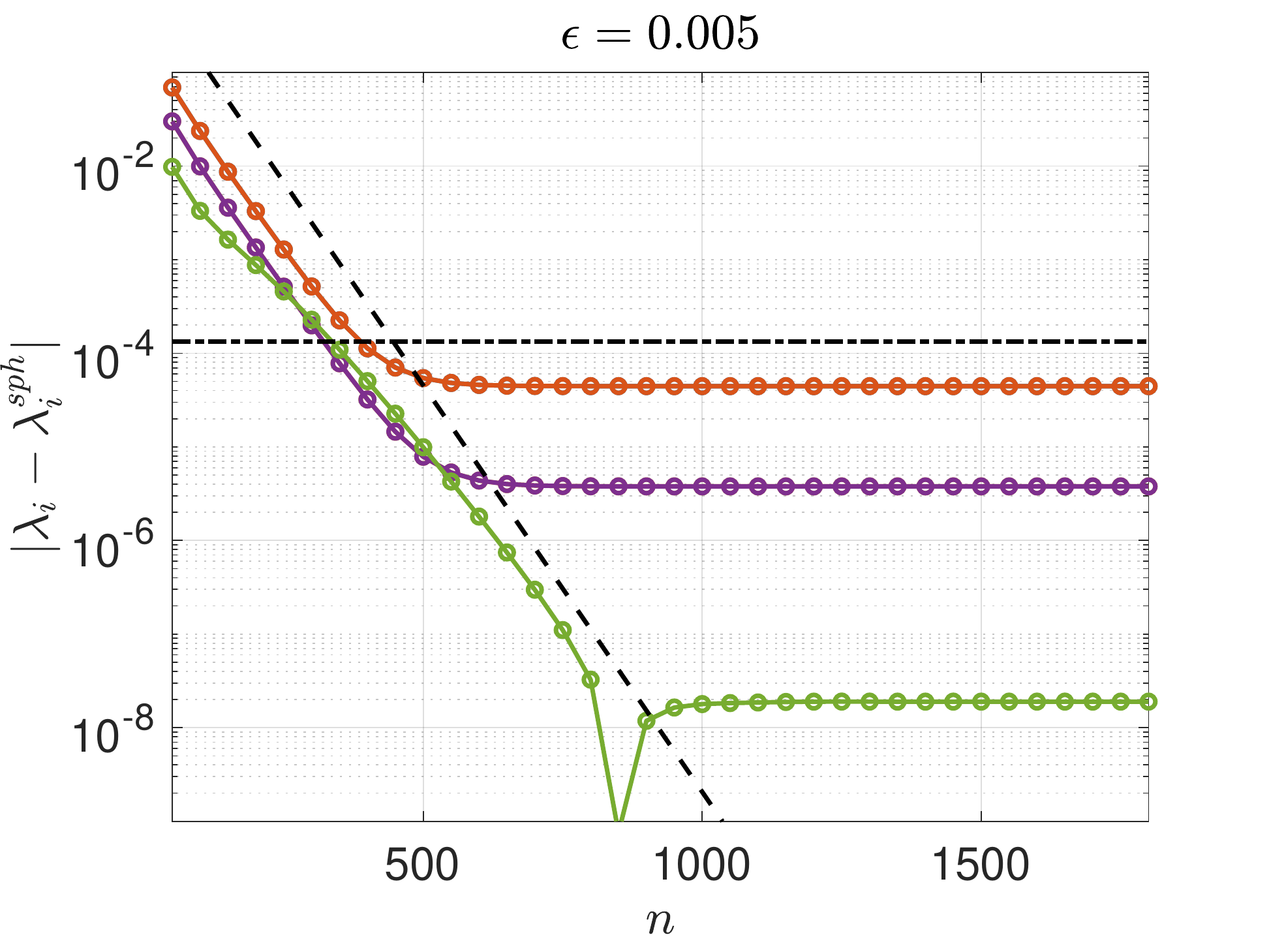}
	\caption{}
	\label{dissc}
\end{subfigure}
		\begin{subfigure}{.45\textwidth}
	\includegraphics[width=\linewidth]{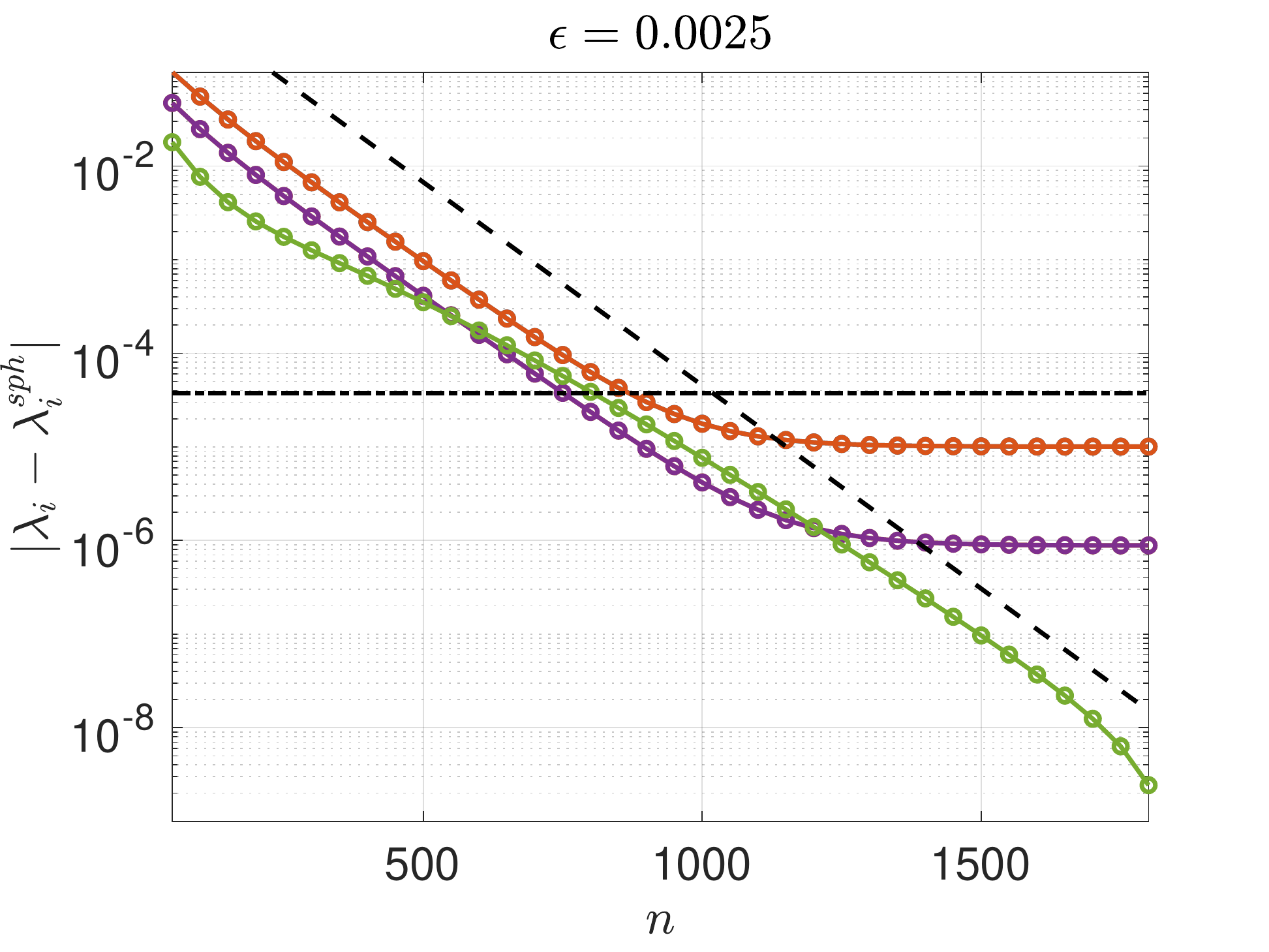}
	\caption{}
	\label{dissd}
\end{subfigure}
	\caption{The difference in the dissipation matrix eigenvalues $|\lambda_i-\lambda_i^{sph}|$, $i=1,3,4$ as a function of $n$ for three different values of $\epsilon$. The black dashed lines are $e^{-4\epsilon n}$ and the horizontal dash-dot lines are $-\epsilon^2\eta^2 \log(\epsilon\eta)$}
	\label{fig:dissipationeigs}
\end{figure}

\subsubsection{Prolate spheroids rotating in shear flow}
Now we calculate the dynamics of a prolate spheroid in shear flow $\bu = (z,0,0)^T$ using our model and compare it with that of the accurate Jeffrey model. The fiber is initially aligned at rest in the $z$-direction and its rotational dynamics are calculated by integrating equation \eqref{eq:rotation} on the interval $t\in[0,100]$ using the splitting method of \cite{tapley2019novel} with a small step size of $h=0.01$. The simulation was repeated with $h=0.05$ with no significant changes to the results and it is therefore concluded that time integration errors are negligible. We repeat the experiment for 20 values of $\epsilon$ logarithmically spaced in the interval $[0.1,0.001]$ and choose $n$ using equation \eqref{n} and $\eta=1+\epsilon^2$. As the spheroids are axisymmetric, they only experience a torque about their $y$ axis, hence all of other angular momentum components are zero (to machine precision). Three examples of the rotational dynamics are shown in Figure \ref{ang_mom}. It is seen here that as $\epsilon$ becomes smaller, the dynamics more closely resemble the Jeffery model. \\

The relative difference between the angular momenta of the Jeffery and slender body solutions are calculated and averaged over the simulation. This average relative error is then plotted against the corresponding value of $\epsilon$ in Figure \ref{fig:angmomspheroid}. We see that the average relative error decreases with $\epsilon$. It is observed that in the region $0.01<\epsilon<0.1$ the error converges at a faster rate than in the region $0.001<\epsilon<0.01$. This could be partially explained by the fact that wider particles (larger $\epsilon$) experience a greater resistive force as seen by the regions where $m_y$ nearly reaches zero. This means that the particle spends more time in the shear plane where the fluid velocity is zero and hence the slender body model does not experience a large torque. However, the fluid gradient is non-zero in this orientation and therefore the Jeffery model, which depends only on the fluid gradient, still experiences a constant torque. This means that compared to the Jeffery model, thicker fibers will see a greater difference in the torque term when the fiber is aligned in the shear plane than thinner fibers.

\begin{figure}[!ht]
	\centering
	\begin{subfigure}{.3\textwidth}
		\includegraphics[width=\linewidth]{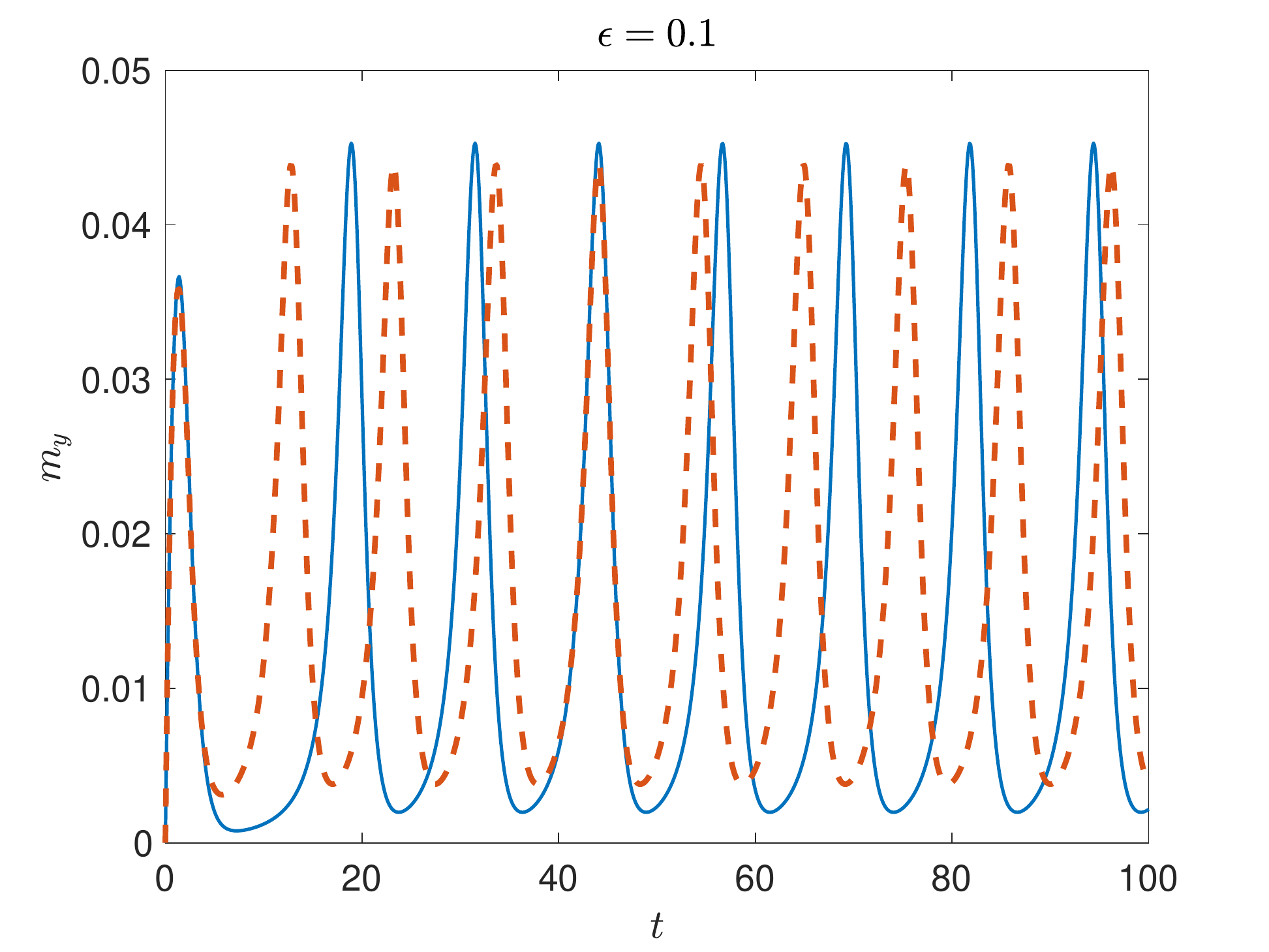}
		\caption{}
		\label{fig:1}
	\end{subfigure}
	\begin{subfigure}{.3\textwidth}
		\includegraphics[width=\linewidth]{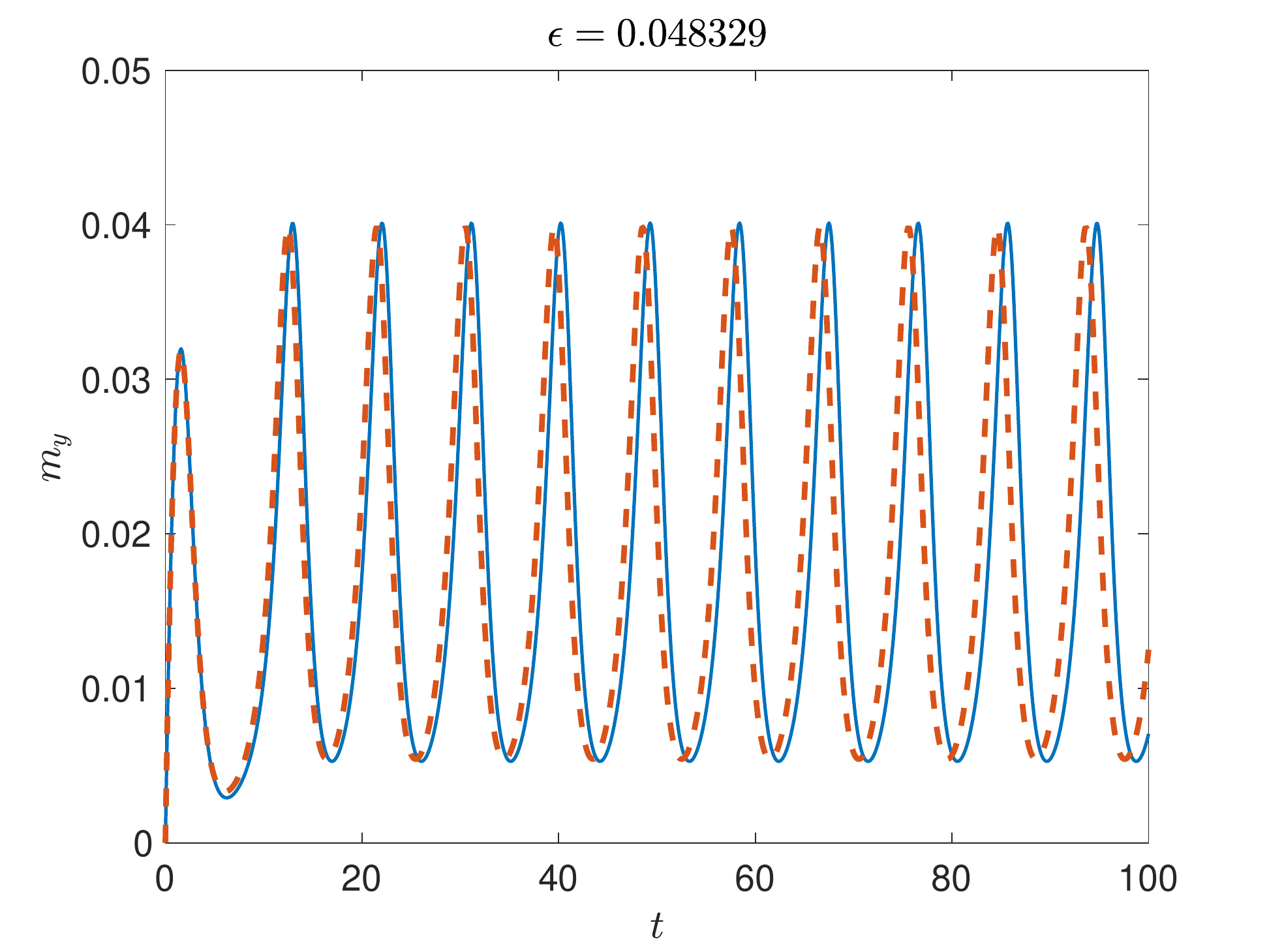}
		\caption{}
		\label{fig:5}
	\end{subfigure}
	\begin{subfigure}{.3\textwidth}
		\includegraphics[width=\linewidth]{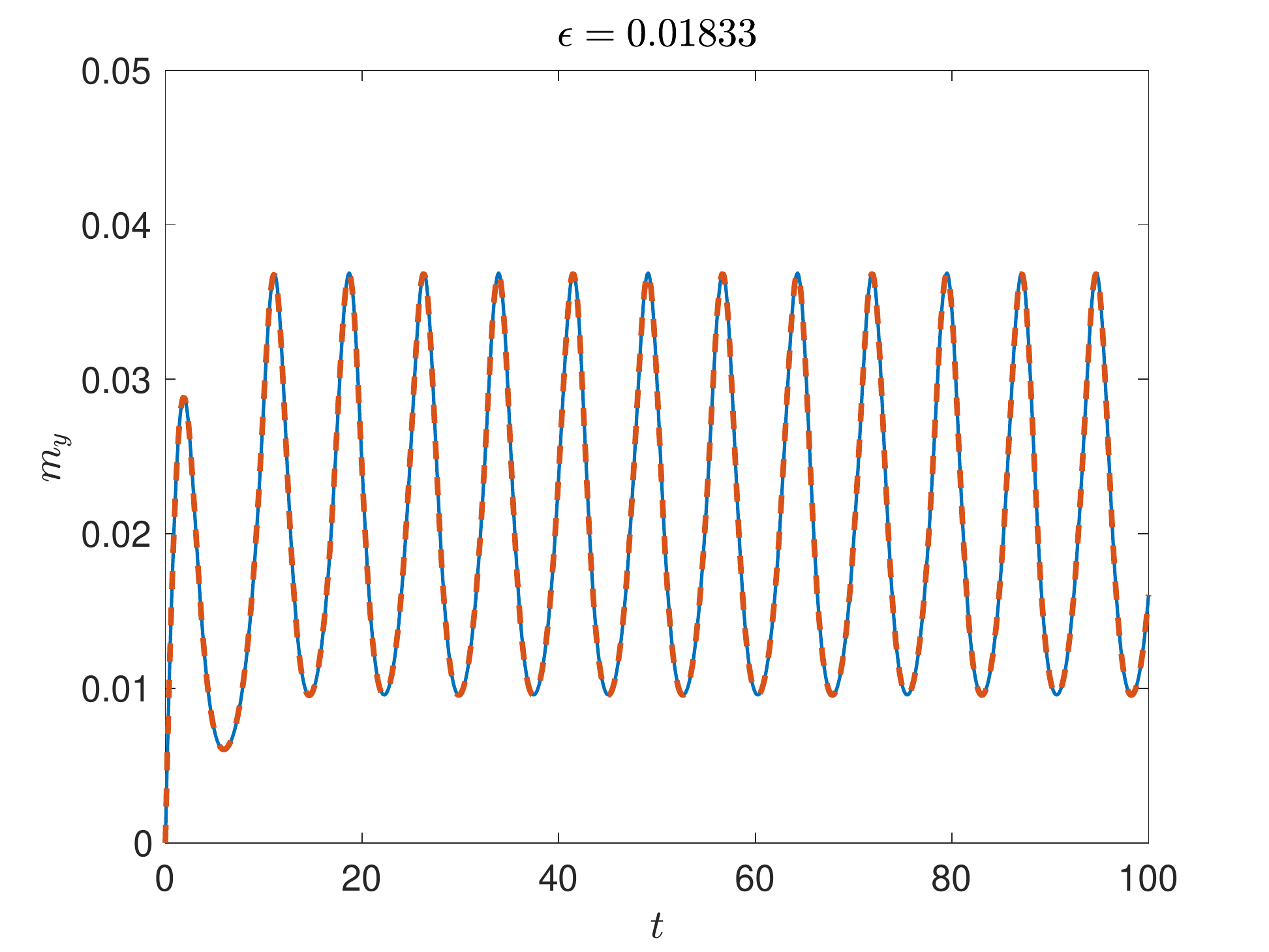}
		\caption{}
		\label{fig:10}
	\end{subfigure}
	\caption{The $y$ component of a spheroid rotating in shear flow for three different values of $\epsilon$. The solid line is the our slender body expression and the dashed line is due to Jeffery.}\label{ang_mom}
\end{figure}

\begin{figure}[!ht]
	\centering
	\includegraphics[width=0.45\linewidth]{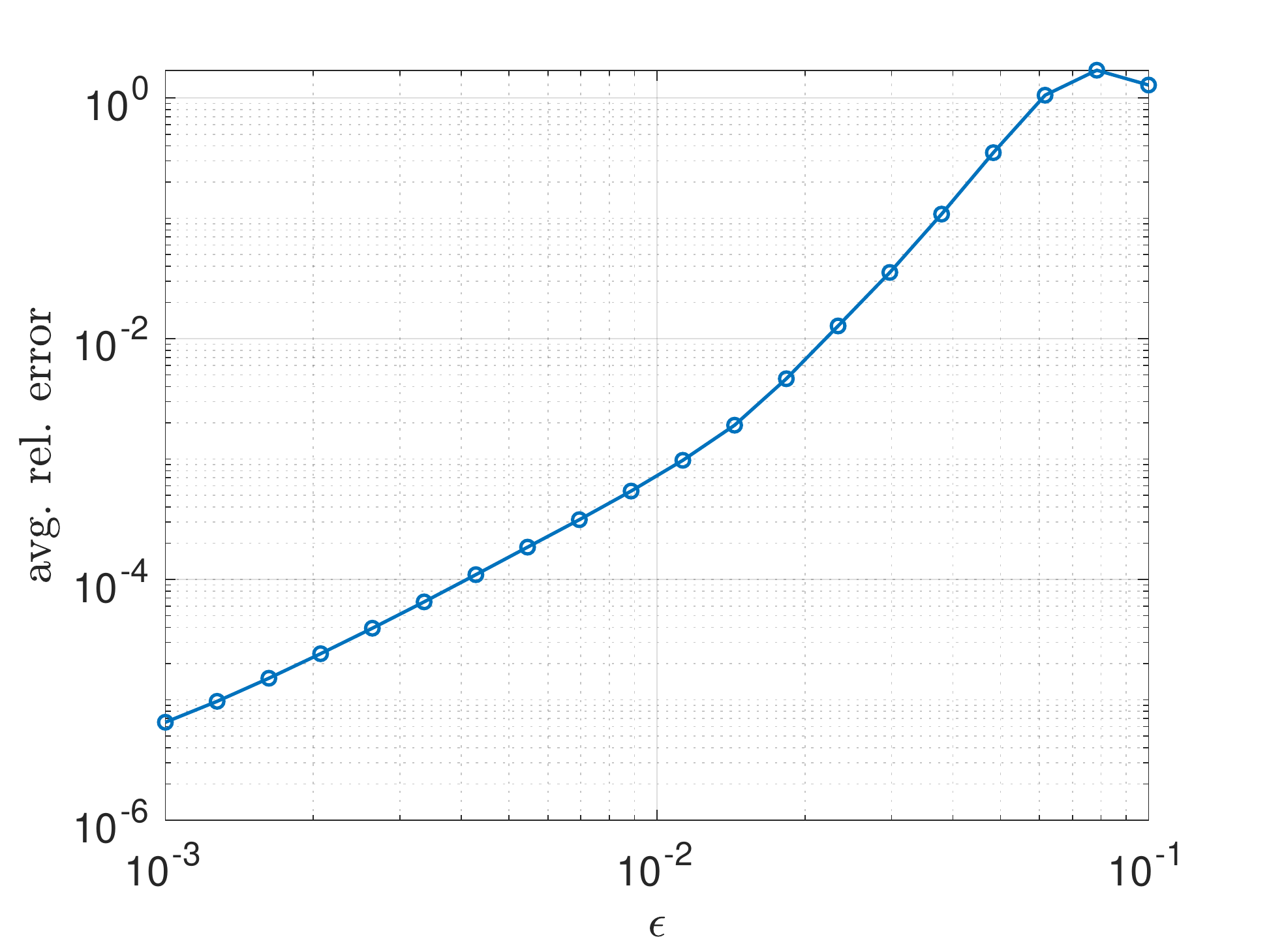}
	\caption{The relative difference in $m_y$ between the slender body and Jeffery solutions averaged over the interval $[0,100]$.}
	\label{fig:angmomspheroid}
\end{figure}

\subsection{Dynamics of randomly curvy fibers}
Here we simulate the dynamics of the randomly curvy fibers of Figure \ref{rand_fibres} as they rotate in shear flow. We show how the rotational variables deviate from a straight fiber as $\delta$ becomes larger. \\

We generate 100 different fiber shapes with $m=10$ using 10 different values of $\delta$ logarithmically spaced in the interval $[5\times 10^{-5},5\times 10^{-2}]$. The 100 fibers are placed in shear flow $\bu = (z,0,0)^T$ and their rotational dynamics are calculated on the interval $t\in[0,100]$. The moment of inertia tensor is approximated by placing point masses along the centerline and using the formula
\begin{equation}
	J_{i,i} = \sum_{j=1}^k m_j (X_{i}(s_j)-c_i)^2, \quad \text{for}\quad i = 1,...,3
\end{equation}
where $X_{i}(s_j)$ is the $i$th component of the centerline function at the point $s_j$ on the centerline and $c_i$ is the $i$th component of the fiber center of mass. We weight $m_j$ by the cross sectional radius and use a very large value for $k$, e.g., $k=10^4$. Here we take $\epsilon=0.1$ and use the spheroidal radius function \eqref{prolate} along with $\eta=1+\epsilon^2$. \\

 Figure \ref{fig:deltas} shows the angular momentum $\bm{m}$ of three fibers compared to the $\delta = 0$ case. As the $\delta=0$ fiber is perfectly straight, it does not exhibit spinning motion and its angular momentum is purely in the $m_y$ component. This is in contrast to the fibers with a non-zero value of $\delta$, in which case some of the momentum is transferred to $m_x$. We therefore compare the value $\sqrt{m_x^2+m_y^2}$ between the fibers to account for this. We see here that the $\delta = 0.017783$ solution is visually very similar to the $\delta = 0$ solution. We notice a significant difference between the other two solutions. Figure \ref{fig:thetas} shows the angle $\theta$ between the $z$-axis of the particle reference frame (that is, a frame that is rotating with the fiber) and the $x$-axis of a fixed inertial reference frame. As the $\delta\ne0$ fibers are not symmetric, they slowly rotate out of the $xz$-plane and therefore after a long time, we see much more significant discrepancies in $\theta$.\\

To quantify the effect that $\delta$ has on the angular momentum, we calculate the difference in the angular momentum $\Delta m$ by subtracting off the $\delta = 0$ solution and averaging over the time interval $t\in[92,100]$, which corresponds to roughly one period of rotation. This value is averaged over all the fibers with similar values of $\delta$ and is expressed as a percentage of the $\delta=0$ solution, which we denote by $\% \Delta m$. The results are plotted in Figure \ref{fig:deltaerr}. We notice that the $\% \Delta m$ is linearly proportional to $\delta$. We observe that at the end of the simulation the $\delta = 0.0003$ fibers correspond to roughly 1\% discrepancy in angular momentum and $\delta = 0.0015$ corresponds to roughly 7.5\% discrepancy. \\

The difference in $\theta$ after one rotation as a function of $\delta$ is displayed in Figure \ref{fig:thetaerr}.  The $\delta = 0.0003$ solution corresponds to about a $3^\circ$ difference in $\theta$ and the $\delta = 0.0015$ solution corresponds to about an $8^\circ$ difference.

\begin{figure}
	\centering
		\begin{subfigure}{.45\textwidth}
		\includegraphics[width=\linewidth]{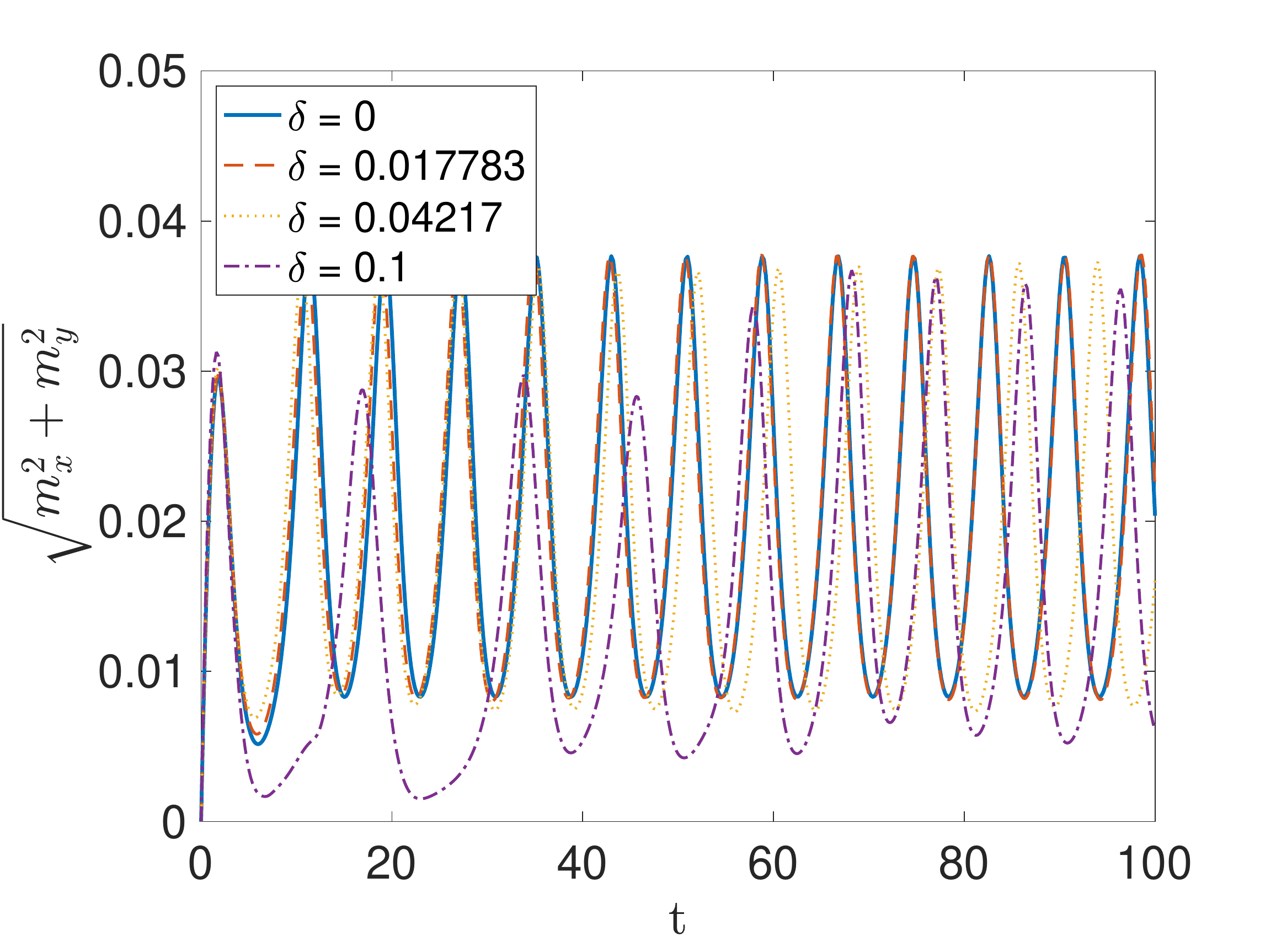}
		\caption{}
		\label{fig:deltas}
	\end{subfigure}
		\begin{subfigure}{.45\textwidth}
		\includegraphics[width=\linewidth]{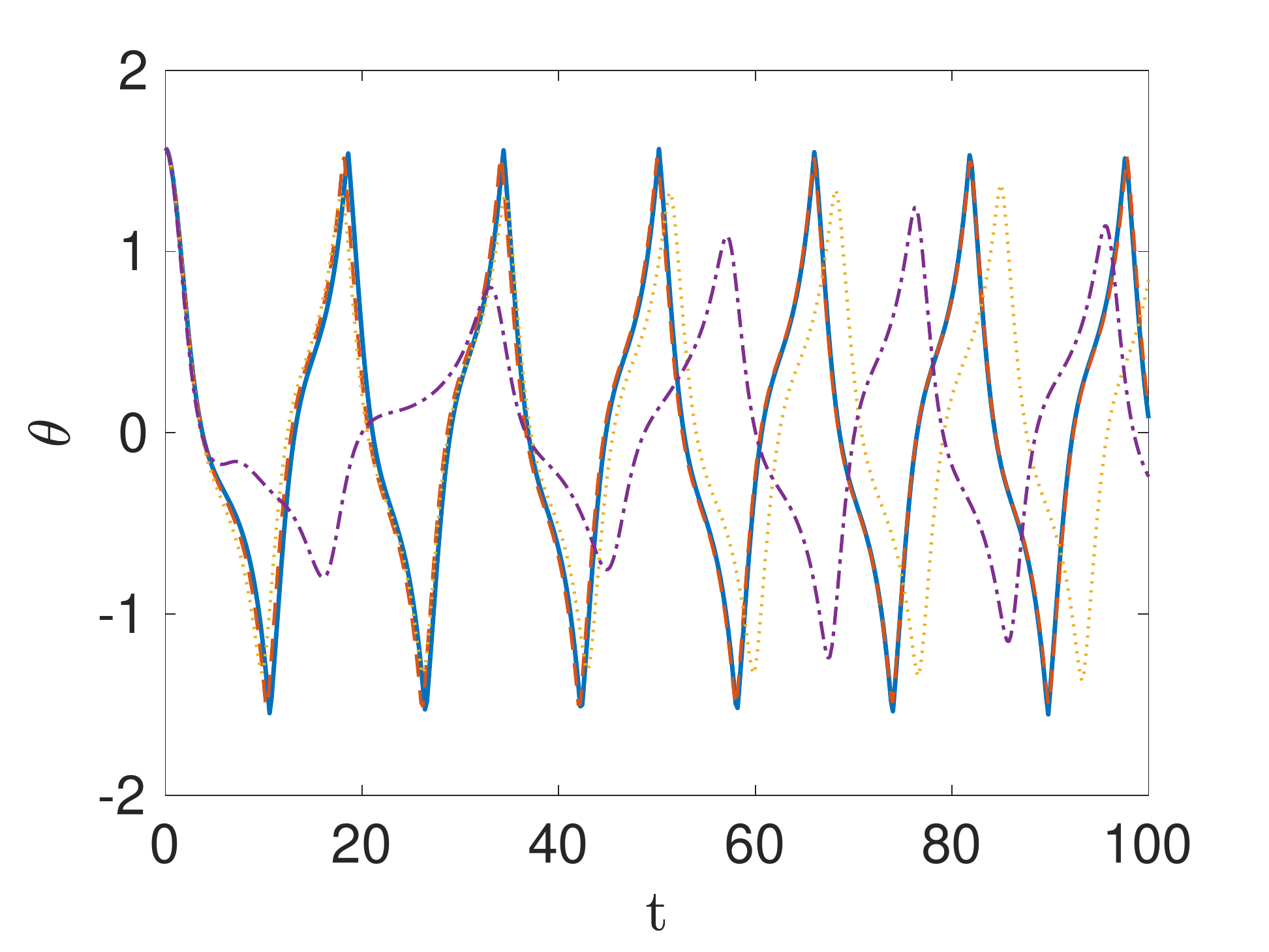}
		\caption{}
		\label{fig:thetas}
	\end{subfigure}
	\caption{The rotational variables of four fibers with different values of $\delta$. Figure (a) shows the angular momentum and Figure (b) is the angle between the fiber's long axis and the $x$-axis of the inertial frame.}
	\label{fig:angmomdeltas}
\end{figure}

\begin{figure}
	\centering
		\centering
	\begin{subfigure}{.45\textwidth}
		\includegraphics[width=\linewidth]{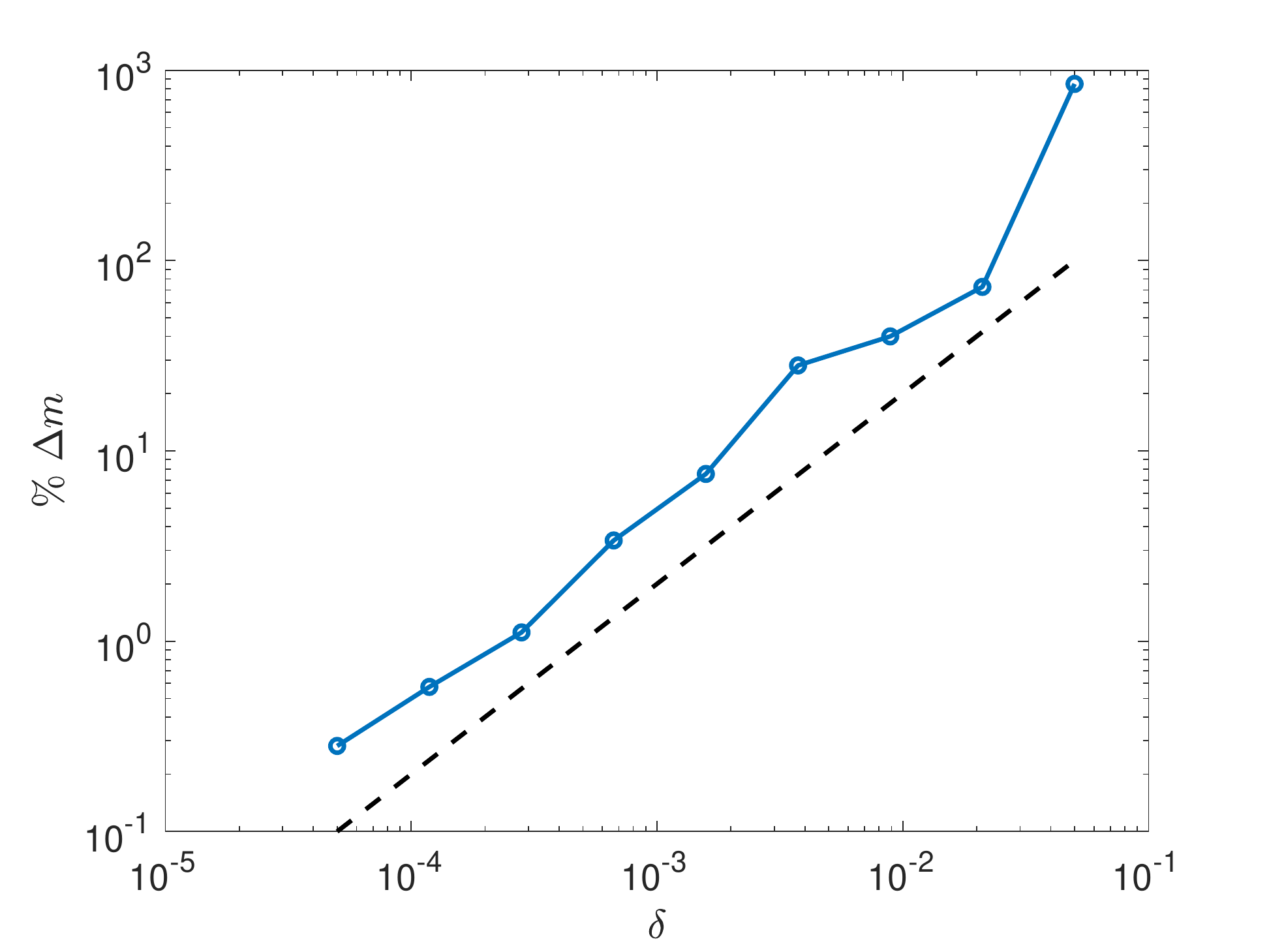}
		\caption{}
		\label{fig:deltaerr}
	\end{subfigure}
	\begin{subfigure}{.45\textwidth}
	\includegraphics[width=\linewidth]{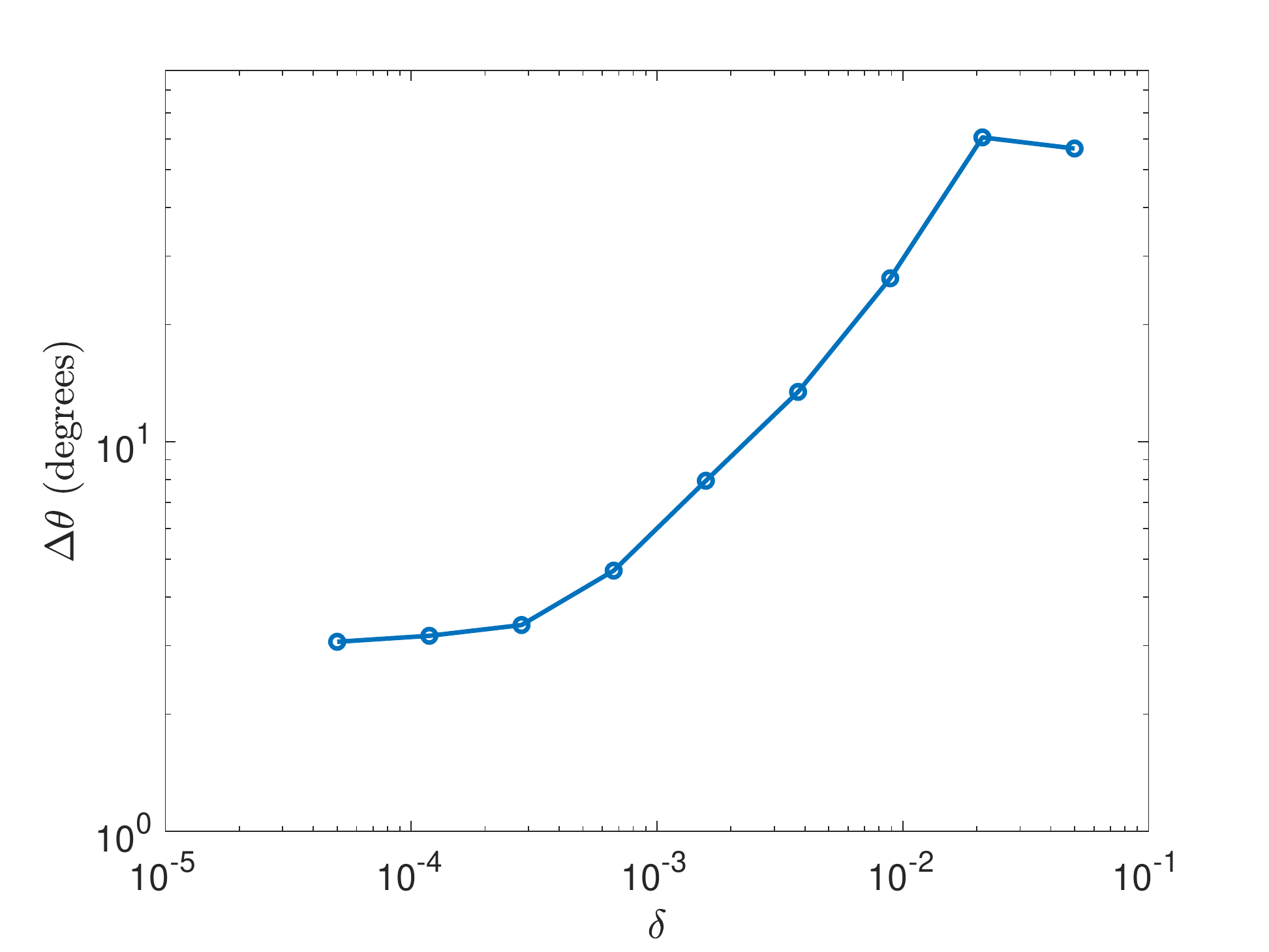}
	\caption{}
	\label{fig:thetaerr}
\end{subfigure}
	\caption{Figure (a) shows the difference in angular momentum $\Delta m$ between the curved fibers and the $\delta=0$ solution after 100 time units. The black dashed line is $O(\delta)$. Figure (b) shows the discrepancy $\Delta \theta$ in the angle between the centerline and the $x$-axis after roughly one rotation.}
	\label{fig:deltavserror}
\end{figure}


\section{Conclusions}\label{conclude}
We have developed an integral model for the motion of a thin filament in a viscous fluid based on nonlocal slender body theory. The model relies on standard singular Stokeslets and doublets but makes use of the fiber integrity condition -- the near-cancellation of angular-dependent terms along the fiber surface -- to derive an asymptotically accurate fiber velocity expression depending only on arclength. The kernel of the resulting integral operator is smooth and retains dependence on the (possibly varying) fiber radius in a natural way. We can show that this integral operator is negative definite in the simplified geometry of a straight-but-periodic filament, and we expect similar high wavenumber behavior for curved filaments with constant radius. It is less clear how a non-constant radius affects the spectrum; however, numerical tests indicate that the discretized integral operator is very close to negative definite. Nevertheless, to ensure invertibility, we develop an asymptotically consistent regularization to convert the first-kind Fredholm integral equation for the force density along the fiber into a second-kind equation and show that this second-kind regularization improves the stability and conditioning of the discretized equation. We develop a numerical method for solving the integral equation based on the the Nystr\"om method \cite{atkinson2005theoretical} and show how constraining the fiber motion to be rigid can be exploited for fast computation of fiber dynamics. We validate the method and model against the prolate spheroid model of Jeffery \cite{jeffery1922motion}, and apply the method to study the rotational deviation of randomly curved rigid fibers from straight fibers. \\

While the fibers considered here are rigid, the model can also be used to simulate the dynamics of semiflexible filaments. The invertibility properties of the integral equation make it particularly well suited for handling simulations involving inextensible fibers, where an additional line tension equation must be solved at each time step \cite{tornberg2004simulating,maxian2020integral}. We may also consider the effects of different choices of radius functions on the model properties, similar to what is done in \cite{walker2020regularised}, although we note the necessity of smooth decay in our radius function near the fiber endpoints. \\

To build on the dynamic simulations for rigid fibers, we aim to consider the effects of fiber shape on particle deposition and aggregation. We are especially interested in more complicated background flows, including suspensions of rigid fibers in turbulence. The novel modelling approach advocated herein will enable earlier explorations based on the point-particle approach \cite{challabotla2016fiber} to be extended to curved fibers particles.

\section{Acknowledgments}
This work has received funding from the European Unions Horizon 2020 research and innovation programme under the Marie Sklodowska-Curie grant agreement (No. 691070) as well as the SPIRIT project (No. 231632) under the Research Council of Norway FRIPRO funding scheme. E. Celledoni, B. Owren and B. K. Tapley would like to thank the Isaac Newton Institute for Mathematical Sciences, Cambridge, for support and hospitality during the programme {\it Geometry, compatibility and structure preservation in computational differential equations} (2019) where part of the work on this paper was undertaken. E. Celledoni and B. Owren also acknowledge funding from the European Union’s Horizon 2020 research and innovation programme under the Marie Skodowska-Curie grant agreement No 860124. L. Ohm was supported by a University of Minnesota Doctoral Dissertation Fellowship and NSF Postdoctoral Research Fellowship Grant No. 2001959. The authors are listed in alphabetical order. 

\bibliographystyle{abbrv}
\bibliography{NTNU_bib}

\begin{thebibliography}{10}

\bibitem{atkinson2005theoretical}
K.~Atkinson and W.~Han.
\newblock {\em Theoretical numerical analysis}, volume~39.
\newblock Springer, 2005.

\bibitem{atkinson1978introduction}
K.~E. Atkinson.
\newblock An introduction to numerical analysis.
\newblock 1978.

\bibitem{brenner1964stokes}
H.~Brenner.
\newblock The stokes resistance of an arbitrary particle—iv arbitrary fields
  of flow.
\newblock {\em Chemical Engineering Science}, 19(10):703--727, 1964.

\bibitem{challabotla2016fiber}
N.~R. Challabotla, L.~Zhao, and H.~I. Andersson.
\newblock On fiber behavior in turbulent vertical channel flow.
\newblock {\em Chemical Engineering Science}, 153:75--86, 2016.

\bibitem{chattopadhyay2009effect}
S.~Chattopadhyay and X.-L. Wu.
\newblock The effect of long-range hydrodynamic interaction on the swimming of
  a single bacterium.
\newblock {\em Biophys. J.}, 96(5):2023--2028, 2009.

\bibitem{chwang1975hydromechanics}
A.~T. Chwang and T.~Y.-T. Wu.
\newblock Hydromechanics of low-{R}eynolds-number flow. {P}art 2: Singularity
  method for {S}tokes flows.
\newblock {\em J. Fluid Mech.}, 67(4):787--815, 1975.

\bibitem{cortez2012slender}
R.~Cortez and M.~Nicholas.
\newblock Slender body theory for {S}tokes flows with regularized forces.
\newblock {\em Commun. Appl. Math. Comput. Sci.}, 7(1):33--62, 2012.

\bibitem{fan1998direct}
X.~Fan, N.~Phan-Thien, and R.~Zheng.
\newblock A direct simulation of fibre suspensions.
\newblock {\em J. Non-Newton. Fluid Mech.}, 74(1):113--135, 1998.

\bibitem{goldstein2002classical}
H.~Goldstein, C.~Poole, and J.~Safko.
\newblock Classical mechanics, 2002.

\bibitem{gotz2000interactions}
T.~G{\"o}tz.
\newblock {\em Interactions of fibers and flow: asymptotics, theory and
  numerics}.
\newblock Doctoral dissertation, University of Kaiserslautern, 2000.

\bibitem{gustavsson2009gravity}
K.~Gustavsson and A.-K. Tornberg.
\newblock Gravity induced sedimentation of slender fibers.
\newblock {\em Phys. Fluids}, 21(12):123301, 2009.

\bibitem{hamalainen2011papermaking}
J.~H{\"a}m{\"a}l{\"a}inen, S.~B. Lindstr{\"o}m, T.~H{\"a}m{\"a}l{\"a}inen, and
  H.~Niskanen.
\newblock Papermaking fibre-suspension flow simulations at multiple scales.
\newblock {\em J. Engrg. Math.}, 71(1):55--79, 2011.

\bibitem{hansen1992numerical}
P.~C. Hansen.
\newblock Numerical tools for analysis and solution of fredholm integral
  equations of the first kind.
\newblock {\em Inverse problems}, 8(6):849, 1992.

\bibitem{jeffery1922motion}
G.~B. Jeffery.
\newblock The motion of ellipsoidal particles immersed in a viscous fluid.
\newblock {\em Proc. R. Soc. Lond. A}, 102(715):161--179, 1922.

\bibitem{johnson1980improved}
R.~E. Johnson.
\newblock An improved slender-body theory for {S}tokes flow.
\newblock {\em J. Fluid Mech.}, 99(02):411--431, 1980.

\bibitem{keller1976slender}
J.~B. Keller and S.~I. Rubinow.
\newblock Slender-body theory for slow viscous flow.
\newblock {\em J. Fluid Mech.}, 75(4):705--714, 1976.

\bibitem{kress1989linear}
R.~Kress, V.~Maz'ya, and V.~Kozlov.
\newblock {\em Linear integral equations}, volume~82.
\newblock Springer, 1989.

\bibitem{lauga2009hydrodynamics}
E.~Lauga and T.~R. Powers.
\newblock The hydrodynamics of swimming microorganisms.
\newblock {\em Rep. Progr. Phys.}, 72(9):096601, 2009.

\bibitem{lighthill1976flagellar}
J.~Lighthill.
\newblock Flagellar hydrodynamics.
\newblock {\em SIAM review}, 18(2):161--230, 1976.

\bibitem{martin2017deposition}
J.~Martin, A.~Lusher, R.~C. Thompson, and A.~Morley.
\newblock The deposition and accumulation of microplastics in marine sediments
  and bottom water from the irish continental shelf.
\newblock {\em Sci. Rep}, 7(1):10772, 2017.

\bibitem{maxian2020integral}
O.~Maxian, A.~Mogilner, and A.~Donev.
\newblock An integral-based spectral method for inextensible slender fibers in
  stokes flow.
\newblock {\em arXiv preprint arXiv:2007.11728}, 2020.

\bibitem{spectral_calc}
Y.~Mori and L.~Ohm.
\newblock Accuracy of slender body theory in approximating force exerted by
  thin fiber on viscous fluid.
\newblock {\em arXiv preprint arXiv:2008.06829}, 2020.

\bibitem{rigid_SBT}
Y.~Mori and L.~Ohm.
\newblock An error bound for the slender body approximation of a thin, rigid
  fiber sedimenting in {S}tokes flow.
\newblock {\em Res. Math. Sci.}, 7(8), 2020.

\bibitem{closed_loop}
Y.~Mori, L.~Ohm, and D.~Spirn.
\newblock Theoretical justification and error analysis for slender body theory.
\newblock {\em Comm. Pure Appl. Math.}, 73(6):1245--1314, 2020.

\bibitem{free_ends}
Y.~Mori, L.~Ohm, and D.~Spirn.
\newblock Theoretical justification and error analysis for slender body theory
  with free ends.
\newblock {\em Arch. Ration. Mech. Anal.}, 235:1905--1978, 2020.

\bibitem{newsom1994dynamics}
R.~Newsom and C.~Bruce.
\newblock The dynamics of fibrous aerosols in a quiescent atmosphere.
\newblock {\em Physics of Fluids}, 6(2):521--530, 1994.

\bibitem{oberbeck1876uber}
A.~Oberbeck.
\newblock Uber stationare flussigkeitsbewegungen mit berucksichtigung der inner
  reibung.
\newblock {\em J. reine angew. Math.}, 81:62--80, 1876.

\bibitem{MEKiT}
L.~Ohm, B.~K. Tapley, H.~I. Andersson, E.~Celledoni, and B.~Owren.
\newblock A slender body model for thin rigid fibers: validation and
  comparisons.
\newblock {\em Proc. of {MEKiT}'19, 10th Nat. Conf. on Comp. Mech.}, 2019.

\bibitem{petrie1999rheology}
C.~J. Petrie.
\newblock The rheology of fibre suspensions.
\newblock {\em J. Non-Newton. Fluid Mech.}, 87(2):369--402, 1999.

\bibitem{rodenborn2013propulsion}
B.~Rodenborn, C.-H. Chen, H.~L. Swinney, B.~Liu, and H.~Zhang.
\newblock Propulsion of microorganisms by a helical flagellum.
\newblock {\em Proc. Natl. Acad. Sci.}, 110(5):E338--E347, 2013.

\bibitem{shelley2000stokesian}
M.~J. Shelley and T.~Ueda.
\newblock The {S}tokesian hydrodynamics of flexing, stretching filaments.
\newblock {\em Phys. D}, 146(1):221--245, 2000.

\bibitem{siewert2014orientation}
C.~Siewert, R.~Kunnen, M.~Meinke, and W.~Schr{\"o}der.
\newblock Orientation statistics and settling velocity of ellipsoids in
  decaying turbulence.
\newblock {\em Atmospheric research}, 142:45--56, 2014.

\bibitem{spagnolie2011comparative}
S.~E. Spagnolie and E.~Lauga.
\newblock Comparative hydrodynamics of bacterial polymorphism.
\newblock {\em Phys. Rev. Lett.}, 106(5):058103, 2011.

\bibitem{tapley2019novel}
B.~Tapley, E.~Celledoni, B.~Owren, and H.~I. Andersson.
\newblock A novel approach to rigid spheroid models in viscous flows using
  operator splitting methods.
\newblock {\em Numer. Algorithms}, pages 1--19, 2019.

\bibitem{tapley2019computational}
B.~K. Tapley, H.~I. Andersson, E.~Celledoni, and B.~Owren.
\newblock Computational methods for tracking inertial particles in discrete
  incompressible flows.
\newblock {\em arXiv preprint arXiv:1907.11936}, 2019.

\bibitem{tornberg2006numerical}
A.-K. Tornberg and K.~Gustavsson.
\newblock A numerical method for simulations of rigid fiber suspensions.
\newblock {\em J. Comput. Phys.}, 215(1):172--196, 2006.

\bibitem{tornberg2004simulating}
A.-K. Tornberg and M.~J. Shelley.
\newblock Simulating the dynamics and interactions of flexible fibers in
  {S}tokes flows.
\newblock {\em J. Comput. Phys.}, 196(1):8--40, 2004.

\bibitem{trefethen2014exponentially}
L.~N. Trefethen and J.~Weideman.
\newblock The exponentially convergent trapezoidal rule.
\newblock {\em siam REVIEW}, 56(3):385--458, 2014.

\bibitem{twomey1963numerical}
S.~Twomey.
\newblock On the numerical solution of fredholm integral equations of the first
  kind by the inversion of the linear system produced by quadrature.
\newblock {\em Journal of the ACM (JACM)}, 10(1):97--101, 1963.

\bibitem{walker2020regularised}
B.~J. Walker, M.~P. Curtis, K.~Ishimoto, and E.~A. Gaffney.
\newblock A regularised slender-body theory of non-uniform filaments.
\newblock {\em Journal of Fluid Mechanics}, 899:A3, 2020.

\end{thebibliography}


\appendix
\section{Dissipation matrix of a prolate spheroid}\label{app:dissmat}
The non-dimensionalized body frame resistance tensor $R_1$ for a spheroid with aspect ratio $\lambda$ was derived by Oberbeck \cite{oberbeck1876uber} and is given by 
\begin{equation}
R_1=16\pi\lambda~\mathrm{diag}\left( \frac{1}{\chi_0+\alpha_0},\frac{1}{\chi_0+\beta_0} ,\frac{1}{\chi_0+\lambda^2\gamma_0}\right).
\end{equation}
The constants $\chi_0$, $\alpha_0$, $\beta_0$ and $\gamma_0$ were calculated by Siewert \cite{siewert2014orientation} and are presented for a prolate ($\lambda>1$) spheroid 
\begin{gather}
	\chi_0 = \frac{-\kappa_0\lambda}{\sqrt{\lambda^2-1}},\\
	\alpha_0 = \beta_0 = \frac{\lambda^2}{\lambda^2-1}+\frac{\lambda\kappa_0}{2(\lambda^2-1)^{3/2}},\\
	\gamma_0 = \frac{-2}{\lambda^2-1}-\frac{\lambda\kappa_0}{(\lambda^2-1)^{3/2}},\\
	\kappa_0 = \ln\left(\frac{\lambda-\sqrt{\lambda^2-1}}{\lambda+\sqrt{\lambda^2-1}}\right).
\end{gather}
The torques $\mathbf{N}=(N_x,N_y,N_z)\trans$ that describe the rotational forces acting on an ellipsoid in creeping Stokes flow in the body frame were calculated by Jeffery \cite{jeffery1922motion} and are presented in their non-dimensional form with zero background flow
\begin{align}
N_x = & -\frac{16\pi \lambda}{3(\beta_0+\lambda^2\gamma_0)}\left[(1+\lambda^2)\omega_x\right],\label{eq:JT1} \\
N_y = & -\frac{16\pi \lambda}{3(\alpha_0+\lambda^2\gamma_0)}\left[(1+\lambda^2)\omega_y\right],\label{eq:JT2} \\
N_z = & -\frac{32\pi \lambda}{3(\alpha_0+\beta_0)}\omega_z.\label{eq:JT3}
\end{align}
Here $\bom = (\omega_x,\omega_y,\omega_z)\trans$ is the body frame angular velocity, which is related to body frame angular momentum by $\bL = J\bom$. Taking derivatives of $\mathbf{N}$ with respect to $\bL$ gives for the rotational dissipation matrix
\begin{equation}\label{eq:A_2}
R_2 = -\frac{16\lambda}{3}
\mathrm{diag}\left(\frac{(1+\lambda^2)}{(\beta_0+\lambda^2\gamma_0)},
\frac{(1+\lambda^2)}{(\alpha_0+\lambda^2\gamma_0)},
\frac{2}{(\alpha_0+\beta_0)}\right)J^{-1}.
\end{equation}
The full dissipation matrix used for the calculation in Figure \ref{fig:dissipationeigs} is given by
\begin{equation}
	A_{sph} = \left(\begin{array}{cc}
	R_1 & 0\\
	0 & R_2 \\
	\end{array}\right).
\end{equation}

\section{Regularizing effect of rigid body integration}\label{app:reg}
Here we consider the regularizing effect of a linear functional of the form \eqref{Fexpr} on numerical solutions to the first-kind equation \eqref{intop} when the fiber radius is constant. We consider specifically $M={\bf I}$. \\

First note that when the fiber radius is constant, the integral kernel of \eqref{intop} is symmetric. Therefore, if $\bff(s)$ is a solution to equation \eqref{intop} then so is $\bff(s)+\bff_o(s)$, where $\bff_o(s)=-\bff_o(-s)$ is any odd function. However, this nonuniqueness is not an issue if we are only interested in the integral of $\bff(s)+\bff_o(s)$ over the same interval as such an odd function would vanish.\\

This can further be illustrated using the singular value expansion. Let  $\bu_i$ and $\bv_i$ be the (orthogonal) left and right singular vectors and $\sigma_i$ be the singular values of $\mathbf{K}$. It is possible to express the integral operator as
\begin{equation}
	\mathbf{K}[f] = \sum_{i=1}^{\infty} \sigma_i(\bff,\bv_i)\bu_i,
\end{equation}
where $(\cdot,\cdot)$ is the inner product on $L^2([-L,L],\mathbb{R}^3)$. Similarly, we can expand the right hand side of the integral equation \eqref{intop} in terms of the basis $\bu_i$:
\begin{equation}
	\by(s) = \sum_{i=1}^{\infty} y_i \bu_i.
\end{equation}
From equation \eqref{intop}, we now obtain the relations
\begin{equation}
\sigma_i f_i = y_i, \quad \forall i.
\end{equation}
We see here that we can only solve for $f_i$ only if $\sigma_i\ne0$; otherwise the operator is singular. Our objective is to compute the integral of the solution $\bff(s)$, which is a linear functional applied to $\bff(s)$ that we denote by $\phi$. By the Reisz lemma this can be represented by a function $\bar{\phi}\in L^2([-L,L],\mathbb{R}^3)$ through ${\phi}[\bff] = (\bar{\phi},\bff)$. Expanding this function in the basis $\bv_i$ gives 
\begin{equation}
	\bar{\phi} = \sum_{i=1}^{\infty} \bar{\phi}_i \bv_i
\end{equation}
where $\bar{\phi}_i = (\bar{\phi},\bv_i)$. We then obtain 
\begin{equation}
	{\phi}(\bff) = \sum_{i=1}^{\infty} \frac{\bar{\phi}_iy_i}{\sigma_i}.
\end{equation}
Since $\sum_{i=1}^{\infty}y_i<\infty$ and $\sum_{i=1}^{\infty}\bar{\phi}_i<\infty$, these coefficients tend to zero and can therefore be approximated by some finite truncation as long as $\sigma_i$ doesn't decay too rapidly. Therefore the integration step has a regularizing effect on the ill-conditioned first-kind equation \eqref{intop}; however, we are still not guaranteed that a solution exists in the first place. \\

Note that we can perform a similar analysis for the second-kind equation \eqref{intop2} to obtain
\begin{equation}
{\phi}(\bff) = \sum_{i=1}^{\infty} \frac{\bar{\phi}_iy_i}{\alpha + \sigma_i}.
\end{equation}
In this case, even if $\sigma_i$ rapidly decays to zero, ${\phi}(\bff)$ exists and can be approximated by a finite sum. 

\end{document}